\documentclass[11pt]{article}%

\usepackage[top=1in, bottom=1in, left=1in, right=1in]{geometry}

\usepackage{caption}
\usepackage{subcaption}
\usepackage{tikz}
\usepackage{pgfplots}
\pgfplotsset{compat=1.17}
\usetikzlibrary{patterns,arrows.meta,backgrounds}

\usepackage{graphicx} 
\usepackage{siunitx}
\usepackage{url} 
\usepackage{amsmath,mathtools, amsthm}
\usepackage{cleveref}
\usepackage{tabularx}
\usepackage{thm-restate}

\usepackage[numbers]{natbib} 
\usepackage{verbatim}
\usepackage{changepage}
\usepackage{algorithm}
\usepackage{todonotes}

\usepackage[caption = false]{subfig}

\usepackage[english]{babel}
\usepackage{multirow}
\usepackage{longtable}

\usepackage{booktabs}

\usepackage{colortbl}
\usepackage{threeparttable}

\definecolor{winnerColor}{RGB}{189, 223, 167}%
\definecolor{nearWinColor}{RGB}{223, 240, 216}    %
\definecolor{contenderColor}{RGB}{253, 245, 206}  %
\definecolor{competitiveColor}{RGB}{253, 231, 208} %
\definecolor{distantColor}{RGB}{248, 218, 205}    %
\definecolor{farBehindColor}{RGB}{242, 201, 198}  %

\newcommand{\Comments}{1}
\newcommand{\mynote}[3]{\ifnum\Comments=1\textcolor{#1}{#2:#3}\fi}

\usepackage{longtable}
\usepackage{booktabs}

\theoremstyle{definition}
\newtheorem{theorem}{Theorem}[section]

\newtheorem{definition}[theorem]{Definition}

\theoremstyle{plain}

\theoremstyle{remark}

\usepackage{soul}
\usepackage{float}
 \usepackage
[english]{babel}

\usepackage{color}

\usepackage{algcompatible}
\usepackage{xcolor}

\newcommand{\RETURN}[1]{\State \textbf{return} #1}
\providecommand{\COMMENT}[1]{\Comment{#1}}

\makeatletter
\newcommand{\LineNumbersBlueVIItoXIII}{%
  \algrenewcommand\alglinenumber[1]{%
    \ifnum##1>6 \ifnum##1<14
      \textcolor{blue}{\arabic{ALG@line}}%
    \else
      \arabic{ALG@line}%
    \fi\else
      \arabic{ALG@line}%
    \fi
  }%
}
\newcommand{\LineNumbersDefault}{%
  \algrenewcommand\alglinenumber[1]{\arabic{ALG@line}}%
}
\makeatother

\LineNumbersDefault %

\makeatletter
\newcommand{\LineNumbersRobustBlue}{%
  \algrenewcommand\alglinenumber[1]{%
    \ifnum##1=5 \textcolor{blue}{\arabic{ALG@line}}%
    \else\ifnum##1=8  \textcolor{blue}{\arabic{ALG@line}}%
    \else\ifnum##1 = 2 \textcolor{blue}{\arabic{ALG@line}}%
    \else\ifnum##1=15 \textcolor{blue}{\arabic{ALG@line}}%
    \else\ifnum##1=17 \textcolor{blue}{\arabic{ALG@line}}%
    \else\ifnum##1=28 \textcolor{blue}{\arabic{ALG@line}}%
    \else                 \arabic{ALG@line}%
    \fi\fi\fi\fi\fi\fi
  }%
}
\makeatother

\title{Simpler Than You Think: The Practical Dynamics of Ranked Choice Voting}
\author{
Sanyukta Deshpande$^{1}$ \quad
Nikhil Garg$^{2}$ \quad
Sheldon Jacobson$^{1}$\\[0.5em]
$^{1}$University of Illinois at Urbana-Champaign\\
$^{2}$Cornell Tech
}

\date{}

\begin{document}
\maketitle

\begin{abstract}
Ranked Choice Voting (RCV) adoption is expanding across U.S. elections, but faces persistent criticism for complexity, strategic manipulation, and ballot exhaustion. We empirically test these concerns on real election data, across three diverse contexts: New York City's 2021 Democratic primaries (54 races), Alaska's 2024 primary-infused statewide elections (52 races), and Portland's 2024 multi-winner City Council elections (4 races). Our algorithmic approach circumvents computational complexity barriers by reducing election instance sizes (via candidate elimination).

Our findings reveal that despite its intricate multi-round process and theoretical vulnerabilities, RCV consistently exhibits \emph{simple} and \emph{transparent} dynamics in practice, closely mirroring the interpretability of plurality elections. 
Following RCV adoption, competitiveness increased substantially compared to prior plurality elections, with average margins of victory declining by 9.2 percentage points in NYC and 11.4 points in Alaska. Empirically, complex ballot-addition strategies are not more efficient than simple ones, and ballot exhaustion has minimal impact, altering outcomes in only 3 of 110 elections. These findings demonstrate that RCV delivers measurable democratic benefits while proving robust to ballot-addition manipulation, resilient to ballot exhaustion effects, and maintaining transparent competitive dynamics in practice. The computational framework offers election administrators and researchers tools for immediate election-night analysis and facilitating clearer discourse around election dynamics.

\end{abstract}

\section{Introduction}\label{sec:intro}

Ranked Choice Voting (RCV) is gaining widespread adoption across the United States, currently reaching approximately 14 million voters and 52 jurisdictions as of October 2025 \citep{FairVoteRCVUSE}. In 2024, Portland conducted its first multi-winner RCV elections, and Washington D.C. voters approved implementation beginning in 2026—joining a growing number of cities and states advancing RCV reform. Legislative proposals such as the Fair Representation Act \citep{fra} would establish multi-member districts with multi-winner RCV for U.S. congressional elections.

RCV's growing popularity stems from its potential to address several electoral challenges: mitigating vote splitting, reducing gerrymandering effects, enhancing campaign civility, and improving minority representation \citep{garg2022combatting, benade2021ranked, tomlinson2024moderating}. However, it is criticized as being overly complex, leading to a lack of transparency and a susceptibility to strategic manipulation; in turn voters may submit incomplete ballots, and such ballot exhaustion may influence outcomes \citep{RCVNightmare2024, clark2020rank, atkeson2024impact, alameda2025agenda, von2019ranked}. %

Central to these criticisms is RCV's intricacy compared to plurality voting's straightforward and accessible ``most votes wins" approach. RCV is most commonly implemented through the Single Transferable Vote (STV) system, and Instant Runoff Voting (IRV) is its popular single-winner version.  In elections with $n$ candidates competing for $k$ seats, RCV determines winners through iterative elimination rounds. Each round elects candidates meeting the Droop quota votes (equals $\lfloor \frac{\text{total votes}}{k+1} \rfloor + 1$, e.g., 50\% vote share for single-winner IRV), redistributing surplus votes proportionally.  When no candidate reaches the quota, the lowest-vote candidate is eliminated and their votes transfer to voters' next preferences. This iterative process continues until all $k$ winners are selected or only $k$ candidates remain. 

RCV thus requires extensive computations presented through multi-round tables, unlike plurality voting's single tally. Figure~\ref{fig:mayor_2017} illustrates New York City's 2017 Democratic mayoral primary official result under plurality voting  \citep{NYCBOE2017}.
This straightforward presentation contrasts sharply with the official 8-round tabulation from the 2021 RCV result in Figure~\ref{fig:mayor_2021} \citep{NYCBOE2021}. While Figure~\ref{fig:mayor_2017} conveys all information at a glance, Figure~\ref{fig:mayor_2021} obscures competitive dynamics, which must be explored via exponentially many elimination order possibilities.

\begin{figure}[H]
\centering
\includegraphics[width=0.5\linewidth]{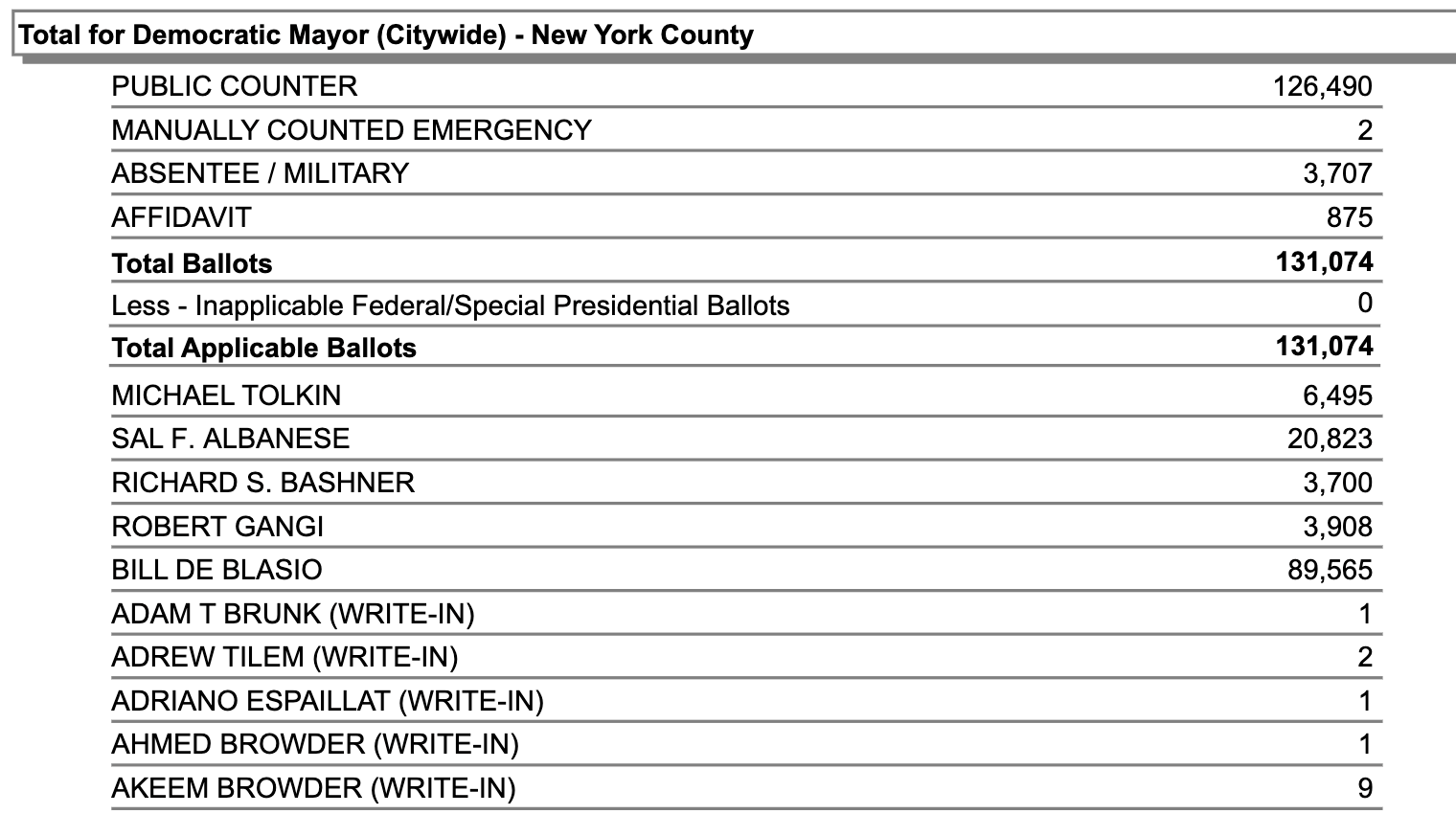}
\caption{The 2017 plurality voting results from New York City's Democratic Mayoral primary, from the official election website. (The extensive list of write-in votes is partly omitted here.) Bill de Blasio was elected in a single-round count, by securing more votes than any other candidate.}\label{fig:mayor_2017}
\end{figure}
\begin{figure}[h]
\centering
\includegraphics[width=0.8\linewidth]{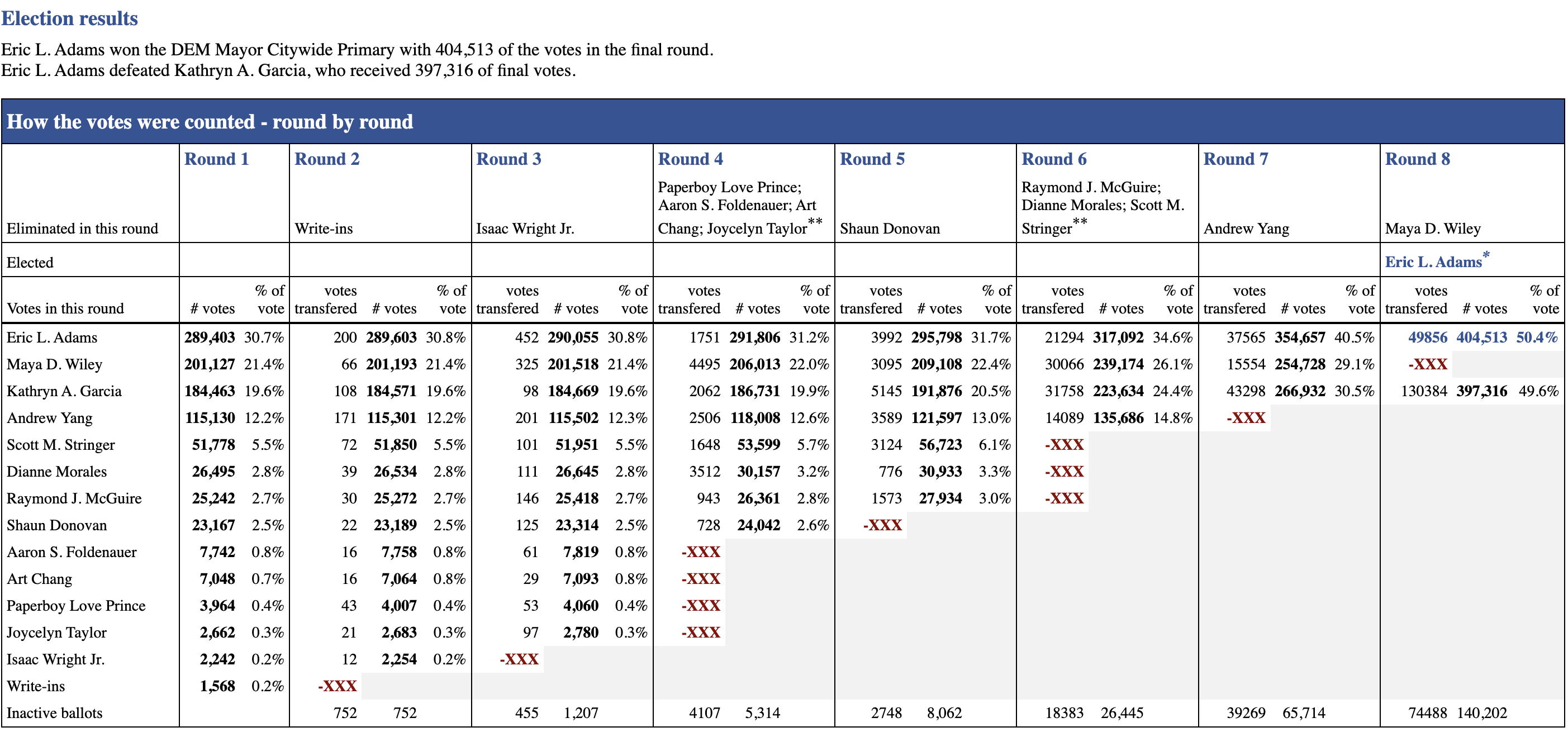}
\caption{The 2021 Ranked Choice Voting (RCV) results from New York City's Democratic Mayoral primary, as reported by the official election website. Eric L. Adams was elected in the 8th round, following the elimination of 10 candidates and the transfer of their votes.}\label{fig:mayor_2021}
\end{figure}

For instance, the official presentation in Figure~\ref{fig:mayor_2021} obscures essential competitive information: How narrowly did early-round second-place candidate Maya D. Wiley miss victory? What strategic actions could have boosted trailing candidates most effectively—building first-choice support, securing cross-endorsements, or supporting a spoiler against the key rival? Would completing the substantial `inactive' or exhausted ballots change the result? Was NYC's rank-5-candidates policy partly to blame, and was this election more competitive than the prior ones? Such questions—easily answered in plurality voting through simple vote tallies—illustrate RCV's transparency problem. 

RCV's complexity poses challenges even for expert stakeholders with access to detailed Cast Vote Records \citep{kean2023ranked}. With exponentially many elimination order possibilities and cascading vote transfers, comparing candidate standings is difficult \citep{xia2011determining}—candidates with broader later-round appeal may risk early elimination, while those with narrower early appeal may ultimately win. The election result in \Cref{fig:mayor_2021}, for instance, represents just one among $ 13!\approx 6.2$ billion possible elimination orders that could arise under different ballot profiles. This complexity obscures competitive dynamics, e.g., understanding the exact impact of voter actions, and contributes to voter confusion, strategic voting, and ballot exhaustion \citep{simmons2024sincere}. Theoretically, RCV remains susceptible to spoiler effects, where voters may achieve better outcomes by strategically ranking spoiler candidates higher than their true preference, and ballot exhaustion effects, where completing exhausted ballots can change outcomes. While such occurrences have been rare in prior analyzed elections \citep{graham2023examination}, these vulnerabilities resist systematic analysis due to RCV's inherent complexity and continue to fuel major criticisms of the system.

We tackle these transparency and complexity challenges through a computational framework that allows a large-scale validation of RCV's competitive dynamics in practice. Our approach covers both wide implementations of RCV -- single-member IRV and multi-member STV -- providing a large-scale analysis across over 100 diverse real-world RCV contests, involving up to 30 candidates and more than 900,000 voters.
For each election, our methods provide a straightforward presentation as available in plurality elections, offering answers to the above-mentioned questions. %
Specifically, our approach includes the following.

\paragraph{Computational Methods to dissect RCV's competitive dynamics:} We extend the algorithmic framework of \citet{deshpande2024optimal} to compute each candidate's optimal path to victory— traditionally an intractable problem due to RCV's cascading vote transfers. The resulting framework, termed  `Enhanced RCV Strategic Framework (ERSF)' is used to dissect elections of varying complexity: for smaller elections ($\leq$10 candidates), we provide comprehensive analysis of all candidates; for larger elections, we establish an algorithmic traceability threshold that identifies candidates who could realistically affect outcomes with additional votes up to that threshold (empirically validated at 8.5\% minimum across the dataset). This approach optimizes computational efficiency without sacrificing accuracy, focusing analysis on competitively relevant candidates.

We compute four key metrics: victory gaps (proximity to winning), ballot exhaustion impact (outcome sensitivity to incomplete ballots), strategic complexity (whether simple or coalition strategies are optimal), and preference order alignment (whether elimination order reflects true competitiveness). We apply the framework to compute these metrics in 110 real-world elections: 54 single-winner contests from New York City's 2021 local primaries, 52 single-winner contests from Alaska's 2024 statewide elections, and four multi-winner RCV contests from Portland's 2024 City Council elections, with bootstrap analysis providing robustness validation. These datasets were selected to provide: (i) sufficient scale for robust statistical analysis ($>$50 elections each for NYC and Alaska), (ii) broad coverage of RCV variants (single-winner IRV and multi-winner STV), (iii) significant electoral contexts—the largest municipal RCV rollout (NYC), the only U.S. state legislative RCV elections (Alaska), and recent multi-winner elections with high candidate density (Portland, 16–30 candidates), and (iv) diverse competition types—single-party primaries (NYC), multi-party general elections (Portland), and top-four primary system with multi-party generals (Alaska).

Our computational methods translate intricate RCV processes into clear, plurality-like metrics, enabling intuitive comparisons. As illustrations, we depict simplified, plurality-akin presentations of two RCV races---NYC's first council district with single-winner IRV and Portland's first 3-member district with multi-winner STV.\footnote{The dataset and the code are available at \url{https://github.com/sanyukta-D/Optimal_Strategies_in_RCV}.}

\paragraph{Empirical Validation of Simpler RCV Dynamics:}
Our large-scale analysis reveals that RCV exhibits surprisingly simple and transparent dynamics in practice, contrary to what theoretical concerns might imply. Elections are more competitive under RCV; ballot exhaustion has had a limited impact on electoral outcomes; and strategic behavior patterns are quantifiable and largely transparent. More precisely, we find:

(A) \textit{Positive Impact on Competitiveness:} We find substantial increases in competitiveness following RCV adoption. The average margin of victory dropped by 9.2 percentage points in NYC primaries (31\% relative reduction from 2017 to 2021) and by 11.4 points in Alaska statewide elections (42\% relative reduction from 2020 to 2024). We note that this analysis is not causal, given many other potential factors that could explain competitiveness changes over time.

 (B) \textit{Limited Impact of Ballot Exhaustion:} Although premature ballot exhaustion theoretically affects election outcomes, our analysis demonstrates low practical impact. Across 107 of the 110 elections analyzed, probabilistic completion of exhausted ballots—using several different models of voter preferences—indicated that completing exhausted ballots would be very unlikely to alter the outcomes,  indicating robustness to ballot incompleteness. Thus, RCV's vote transfer mechanism functioned effectively in practice, successfully aggregating voter preferences despite incomplete rankings in the vast majority of elections.

(C) \textit{Simple Strategic Behavior:} Despite vast theoretical possibilities, actual strategic behavior under RCV is surprisingly straightforward and predictable. Under traceability thresholds below 20\%, all 106 single-winner contests (NYC, Alaska) and 4 multi-winner contests (Portland) show exclusively self-support strategies. This shows empirical resistance to incentives for elaborate strategic voting in the ballot-addition framework. 

(D) \emph{Transparent Results}: RCV's elimination order typically mirrors underlying competitive dynamics rather than obscuring them. Under traceability thresholds below 20\%, 104 of 106 single-winner elections demonstrate perfect alignment between candidates' proximity to victory and their elimination sequence. Even at larger thresholds, 92 of 106 elections maintain full alignment, with partial alignment preserved for competitively relevant candidates. Portland's multi-winner elections exhibit similar transparency despite their added complexity.

~\\\noindent
Putting things together, these findings demonstrate that RCV's theoretical complexity rarely translates into practical complications. Our computational framework transforms RCV from an opaque multi-round process into an interpretable system, providing election administrators with tools for real-time analysis and giving policymakers empirical evidence to evaluate electoral reforms.

The remainder of the paper is structured as follows. Section~\ref{sec:lit_review} situates our work within the broader literature. Section~\ref{sec:theory} presents our computational approach and theoretical framework for assessing RCV's dynamics. We then turn to empirical analyses: Section~\ref{sec:case_studies_rcv} focuses on single-winner RCV elections in New York City (2021) and Alaska (2024), while Section~\ref{sec:casestudy_portland} examines multi-winner RCV elections in Portland’s 2024 City Council races. Apart from large-scale empirical results, these sections illustrate plurality-akin presentations of single and multi-winner RCV results, respectively. Section~\ref{sec:conclusions} concludes.

\subsection{Related Literature}\label{sec:lit_review}

Ranked Choice Voting (RCV) is an established alternative to plurality voting, implemented for over a century in Australia, Ireland, and Malta, and in U.S. cities like Cambridge, MA, since 1941. While research documents RCV's benefits in reducing vote splitting and improving representation, concerns persist regarding voter confusion, computational complexity, and strategic vulnerabilities. Despite extensive study across multiple disciplines, systematic empirical analysis remains limited by computational constraints that restrict most studies to small elections. This paper addresses these limitations through scalable computational methods and large-scale empirical investigation of real-world RCV elections.

\paragraph{Theoretical Foundations and Computational Complexity:}
RCV inherits fundamental strategic vulnerabilities from the Gibbard-Satterthwaite impossibility theorem, making it non-strategy-proof \citep{gibbard1973manipulation, arrow2012social}. More critically, single-transferable-vote manipulation is NP-hard \citep{bartholdi1991single}, with subsequent work extending hardness results to various manipulation types in both IRV and STV \citep{bartholdi1992hard, conitzer2007elections, faliszewski2009hard, ganzer2023model, xia2012computing, walsh2009manipulability}. This complexity manifests in several key questions that are NP-hard to answer: determining whether strategic voting can change outcomes, calculating exact margins of victory, predicting results with incomplete ballots, and analyzing the impact of candidate withdrawals \citep{blom2016efficient, han2023determining}. The root of this complexity lies in the exponential number of ways STV results can be realized, making it difficult to analyze alternate scenarios as the number of candidates grows \citep{brandt2016handbook}.

\paragraph{Algorithmic Approaches to RCV Analysis:}
Despite complexity challenges, several heuristics and algorithms have been devised to address crucial RCV questions  \citep{iceland2024sampling, sarwate2014efficient, ayadi2019single, ek2024improving, mccune2023ranked}.  For single-winner elections, \citet{magrino2011computing, blom2016efficient, cary2011estimating} develop algorithms to estimate margins of victory, while \citet{blom2020shifting} develop bounds on margins of victory for multi-winner STV elections. Given known complexity results, these approaches remain computationally feasible only for smaller contests and lack broader implementation to real-world data. 
Recently, \citet{deshpande2024optimal} developed algorithms for finding optimal strategic approaches, employing kernelization-like approaches by reducing the number of candidates considered, circumventing complexity barriers.  In their poll-based case study, backing rivals or spoilers was a dominant strategy. We extend their algorithms and apply them extensively to real-world data to broadly understand patterns in RCV, including the prevalence of such strategic concerns.

\paragraph{Voter Behavior and Ballot Exhaustion:}
The complexity of RCV systems creates significant challenges for voter behavior, with implications extending beyond individual voting decisions to broader electoral outcomes. \citet{simmons2024sincere} demonstrate that ``RCV appears to increase voter uncertainty around how to decide which candidates to support and leads to voters who appear to be neither sincere nor strategic." This aligns with \citet{atkeson2024impact}'s evidence that RCV's informational complexity reduces voter confidence and increases confusion, while \citet{wendland2023new} identify demographic and campaign factors that influence voters' likelihood to rank multiple candidates.

The phenomenon of ballot exhaustion—when voters submit incomplete rankings—critically threatens RCV's democratic promise. Extensive research has examined ballot length and exhaustion patterns \citep{dickerson2023empirical, baumeister2012campaigns, ayadi2019single, tomlinson2023ballot, blais2021smarter}, with disproportionate effects on minority communities \citep{mccarty2024minority} and other demographic groups \citep{neely2015overvoting, neely2008whose, santucci2017party, ClellandDuchinMcCune2025PortlandSTV}. When exhausted ballots exceed winners' margins, completion of ballots can alter outcomes \citep{mccune2023ranked, burnett2015ballot, kilgour2020prevalence}. Yet the complexity in measuring margins as well as the impact of completing exhausted ballots \citep{baumeister2012campaigns} has hampered systematic analysis of how non-uniform exhaustion undermines RCV's democratic representation—a question we address through comprehensive quantitative analysis.

Overall, while much of the previous theoretical work focuses on worst-case properties of RCV rather than typical electoral behavior, the empirical studies often lack systematic frameworks to dissect RCV or allow direct comparison with plurality systems. Recent systematic reviews (New America Foundation, 2024) synthesize the growing literature on RCV but note the continued need for large-scale empirical analysis, overcoming computational and methodological limitations.
\section{Theory}\label{sec:theory}
This section presents the computational framework for analyzing strategic dynamics in RCV and introduces key metrics that capture its characteristics. Building upon \citet{deshpande2024optimal}'s foundational algorithms, we develop the Enhanced RCV Strategic Framework (ERSF) applicable to real-world RCV elections.  The ERSF allows for efficient computation of election attributes as well as comparisons with plurality voting.

\subsection{Strategic Complexity and Foundational Approach}\label{sec:foundational_approach}
Plurality voting enables straightforward strategic analysis: the winner simply needs the most votes, making the margin of victory a direct measure of electoral competitiveness. In contrast, RCV introduces fundamental complexity through its elimination process, where vote transfers cascade across multiple rounds depending on the interaction of all ranked ballots. 
This structural difference makes strategic assessments computationally prohibitive in RCV systems. Questions trivial in plurality voting—such as determining the minimum ballot modifications needed to alter outcomes or quantifying electoral robustness—become challenging computational problems rather than simple arithmetic calculations \citep{bartholdi1992hard}.

Given the challenge, \citet{deshpande2024optimal} developed algorithms for strategic analysis in real-world scenarios. Their approach addresses finding the minimum number of ranked ballot additions needed to make particular election results viable, and is demonstrated on 2024 Republican primaries. We apply and extend their framework to analyze a large set of US elections.

\subsubsection{Preliminaries and Notation}\label{sec:preliminaries}

Let $\mathcal{C}=\{1,\dots,m\}$ be the set of candidates and $\mathcal{B}$ the multiset of (possibly partial) ranked ballots.  
We fix a positive integer $k\le m$—the \emph{seat count}—so the top $k$ positions of the final ranking are declared winners.  
Throughout the paper, we use the term \emph{Ranked Choice Voting (RCV)} to denote multi-winner \emph{Single Transferable Vote} (STV), which simplifies to single winner \emph{Instant-Runoff Voting} (IRV) when $k=1$. 

We adopt the widely used weighted-inclusive Gregory transfer \citep{PRF2024WIGM} with \emph{Droop quota} $Q$ defined as:
\[
Q = \lfloor | \mathcal{B} |/(k+1) \rfloor + 1
\]

\paragraph{RCV function.}
In each round, an active candidate whose vote total reaches $Q$ becomes a winner, exits the contest, and her \emph{surplus} votes (the excess over $Q$) are redistributed fractionally to later active choices.  
If no candidate reaches $Q$, the active candidate with the fewest votes is eliminated and all her votes transfer fully to the next active preference. Initialized empty, the social choice order places a winning candidate at the highest available spot, while a losing candidate is placed at the lowest.
A fixed tie-breaking hierarchy on $\mathcal{C}$ resolves any ties.  
The process terminates once every candidate is inactive; the first $k$ candidates in the resulting social choice order form the winning set. Candidates still in competition in a round are called \emph{active} in that round, and those already removed—either as winners or losers—are \emph{inactive}.  

Running the RCV function on $\{\mathcal{C},\mathcal{B}\}$ thus outputs a \emph{social choice order} $f\in \{ m! \text{ orders}\}$ (the final permutation of $m$ candidates) \emph{and} a binary \emph{round sequence} $s\in\{W,L\}^{\,m-1}$ recording whether each round ends with a win ($W$) or an elimination ($L$).  
The pair $(f,s)$ is called a \emph{structure}.  With the tie-breaking hierarchy fixed, every ballot profile maps to a unique structure (its election result), and the ballot simplex partitions into $m!2^{m-1}$ structures. This enables optimization algorithms to operate locally within each structure, with global optimality achieved by evaluating all possible structures.

\subsubsection{Algorithmic Approach}\label{sec:existing_approach}

A central aspect of analyzing RCV dynamics is quantifying the ``distance'' to alternative election outcomes (structures). We operationalize this distance as the \emph{minimum number of additional ballots} required to move from one structure to another. This additive metric is strategically meaningful as it mirrors real-world campaign tactics—such as Get Out The Vote (GOTV) initiatives—informing on voter mobilization\footnote{Our metric measures vote addition (`bribery' in the computational social choice literature) rather than altering existing preferences (`control'). While both manipulation strategies are computationally hard, the number of additions needed provides an upper bound for the number of alterations required—since each alteration is at least as meaningful as an addition. Thus, the additive metric offers a conservative estimate of manipulation difficulty.}.

The framework introduced in \citet{deshpande2024optimal} provides two algorithmic components to compute the metric:

\begin{enumerate}
    \item  \underline{Optimal Ballot Addition:} An algorithm to compute the minimum set of ballots (with explicit rankings) required to realize a desired structure. The algorithm processes each round sequentially in a greedy way, adding only the votes necessary to establish the required elimination or win at that stage  (see \Cref{app:prior_algo} for details). Iterating over all possible structures enables optimal strategy identification for each candidate. However, this process becomes computationally prohibitive for elections with more than $m=8$ candidates.

    \item \underline{Candidate Reduction:} An algorithm to identify candidates who cannot influence the election outcome, even with the addition of up to $B$ new ballots to $\mathcal{B}$. These candidates are called \emph{irrelevant} up to parameter $B$, the \emph{allowance}, that represents the maximum number of additional ballots considered when determining candidate irrelevance. Once irrelevant candidates are removed, the optimal ballot addition algorithm operates on the remaining  candidates to find optimal strategies within $B$. The \emph{algorithmic traceability threshold} is the maximum allowance that successfully reduces the candidate set.
\end{enumerate}

Both algorithms run in polynomial time with respect to $m$ (number of candidates) and $|\mathcal{B}|$ (number of ballots). For completeness, we reproduce the algorithms with brief descriptions in \Cref{app:prior_algo} as \Cref{algo: allocation} and \Cref{algo:candidate-removal}. Note that, since it is not guaranteed that the number of candidates can be reduced (especially with a large tolerance $B$), the worst case runtime remains intractable.

\subsection{Enhanced RCV Strategic Framework}\label{sec:computational_framework}

We address a few specific limitations of the foundational algorithms while maintaining polynomial-time complexity. We refer to our comprehensive solution as the ERSF. 

\begin{enumerate}
\item \textbf{Robustness and ballot exhaustion}: \Cref{algo: allocation} outputs optimal strategies as ballots with specific fixed rankings, but these strategies may fail if voters add preferences beyond those specified. We extend the algorithm to generate length-constrained strategies that remain optimal even when arbitrary additional preferences are appended (\Cref{sec:robust_allocation}), aligning outputs with real-world ballot constraints and unpredictable voter behavior. This robustness extension forms the basis of measuring the impact of ballot exhaustion.

\item \textbf{Improving algorithmic traceability thresholds}: The candidate removal conditions in \Cref{algo:candidate-removal} may often be conservative, yielding lower algorithmic traceability thresholds than achievable.  We strengthen removal conditions to handle broader election instances (\Cref{sec:improve_last_condition}), especially enabling analysis of larger elections with up to 120\% higher algorithmic traceability thresholds compared to earlier. 

\item \textbf{Containing early winner STV elections}: \Cref{algo:candidate-removal} can only remove consecutively eliminated candidates, thus fails to handle multi-winner elections with early winner(s), leaving such instances computationally prohibitive. We extend the algorithm for these scenarios (\Cref{sec:extension_multi-winner}). Our resulting approach can analyze, for example, Portland's elections with up to 30 initial candidates.

\end{enumerate}

Additionally, we implement parallelization and memoization techniques that extend practical computational limits from 8 to 10 candidates, after candidate reduction. (\Cref{app:proofs})

The ERSF powers our empirical analysis of large-scale RCV elections across diverse electoral contexts.  We use this analysis to measure RCV election attributes (\Cref{sec:rcv_attributes}). We next overview the computational enhancements, while relegating longer technical proofs to \Cref{app:proofs}. The formal specifications of the ERSF—\Cref{algo: robust_allocation} and \Cref{algo:candidate-removal-updated}—are detailed in \Cref{app:enhanced_algo}.

\subsubsection{Robust Ballot Allocation Strategies}\label{sec:robust_allocation}

Algorithm~\ref{algo: allocation} outputs optimal ballot addition strategies to achieve particular election outcomes. However, these strategies have two practical limitations: they specify ballots with exact fixed rankings (e.g., [A, B, C]), and their optimality can break if voters add preferences beyond those specified (e.g., submitting [A, B, C, D, E] instead of [A, B, C]). Understanding real-world voter mobilization requires strategies that remain optimal despite unpredictable voter behavior---specifically, strategies robust to voters adding arbitrary preferences beyond what's specified, and strategies that respect length constraints (such as ballot formats limiting rankings to 5 candidates). We thus extend \Cref{algo: allocation} to guarantee robustness to additional preferences and to output length-constrained strategies. The enhanced allocation algorithm is presented as \Cref{algo: robust_allocation} in \Cref{app:enhanced_algo}.

\begin{restatable}{proposition}{robustchoices} \label{thm: robust_subsequent_choice_part1}
Algorithm~\ref{algo: allocation} admits extensions that output:
\begin{enumerate}
\item Optimal strategies robust to suffixing arbitrary subsequent preferences to any ballot in the strategy.
\item Optimal strategies robust to prefixing ballots with irrelevant candidates, as determined by Algorithm~\ref{algo:candidate-removal}.
\item Optimal length-restricted strategies, including single-choice ballots.
\end{enumerate}
All extensions preserve polynomial-time complexity and optimality within their respective constraints.
\end{restatable}
\begin{proof}[Proof Sketch]
\Cref{algo: allocation} functions in a greedy way, adding votes to align the round result with desired round result, either via introducing new empty ballots or adding preferences to previously introduced ballots. 
For suffix-robust strategies, we track potential vote transfers from eliminated candidates and modify the vote allocation logic to account for all worst-case transfer scenarios that may challenge a round result. This ensures strategies remain effective regardless of subsequent preferences on added ballots. For prefix-robust strategies, we show that the irrelevant candidate-removal logic supports prefixing irrelevant candidates. For length-restricted strategies, we adapt the vote tracking mechanism to respect ballot length constraints, creating new ballots instead of reusing existing ones when necessary.  See \Cref{app:proofs} for complete proofs.
\end{proof}

Since strategies are additive, the robustness extends beyond practical implementation to modeling ballot exhaustion's electoral impact. A ballot exhausts when all its ranked candidates have been eliminated. When such exhausted ballots are ``completed'' by adding preferences for remaining candidates, this completion can be seen as equivalent to adding new strategic ballots from the exhaustion round onward. This is because the exhausted portion contains only inactive candidates already eliminated, and the ballots have `freed up' to accommodate strategic vote transfers.  

However, this equivalence would be meaningful if the optimal strategy (coming from \Cref{algo: robust_allocation}'s output) \emph{remains effective} when exhausted ballots serve as ``prefixes'' to optimal additions. That is, modified ballots of the form [original exhausted preferences,  \Cref{algo: robust_allocation} specified strategic preferences] are as effective as new additions [\Cref{algo: robust_allocation} specified strategic preferences].

The robustness guarantees in \Cref{thm: robust_subsequent_choice_part1} ensure additive strategies work despite irrelevant candidate prefixes and arbitrary suffix preferences. For exhausted ballot completion, we extend this concept: already-eliminated (not necessarily irrelevant) candidates in exhausted ballots function as prefixes to additive strategies without affecting their optimality. This
enables complete characterization of which candidates could win through ballot completion.

\begin{restatable}{proposition}{ballotexhaust} \label{prop: ballot_exhaust} 
\Cref{algo: robust_allocation} determines the set of candidates that can achieve victory through completion of exhausted ballots.
\end{restatable}

\begin{proof}[Proof Sketch]
A candidate can win via ballot completion if their required strategic votes are fewer than the available exhausted ballots by the time their strategy would activate. In single-winner IRV, we show that this condition is sufficient. In multi-winner STV, adding ballots increases the Droop quota, which can strategically alter the round-result sequence (changing `win' rounds to `elimination' rounds or vice versa). We show that the condition remains sufficient for strategies that preserve the sequence structure---such strategies are identified through \Cref{algo: robust_allocation}'s output constrained by the available exhausted ballot pool. (See \Cref{app:proofs} for complete details.)
\end{proof}

This characterization allows assessment of whether ballot exhaustion, i.e., incomplete presentation of voter preferences, impacted the electoral outcome, quantifying how completely the result captures the underlying preferences. 

\subsubsection{Strengthened Candidate Removal Conditions}\label{sec:improve_last_condition}

The original removal condition eliminates candidate groups when \emph{all} candidates in the lower group have fewer votes than any upper group candidate, even when up to $B$ additional ballots may be added. We strengthen this by recognizing that candidates may be safely removed even if they temporarily survive an elimination round.

\begin{restatable}{theorem}{augmentremoval} \label{thm: augment_removal}
Algorithm~\ref{algo:candidate-removal} can be strengthened to remove candidates who survive one elimination round but are guaranteed elimination in the subsequent round, preserving optimality while enabling more aggressive pruning.
\end{restatable}

\begin{proof}[Proof Sketch]
When a candidate $C$ in the lower group could potentially outlast the lowest-vote candidate in the upper group, we verify: (1) whether the budget $B$ is sufficient to save both candidates simultaneously, and (2) whether $C$ would still be eliminated in the next round after the lowest-vote upper group candidate is eliminated. If both conditions confirm inevitable elimination, we can safely remove $C$ along with the lower group. See \Cref{app:proofs} for the complete mathematical proof and the corresponding algorithm.
\end{proof}

The extension preserves the original algorithm's complexity of $O(nm^4)$ while enabling more effective reduction of the problem size in practice. This improves algorithm search capability, 
enabling 55\% medium to big-sized ($\geq 6$ candidates) elections to achieve higher algorithmic traceability thresholds, with an average 39\% gain for larger elections ($\geq 10$ candidates). See \Cref{app:strong_removal_impact} for details.

\subsubsection{Containing Multi-Winner Instances} \label{sec:extension_multi-winner}
We next extend the applicability of \Cref{algo:candidate-removal} to effectively contain \emph{multi-winner} election instances with an \emph{early} round-winner(s). In such cases, \Cref{algo:candidate-removal} may result in a small set of irrelevant candidates (i.e., limiting to few eliminations before any win), leaving a computationally prohibitive relevant search space. For such instances, we write conditions to contain winner(s) in the set of irrelevant candidates, while still ensuring that the procedure doesn't affect the optimality. These crucially enable the multi-winner Portland case study in \Cref{sec:casestudy_portland}. 

\begin{restatable}{theorem}{candidateremovalextension} \label{thm: irrelevant_extension}
    Let an application of \Cref{algo:candidate-removal} with $Q_{new} = (k+1) Q$ result in a set of irrelevant candidates $L$, and relevant candidates $\mathcal{C}\setminus L$. After eliminating candidates in $L$ (and transferring subsequent votes to $\mathcal{C}\setminus L$), if candidate $C_w \in \mathcal{C}\setminus L$ secures more than $Q$ votes, 
    then satisfiability of polynomial-time verifiable conditions validates the optimality of \Cref{algo:candidate-removal} in retaining $\mathcal{C}\setminus 
 L$.
\end{restatable}

\begin{proof}[Proof Sketch]
By setting \(Q_{new} = (k+1)Q\), the Droop quota is big enough to ensure that the candidates in \(L\) are chosen to be strictly eliminated in order. Since after eliminating \(L\) and transferring their votes, candidate \(C_w \in \mathcal{C} \setminus L\) receives more than \(Q\) votes, \(C_w\) must win during the elimination process of $L$ using original quota $Q$. The precise position at which \(C_w\) wins may vary if new votes from \(B\) are added. 
The key argument is that—even if the new votes from \(B\) may alter \(C_w\)'s position within \(L\)'s elimination sequence—\(C_w\)'s influence may be bounded in terms of STV transfers when \(C_w\) goes out of contest as well as direct surplus transfers from $C_w$'s win. This ensures that the candidates in \(L\) remain irrelevant. Specifically, we verify whether the removal condition holds for each \(C_i \in L\) and \(C_j \in \mathcal{C} \setminus L\), if:  both $C_i\in L$ and $C_j \in \mathcal{C} \setminus L$ have updates to their active votes due to $C_w$'s win during $L$'s elimination sequence. The formal condition and proof details, along with algorithmic implementation are provided in \Cref{app:proofs}.
\end{proof}

\Cref{algo:candidate-removal-updated} in \Cref{app:enhanced_algo} extends \Cref{algo:candidate-removal} by incorporating both the strengthened removal conditions from \Cref{thm: augment_removal} and the multi-winner containment procedures from \Cref{thm: irrelevant_extension}. 

Finally, observe that after \Cref{algo:candidate-removal-updated} identifies the relevant candidate set of size $m'$, the ERSF evaluates all possible structures—specifically, the $m'!$ social choice orders combined with $2^{\,m'-1}$ round-result sequences—yielding a total of $(m'!)\,2^{\,m'-1}$ structures that \Cref{algo: robust_allocation} searches for optimal strategies.  We utilize the structural independence of these searches to implement pruning, memoization, and parallelization techniques to improve scalability while preserving polynomial complexity. This increases the computational traceability from elections with 8 candidates up to those with 10 candidates. See \Cref{thm: computational_enhancements} for details.

\subsection{RCV Election Attributes} \label{sec:rcv_attributes}

The ERSF in \Cref{sec:computational_framework} enables characterization of RCV's competitive dynamics by computing optimal ballot additions for any feasible outcome.  We now describe metrics to understand RCV competitiveness, the effects of ballot exhaustion, the level of strategic complexity required to achieve victory, and whether the elimination sequence reflects candidate competitiveness. 

\begin{enumerate}
\item \textbf{Victory Gap and Margin of Victory (Competitiveness):} We define the \emph{Victory Gap} for any candidate as the minimum percentage of additional ballots needed for that candidate to win the election. The \emph{Margin of Victory} (or \emph{Competitiveness}) is quantified as the minimum percentage of strategic ballot additions required to change the election winner (i.e., the smallest victory gap among all non-winning candidates). Lower margins indicate highly competitive elections where minimal interventions could alter outcomes.

\item \textbf{Ballot Exhaustion:} 
An election exhibits outcome sensitivity to ballot exhaustion when a candidate's victory gap is smaller than the cumulative exhausted ballots at their elimination round. In such cases, completing the exhausted ballots could potentially alter the election outcome. \Cref{prop: ballot_exhaust} characterizes all candidates whose victories could be affected by ballot completion. We estimate the likelihood of alternate outcomes using several probability models (bootstrap and beta distributions) that simulate various ballot completion scenarios under different assumptions about voter preferences (see \Cref{app:ballot_exhaustion} for details). %

ERSF quantifies this outcome sensitivity through completion model probabilities. Low probabilities indicate robust elections where outcomes are insensitive to ballot incompleteness, while high probabilities reveal elections vulnerable to changes upon ballot completion.

\item \textbf{Strategic Complexity:} We define strategic complexity based on the composition of each candidate's \emph{victory gap}, i.e., the minimum ballot addition to win. A candidate exhibits a \emph{Selfish Strategy} if their minimum ballot addition consists entirely of ballots ranking them first, as bullet votes. A candidate exhibits a \emph{Non-Selfish Strategy} if every minimum ballot addition for that candidate must include at least one ballot that ranks a different candidate first or at any position in the ballot, indicating that optimal manipulation requires supporting rivals to engineer favorable elimination sequences. An election is \emph{strategically complex up to allowance $B$} if any candidate has non-selfish optimal strategies within $B$. Elections where all candidates exhibit only selfish strategies indicate strategic environments equivalent to plurality's vote-maximization paradigm. 

\item \textbf{Preference Order Alignment:} We examine the Social Choice Order (the ranking of candidates based on their elimination sequence and final outcome in the actual RCV election) against the \emph{Victory Gap Order} (the ranking based on victory gaps). An election shows preference order \emph{Match} when these orders align, while \emph{No-Match} cases reveal elections where formal results may obscure true competitive dynamics. An election has a \emph{match up to allowance $B$} if both orders are aligned when restricted to candidates whose computed victory gaps are at most the stated allowance.
\end{enumerate}

\section{Single-Winner STV: NYC'21 and Alaska'24} \label{sec:case_studies_rcv}
We first demonstrate the ERSF on single-winner STV, or Instant Runoff Voting (IRV).
We analyze two sets of elections: 54 New York City Council Democratic primary elections (2021) and 52 Alaska statewide elections with top-four nonpartisan primaries (2024). These cases represent contrasting electoral environments—urban municipal primaries featuring numerous same-party candidates versus statewide general elections with fewer multi-party contenders. We begin by detailing the elections and methodology, then elucidate our approach on NYC's first district, and finally describe results on election attributes across all contests.

\subsection{Electoral Context and Data}

The 2021 New York City Council Democratic primary elections marked the first widespread use of RCV in the city, following its approval via a 2019 ballot measure. RCV allowed voters to rank up to five candidates in order of preference. The primary was held on June 22, 2021, alongside the mayoral and other citywide races. This system coincided with an increase in voter turnout and candidate diversity, with significant gains in representation for women and minority groups in the City Council \citep{citizensunion2021rcv, fairvote2023rcvnyc}. 

Alaska's 2024 statewide elections featured a top-four primary-infused RCV system, implemented first in 2022 \citep{alaskaelections}. In the top-four primary, held on August 20, 2024, all candidates for an office competed in a single nonpartisan primary, with the top four vote-getters advancing to the general election, regardless of party affiliation. The general election on November 5 then used RCV, allowing voters to rank candidates by preference to ensure winners had majority support \citep{ballotpedia2024alaska}.

The cast vote records (CVR) for both elections were obtained from official sources: NYC data from \cite{NYCBOE2021} and Alaska data from \cite{alaskaelections}. The NYC elections involved 2 to 15 candidates requiring up to 15 elimination rounds, while Alaska elections had 2 to 8 candidates. This variation in scale allowed us to test our algorithmic framework across different levels of complexity. For clarity, we represent candidates using letters `A', `B', `C', etc., ordered by their final social choice ranking: candidate A is the winner, B is the runner-up, and so on.

\subsection{Methodology}

We applied the ERSF from \Cref{sec:theory} to both case studies. For NYC's 54 elections, we used \Cref{algo:candidate-removal-updated} to first identify relevant candidates. Elections with fewer than 11 candidates were analyzed with up to 40\% vote addition allowance, while the larger elections had allowances varying from 8.5\% to 35\% based on algorithmic traceability. For Alaska's 52 elections, we characterized up to 100\% allowance as the typical election sizes are smaller.

The ERSF output is used to compute (1) victory gaps and margins of victory, (2) ballot exhaustion effects, (3) strategic complexity, and (4) preference order alignment between social choice order and strategic vulnerability order. For the ballot exhaustion analysis, as illustrated in \Cref{app:ballot_exhaustion}, we estimated preferences using two approaches: based on all non-exhausted ballots and based on non-exhausted ballots with identical first preferences\footnote{Prior research suggests that ballot exhaustion disproportionately affects minority communities' representation and influence as their ballots have higher exhaustion rates \citep{mccarty2024minority}.}. We applied both Beta and bootstrap models under each conditioning approach.

A complete breakdown of all NYC elections is provided in \Cref{tab:nyc_elections}, and that for Alaska in \Cref{tab:alaska_elections}, both in \Cref{app:single_winner_details}. Here we present a detailed illustration of the District 1 election and summarize insights from the full set of analyzed elections.

\subsection{Illustrative Example: NYC Council District 1} \label{subsec:dis1_nyc}

The NYC Democratic Council District 1 primary involved 9 candidates, with Christopher Marte winning after multiple elimination rounds and vote transfers. \Cref{tab:victory-gap} demonstrates how ERSF transforms RCV dynamics into a plurality-style presentation for clearer strategic interpretation. The official results of this election are reproduced in \Cref{app:dis1_nyc_full}.

\begin{table}[h]
\centering
\small
\caption{NYC RCV Council District 1 Primary (2021), akin to Plurality Voting}
\begin{threeparttable}
\begin{tabular}{llrrrc}
\toprule
\textbf{Candidate} & \textbf{ID} & \multicolumn{2}{l}{\textbf{Victory Gap (\%)}} & \textbf{Required Strategy} & \textbf{Exhaustion} \\
\midrule
\rowcolor{winnerColor} Christopher Marte & A & 0.00 & \textbf{Winner} & Actual Winner & 39.93\% \\
\rowcolor{contenderColor} Jenny L. Low & B & 17.07 & \textbf{Contender} & Low self-support (17.07 \%) & 18.23\% \\
\rowcolor{competitiveColor} Gigi Li & C & 20.13 & \textbf{Competitive} & Li self-support (20.13 \%) & 10.50\% \\
\rowcolor{competitiveColor} Susan Lee & E & 28.95 & \textbf{Competitive} & Lee self-support (28.95 \%) & 3.04\% \\
\rowcolor{distantColor} Maud Maron & D & 37.49 & \textbf{Distant} & Maron self-support (37.49 \%) & 5.32\% \\
\rowcolor{farBehindColor} T. Johnson-Winbush & G & 47.37 & \textbf{Far Behind} & Winbush self-support (47.37  \%) & 0.60\% \\
\rowcolor{farBehindColor} Sean C. Hayes & F & 47.45 & \textbf{Far Behind} & Hayes self-support (47.45 \%) & 1.53\% \\
\rowcolor{farBehindColor} Susan Damplo & H & 50.15 & \textbf{Far Behind} & Damplo self-support (50.15 \%) & 0.26\% \\
\rowcolor{farBehindColor} Denny R. Salas & I & 53.47 & \textbf{Far Behind} & Salas self-support (53.47 \%) & 0.00\% \\
\bottomrule
\end{tabular}

\begin{tablenotes}[para,flushleft]
\small
\item Colors indicate competitive categories, highlighting the clustering of candidates within the field. 

\vspace{0.3cm}
\item \textbf{Competitive Categories:}\\
\item \colorbox{winnerColor}{\textbf{Winner}} (0\%): Actual election winner\\
\item \colorbox{nearWinColor}{\textbf{Near Winner}} (0-5\%): Close to victory\\
\item \colorbox{contenderColor}{\textbf{Contender}} (5-20\%): Possible contender\\
\item \colorbox{competitiveColor}{\textbf{Competitive}} (20-30\%): Significant but not insurmountable\\
\item \colorbox{distantColor}{\textbf{Distant}} (30-45\%): Substantial gap to victory\\
\item \colorbox{farBehindColor}{\textbf{Far Behind}} ($>$45\%): Would have required transformative change to win
\end{tablenotes}
\end{threeparttable}
\label{tab:victory-gap}
\end{table}

\noindent
\textbf{Election Attributes:} The \textit{Victory Gap} column in \Cref{tab:victory-gap} shows the proximity to victory for each candidate; the lowest gap among non-winners, i.e., the \textit{margin of victory, is 17.07\%}.  The Exhaustion column reports the percentage of exhausted ballots at the time of elimination for each candidate. B is the only candidate with higher exhaustion percentage than the victory gap; however, all ballot completion models suggest \textit{close to zero impact of ballot exhaustion}, i.e., the result would not change when the exhausted ballots are completed according to the models. The Required Strategy column shows the nature of each candidate's optimal strategy, which is \textit{consistently ``selfish"} in all instances. The preference order alignment shows \textit{``no-match'' with the social choice order} as candidate E is closer to victory (28.95\% additions) despite getting eliminated prior to D (37.49\% additions).   

\subsection{Empirical Results across elections}
We now provide results on election attributes across all NYC'21 and Alaska'24 elections, with further findings summarized in \Cref{sec:summary_attributes}.
\subsubsection{Victory Gap and Margins of Victory}
\begin{figure}[h]
  \centering
  \includegraphics[width=0.48\linewidth]{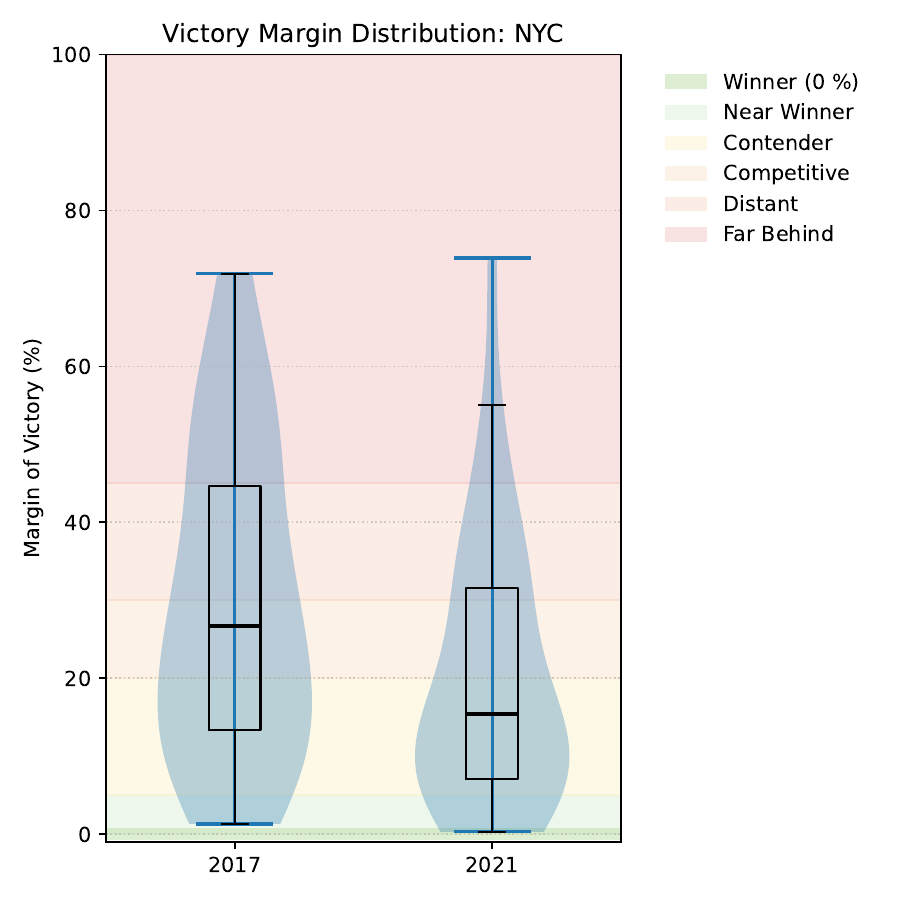}
  \hfill
  \includegraphics[width=0.48\linewidth]{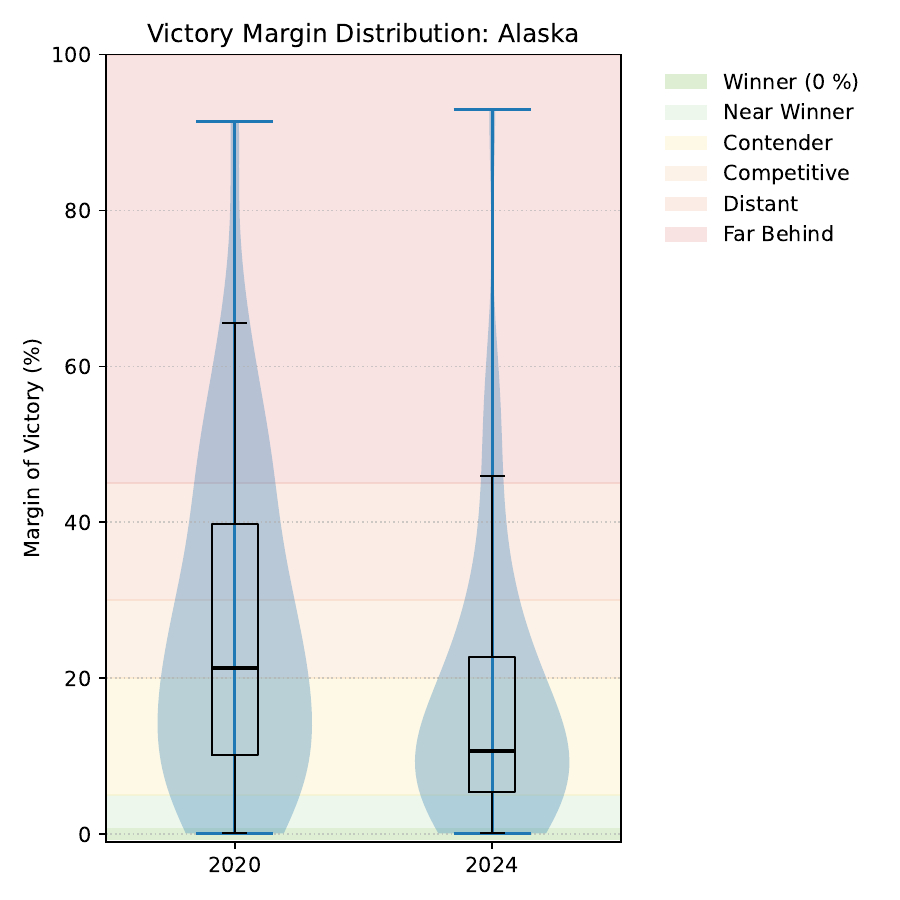}
  
  \caption{Margin of victory or competitiveness distributions before and after RCV adoption in NYC Primaries (left) and Alaska statewide contests (right), overlaid on competitiveness categories. The violin outlines smoothed density curves, and the internal boxplots mark the median and inter‑quartile range. Both jurisdictions show clear shifts toward tighter races after RCV implementation.}
  
  \label{fig:margins-violin-comparison}
\end{figure}
\Cref{tab:nyc_elections} and \Cref{tab:alaska_elections} in \Cref{app:single_winner_details} include the margins of victory for all candidates who can achieve a win under the analyzed allowance. \Cref{fig:margins-violin-comparison} presents a direct comparison of margin of victory distributions before and after RCV implementation in both NYC and Alaska. For previous plurality elections, the margin is simply calculated as the difference between the top two candidates, whereas for RCV elections, it is facilitated by ERSF. Both jurisdictions demonstrate significantly increased competitiveness following RCV adoption. 
Overall, the average margin of victory dropped by 9.2 percentage points in NYC primaries (a 31\% relative reduction from 2017 to 2021) and by 11.4 percentage points in Alaska statewide elections (a 42\% relative reduction from 2020 to 2024).

\subsubsection{Ballot Exhaustion}
\begin{figure}[h]
    \centering
     \includegraphics[width=0.52\linewidth]{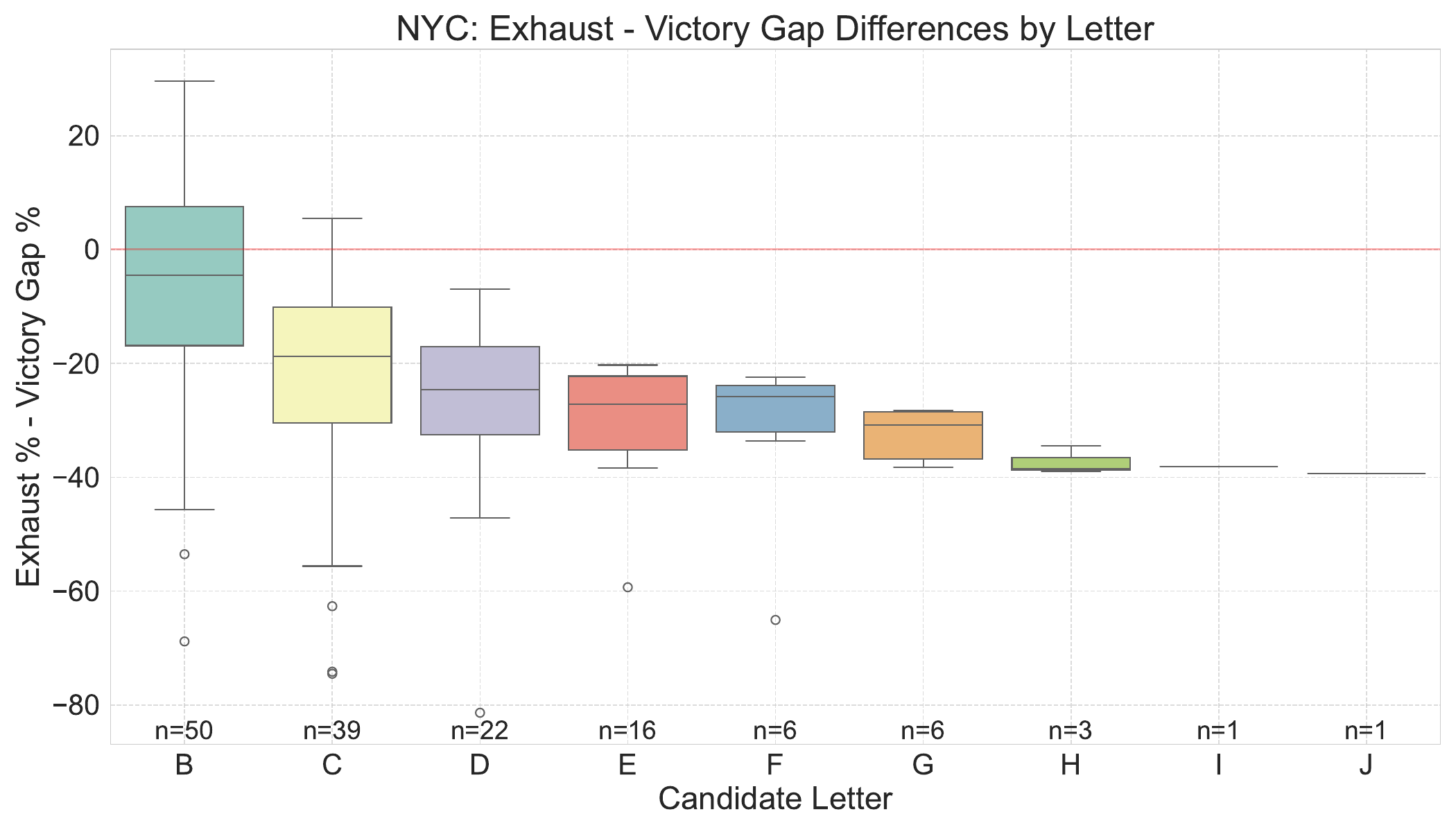}
    \includegraphics[width=0.47\linewidth]{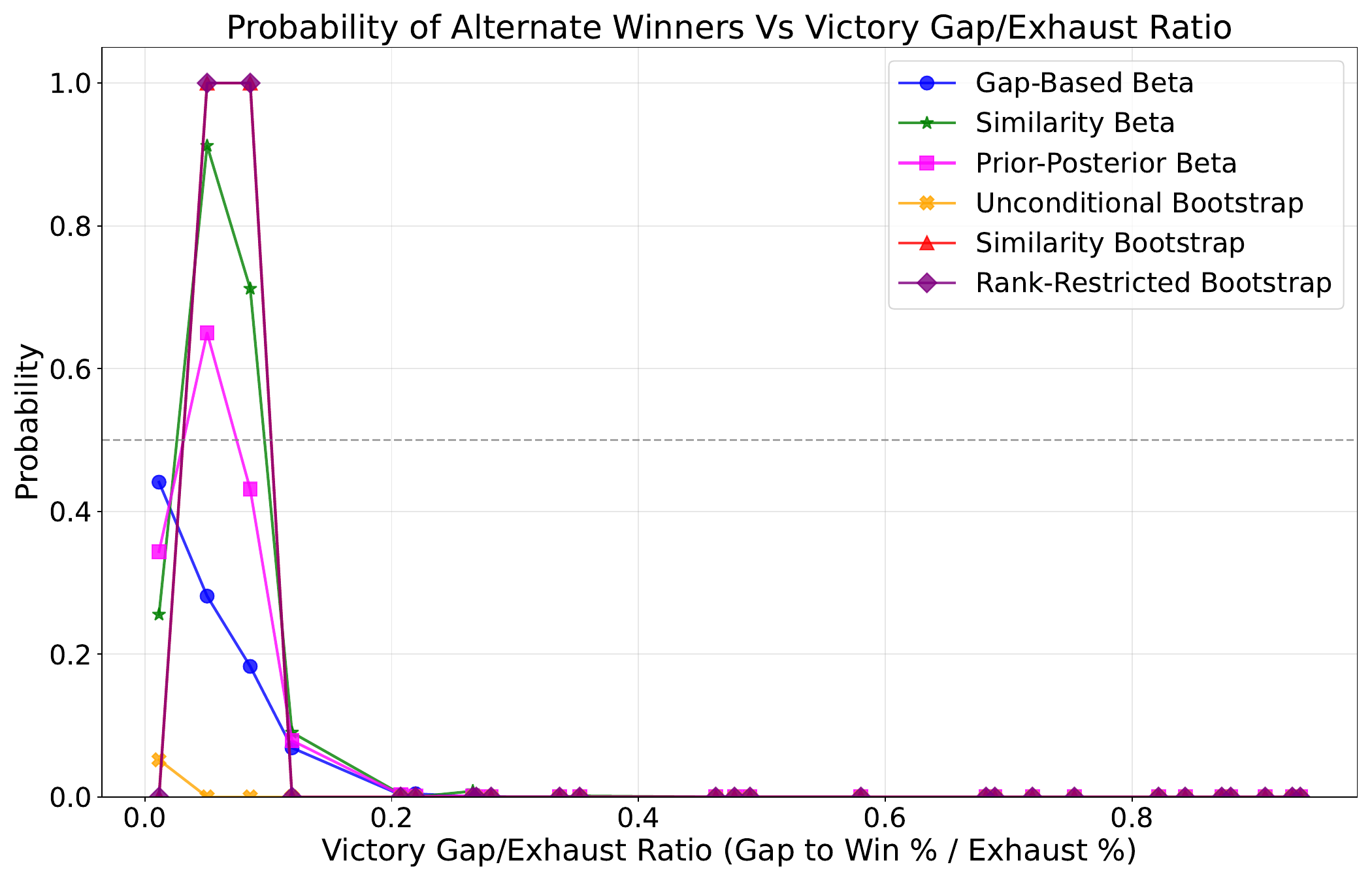}
    \caption{Ballot exhaustion in NYC elections. Left: Excess of exhausted ballots over victory gaps; if this difference is positive, ballot completion can lead to alternative winner(s). Right: Probability of alternate winners vs. victory gap-to-exhaustion ratio, using six different models.}
    \label{fig:exhaustion_nyc}
\end{figure}
\begin{figure}[h]
    \centering
    \includegraphics[width=0.52\linewidth]{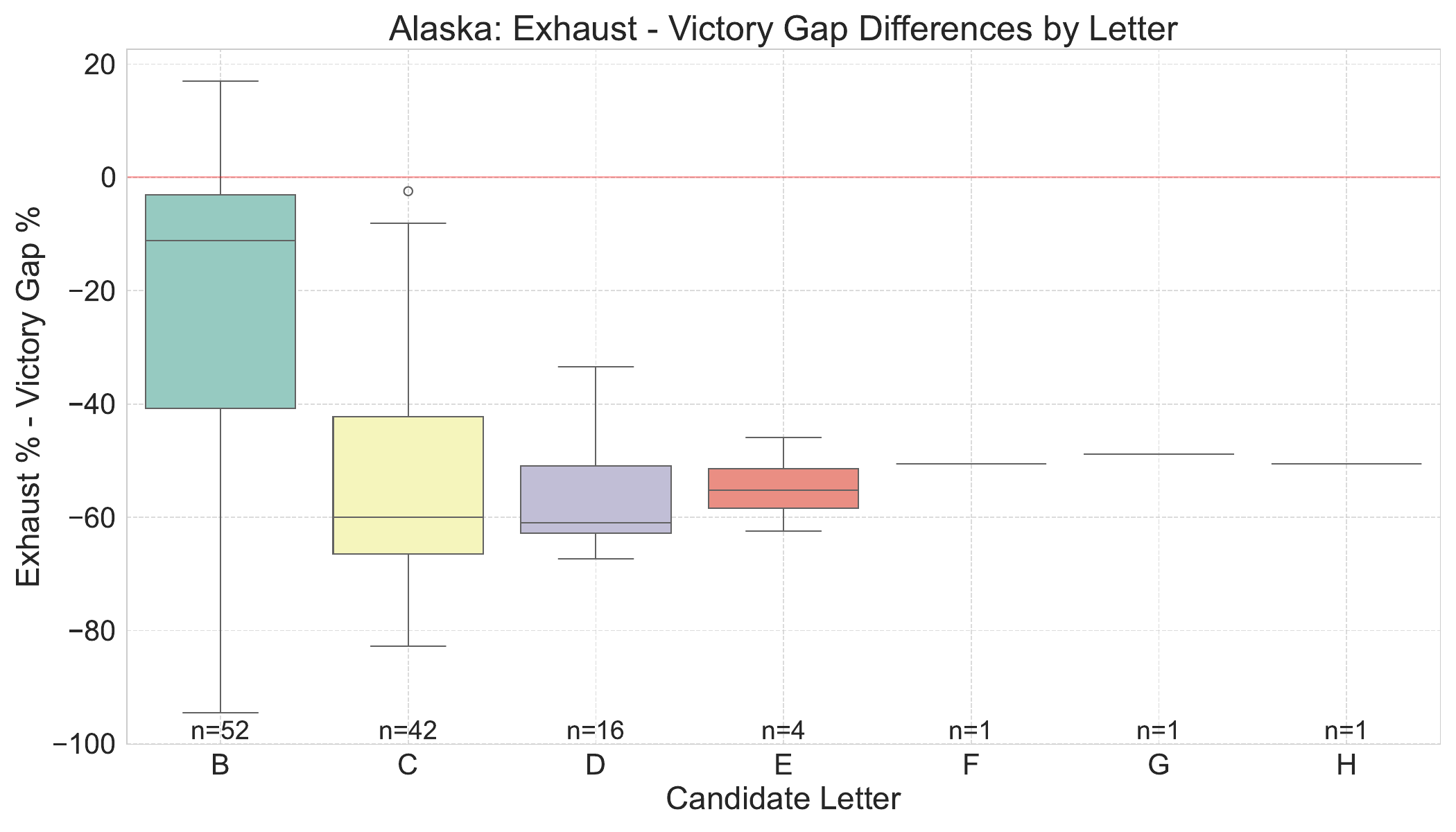}
    \includegraphics[width=0.47\linewidth]{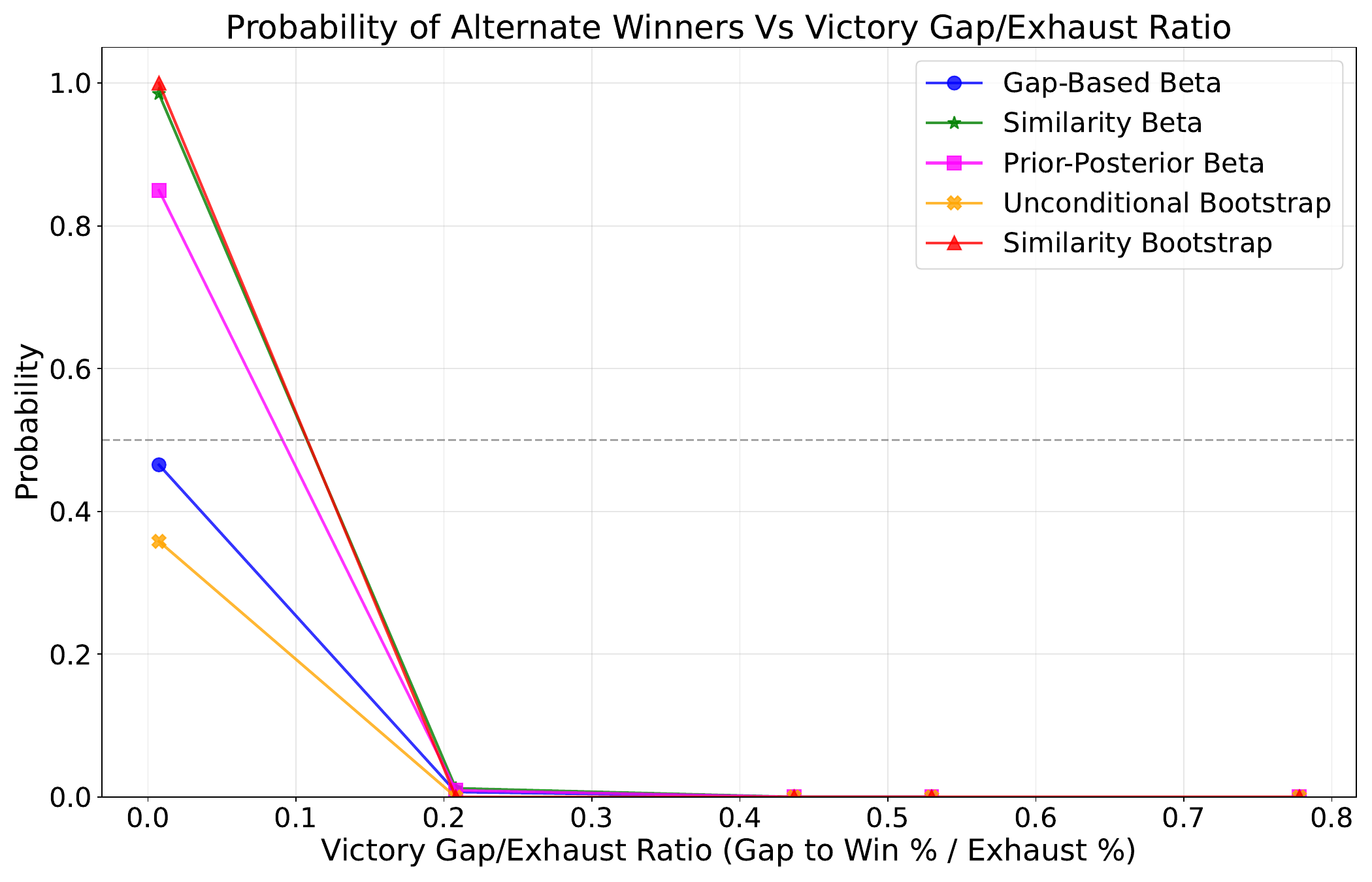}
    \caption{Ballot exhaustion in Alaska elections (same interpretation as Fig. \ref{fig:exhaustion_nyc}).}
    \label{fig:exhaustion_alaska}
\end{figure}
Both case studies reveal instances where ballot exhaustion exceeded victory gaps; however, ballot completion models demonstrate that 3/106 elections would yield outcome changes. \Cref{fig:exhaustion_nyc} illustrates that among NYC's 54 elections, 27 second-placed candidates (B) and 4 third-placed candidates (C) in the social choice order exhibited non-zero potential to win if all exhausted ballots were complete; all candidates post third-placed have zero probability of winning upon any completion of exhausted ballots. Of these, four elections demonstrated reasonable potential for candidate B winning across all models, though high flip probability emerged in only 2 cases.

The general completion models (the ``gap-based beta" and ``unconditional bootstrap," which do not condition on first-choice preferences) yield negligible outcome-flip probabilities across nearly all elections. However, conditioning on first-choice preferences via the two ``similarity"-based completion models identifies two exceptions: the NYC mayoral (margin of victory = 0.8\%) and Queens borough president (margin of victory = 0.5\%) elections. In the NYC mayoral primary, Eric Adams defeated Kathryn Garcia by 7,197 votes in the final round; however, over 140,000 ballots exhausted before this stage. Under similarity-based completion, we complete these exhausted ballots according to the empirical distribution of fully-ranked ballots sharing the same first-choice candidate, and compute the probability that Garcia gains more than 7,197 votes over Adams. Conditioning on first-choice preferences, Garcia wins with high probability—indicating that her supporters disproportionately submitted incomplete ballots.\footnote{This sensitivity is not driven by NYC's five-candidate ranking limit: enough voters ranked fewer than five candidates that the rank-restricted similarity bootstrap produces the same results as the unrestricted similarity bootstrap, indicating genuine sensitivity to ballot exhaustion rather than a ballot-design artifact.} The Queens Borough President election exhibited similar sensitivity.

Alaska demonstrated even greater stability, as illustrated in \Cref{fig:exhaustion_alaska}. Only 4 of 52 elections exhibited any outcome sensitivity to ballot exhaustion, where exclusively the second-ranked candidate (B) could win upon completing all exhausted ballots. Of these contests, only one—House District 28 (margin of victory = 0.1\%)—showed a substantial probability of a different outcome under general completion models; when first preferences were incorporated, the likelihood of an outcome change rose to near certainty.

Our analysis reveals that while ballot exhaustion is a necessary condition for outcome changes, election competitiveness and systematic bias in ballot exhaustion patterns were critical determinants of whether completing ballots would actually alter results. Only highly competitive elections with biased exhaustion patterns generated meaningful probabilities of alternate outcomes upon ballot completion. Additional details appear in \Cref{app:ballot_exhaustion_single_winner}.
\subsubsection{Strategic Complexity}

Both jurisdictions exhibited similar patterns regarding strategic complexity. All NYC elections showed selfish strategies for vote additions up to at least 20\%, with 43 out of 54 elections exhibiting exclusively selfish strategies within the tested range. Across all 52 Alaska elections, optimal strategies were exclusively selfish, with no non-selfish strategies identified.

For the 11 NYC elections with non-selfish strategies, the manipulation involved supporting rival candidates to engineer favorable elimination sequences. For instance, in District 18, candidate D would require 27.9\% under purely selfish strategy but 27.4\% with a non-selfish approach (2.4\% to rival B, 24.9\% to D); in District 23, candidate D would require 27.0\% selfishly but 20.3\% non-selfishly (2.2\% to rival B, 18.1\% to D). See Appendix~\ref{app:single_winner_details} for complete breakdown of all strategies (Table~\ref{tab:nyc_elections}) and detailed illustration of District 23's non-selfish strategies (Table~\ref{tab:district23_detail}).

\subsubsection{Preference Order Alignment}

Both case studies demonstrated strong preference order alignment between social choice order and the victory gap order. Among NYC's 54 elections, 41 cases showed perfect match between social choice order and candidates' ranking based on victory gaps under 40\% allowance. For up to 20\% allowance, this alignment became nearly perfect with 52 out of 54 elections accurately reflecting candidates' relative winning potential. 

Alaska showed an even stronger alignment: in all elections except the presidential race, candidates' strategic vulnerability order perfectly matched their social choice order. Almost every candidate found a feasible winning path within the 100\% allowance, and the relative order among competitive candidates consistently matched the order in which they got eliminated.

~\\\noindent To summarize, RCV elections in NYC and Alaska, despite being contrasting environments themselves, consistently exhibit strategic dynamics similar to plurality voting but with notable competitive advantages. Compared to prior plurality systems, margins of victory have become significantly tighter, indicating more competitive elections. Ballot exhaustion analysis confirms RCV's robustness: fewer than 8 out of 106 elections exhibited nontrivial probability to outcome change across completion models, and 3 demonstrated substantial impact where completing exhausted ballots would likely change results. Strategic complexity analysis shows that optimal strategies are consistently self-focused, negating incentives for elaborate strategic voting. Finally, candidate elimination sequences align closely with their competitive positions, indicating that RCV's procedures typically reflect competitive dynamics accurately rather than obscuring them.
\section{Multi-Winner STV: Portland City Council, 2024} \label{sec:casestudy_portland}

We next analyze multi-winner RCV by examining the 2024 Portland City Council elections. These elections provide a more intricate multi-winner Single Transferable Vote (STV) context with multi party dynamics and larger elections with 16-30 candidates, in contrast with the single-winner analyses in \Cref{sec:case_studies_rcv}. 

\subsection{Electoral Context and Data}

In November 2024, Portland used Ranked Choice Voting via STV for the first time to elect all 12 council members using multi-member council districts. This represented a fundamental structural shift from the city's previous commission government system, which had featured five at-large council members since 1913 \citep{portland_transition_2025}.

The multi-winner format supported coalition-based campaigning, with candidates securing cross-endorsements and establishing their own alliances \citep{FairVote2024PortlandRCV}. This raised questions about the system's susceptibility to strategic manipulation versus the role of voter engagement in determining outcomes \citep{Hutchinson2024PortlandRCV}.

Given Portland's structural transformation and the absence of historical plurality data, we could not employ the before-after competitiveness comparisons used in \Cref{sec:case_studies_rcv}. Instead, we conducted bootstrap resampling analysis on each district, generating multiple resampled elections and computing the election attributes across all samples. This approach provides statistical confidence despite the small sample size (4 elections versus over 50 for NYC and Alaska) and addresses the substantially more complex strategic space created by multi-winner STV, with each district featuring large fields of 16 to 30 candidates. We use the cast vote record \citep{MultnomahRCVResults}.

\subsection{Methodology}

For each district, we generated 100 bootstrap samples by resampling ballots with replacement from the cast vote records, creating simulated election instances that preserve underlying voter preference distributions. We apply ERSF to each sample, analyzing both the original election data (strategies under certainty) and bootstrap samples (strategies under uncertainty).

Given the large candidate fields, directly finding strategies is computationally prohibitive without candidate removal. We hence apply the ERSF, specifically using \Cref{algo:candidate-removal-updated} with the new removal condition in \Cref{thm: irrelevant_extension} for the removal of irrelevant candidates and then \Cref{algo: robust_allocation} to find optimal additions to reach a desired election result.

Given the differing characteristics of the four elections, we determine the allowance based on algorithmic traceability thresholds for each district. The highest feasible percentages for analysis are 4.7\%, 6.5\%, 12.36\%, and 9.6\% for the four districts---above these percentages, the ERSF does not eliminate enough candidates for the analysis to be computationally tractable. These thresholds may be contextualized relative to the Droop quota of 25\%, given the three-winner election structure. We conduct the bootstrap analyses with slightly lower allowances: 4\%, 6\%, 11.5\%, and 9\%, achieving 99.7\% algorithmic efficiency for District 1 and 100\% for Districts 2, 3, 4 in eliminating irrelevant candidates from the bootstrap samples.

\begin{table}[h]
\resizebox{\columnwidth}{!}{%
\begin{tabular}{|c|c|c|c|c|c|}
\hline
District & Number of votes & Candidates & Valid Ballot Turnout & \textbf{Possible winners} & \textbf{Max \% additions} \\ \hline
1        & 42686           & 16         & 42.64\%              & \textbf{7} & \textbf{4.7\%}                                       \\ \hline
2        & 77036           & 22         & 64.77\%              & \textbf{4} & \textbf{6.5\%}                                         \\ \hline
3        & 84323           & 30         & 65.97\%              & \textbf{3} &  \textbf{12.36\%}                                         \\ \hline
4        & 76563           & 30         & 63.85\%              & \textbf{4} &  \textbf{9.6\%}                                         \\ \hline
\end{tabular}%
}
\caption{Summary of the 2024 Portland RCV elections and our analysis of candidate relevancy (in bold). In District 1, the race was the most competitive—7 of the 16 candidates are either winning or within reach with up to a 4.7\% increase in votes. In contrast, the other districts showed little competition, with only one or no additional candidates close to winning, even with substantial allowance. } \label{Tab:portland_summary}
\end{table}

Among the four districts, District 1 saw the lowest voter turnout as well as highest competitiveness with many candidates having closer vote-shares in the last few rounds. Similar to NYC's District 1 analysis in \Cref{subsec:dis1_nyc}, we first illustrate our approach on Portland's District 1, and then present comprehensive analysis of the remaining districts, including the bootstrap results. We refer to \Cref{app:portland} for additional computational details.

\subsection{Illustrative Example: Portland's First 3-Member District}

District 1 election featured 16 candidates, and Candace Avalos, Loretta Smith, and Jamie Dumphy ultimately won (see \Cref{app:portland_dis1_info} for result details). Here, applying \Cref{algo:candidate-removal-updated} removed 8 irrelevant candidates, with the traceability threshold being 4.7\%. Among the 8 \emph{relevant} candidates within the threshold, 7 appear in the set of potentially winning candidates. The remaining candidates would \emph{necessarily} require more than 4.7\% additions to win, and their strategic analysis remains constrained by complexity.  \Cref{tab:portland_dis1_strats} illustrates how our analysis transforms the multi-winner RCV dynamics into a plurality-akin presentation, under the three-winner context.

\begin{table}[h]
\centering
\caption{Plurality-akin presentation of Portland RCV District 1 (2024), Top 7 of 16 Candidates}
\begin{threeparttable}
\small
\begin{tabular}{lccrc}
\toprule
\textbf{Candidates} & ID & \textbf{Victory Gap (\%)} & \textbf{Strategic Complexity} & \textbf{Ballot Exhaustion} \\
\midrule
\rowcolor{winnerColor} Candace Avalos & (A) & 0.00 & \textit{Actual winner} & - \\
\rowcolor{winnerColor} Loretta Smith & (B) & 0.00 & \textit{Actual winner} & - \\
\rowcolor{winnerColor} Jamie Dumphy & (C) & 0.00 & \textit{Actual winner}  & - \\
\rowcolor{nearWinColor} Terrence Hayes & (D) & 1.60 & Selfish (1.60\%)  & 14.59\% \\
\rowcolor{nearWinColor} Steph Routh & (F)  & 1.93  & Selfish (1.93\%) & 6.56\% \\
\rowcolor{contenderColor} Noah Ernst & (E)  & 2.81 & Selfish (2.81\%) &  8.41\%  \\
\rowcolor{contenderColor}  Timur Ender & (G) & 4.21 & Selfish (4.21\%) &  5.14\%  \\
\bottomrule
\end{tabular}

\begin{tablenotes}[para,flushleft]
\small

\vspace{0.3cm}
\item \textbf{Categories:} (Normalized from 50\% quota for 1-winner to 25\% for 3-Winner RCV)\\
\item \colorbox{winnerColor}{\textbf{Winner}} (0\%): No additional votes needed\\
\item \colorbox{nearWinColor}{\textbf{Near Winner}} (0-2.5\%): Would have won with small additional support\\
\item \colorbox{contenderColor}{\textbf{Contender}} (2.5-10\%): Needed moderate additional support to win
\end{tablenotes}
\end{threeparttable}
\label{tab:portland_dis1_strats}
\end{table}

\noindent \textbf{Election Attributes:} 
The \textit{Victory Gap} column shows the percentage of additional votes a candidate would need to displace Dumphy while keeping Avalos and Smith as winners. The `Ballot Exhaustion' column shows the percentage of ballots exhausted when each candidate was eliminated: While all 4 candidates remain viable, ballot completion models suggest \textit{at most a 10\% likelihood of an alternate outcome}, where Hayes emerges as the winner. All optimal strategies are \textit{selfish} (self-support). Notably, although Steph Routh (F) was eliminated before Noah Ernst (E), Routh had a narrower victory gap, highlighting the \textit{partial preference order alignment mismatch} where the elimination sequence partially conceals the real competitiveness.  

\subsection{Results}
We applied the ERSF to both the actual election data and bootstrap samples for each district. For the bootstrap samples, we analyzed strategic possibilities within allowances slightly below the algorithmic traceability thresholds to ensure computational reliability.

\Cref{tab:summary_dis1,tab:multi_district_summary} present comprehensive breakdowns of strategic complexity distributions across all bootstrap samples. 

\begin{table}[htbp]
\centering
\caption{District 1 Bootstrap Analysis Summary }
\label{tab:summary_dis1}
\footnotesize
\begin{tabular}{>{\raggedright\arraybackslash}p{1.5cm}cp{0.9cm}p{0.9cm}p{0.9cm}p{0.9cm}p{6.2cm}}
\toprule
\textbf{Winning Coalition} & \textbf{Freq.  (\%)} & \textbf{Mean (\%)} & \textbf{Std. (\%)} & \textbf{Min (\%)} & \textbf{Max (\%)} & \textbf{Strategic Complexity} \\
\midrule
\rowcolor{winnerColor} A, B, C & 100.0 & 0.00 & 0.00 & 0.00 & 0.00 & None (actual winners) \\
\addlinespace[0.5ex]
\rowcolor{nearWinColor} A, B, D & 98.8 & 1.70 & 0.27 & 1.19 & 2.46 & Selfish: D self-support (100\%) \\
\addlinespace[0.5ex]
\rowcolor{nearWinColor} A, B, F & 98.8 & 1.90 & 0.24 & 1.39 & 2.42 & Selfish: F self-support (100\%) \\
\addlinespace[0.5ex]
\rowcolor{contenderColor} A, B, E & 97.6 & 2.88 & 0.25 & 2.36 & 3.68 & 
\begin{tabular}[t]{@{}l@{}}
• 87.8\%: Selfish (E self-support 100\%)\\
• 12.2\%: Non-Selfish (E 97.5\% + DE  2.5\%) 
\end{tabular} \\
\addlinespace[0.5ex]
\rowcolor{contenderColor} A, B, G & 21.4 & 3.86 & 0.12 & 3.49 & 3.96 & 
\begin{tabular}[t]{@{}l@{}}
• 38.9\%: Selfish (G self-support 100\%)\\
• 61.1\%: Non-Selfish \\
(G 95.8\% + DG 2.48\% + D 1.73\%) 
\end{tabular} \\
\addlinespace[0.5ex]
\rowcolor{contenderColor} A, C, D & 16.7 & 3.78 & 0.17 & 3.50 & 3.99 & 
\begin{tabular}[t]{@{}l@{}}
• 85.7\%: Non-Selfish \\(C 17.68\% + D 49.70\% + FD 32.62\%)  \\
• 14.3\%: Non-Selfish \\(A 0.63\% + C 18.70\% + D 48.61\% + \\
\phantom{• 14.3\%:} FD 32.06\%) 
\end{tabular} \\
\bottomrule
\end{tabular}

\smallskip
\begin{minipage}{\textwidth}
\footnotesize
\textit{Note:} Strategies for District 1 over 100 bootstrap samples with up to 4\% strategic vote additions. ``Freq." shows the percentage of samples in which each coalition could win within the 4\% allowance. ``Mean" represents the average victory gap (additional votes needed as \% of total votes) across samples where the coalition was achievable. The actual winners (A,B,C) emerged in all 100 samples with 0\% gap. Although strategies are selfish on election data, some non-selfish strategies appear on bootstrap samples.
\end{minipage}
\end{table}
\begin{table}[h!]
\centering
\caption{District 2,3, and 4 Bootstrap Analysis Summary }
\label{tab:multi_district_summary}
\footnotesize
\begin{tabular}{c l c r r r r l}
\toprule
\textbf{District} & 
\textbf{Winners} & \textbf{Freq} & \textbf{Mean} & 
\textbf{Std.} & \textbf{Min} & \textbf{Max } & 
\textbf{Strategic Complexity} \\
\midrule
\rowcolor{winnerColor}
2 & A, B, C & 100 & 0.0     & 0.0   & 0    & 0    & None \\
\addlinespace[0.5ex] \rowcolor{contenderColor}
 & A, C, D & 85 & 5.64\%     & 0.17\%   & 5.21\%    & 5.99\%    & Non-Selfish: D= 99.98\%, A= 0.02\% \\
\addlinespace[0.5ex] \rowcolor{contenderColor}
  & A, B, D & 1 & 5.93\%     & 0   & 5.93\%    & 5.93\%    & Selfish: D= 100\% \\
\hline
\rowcolor{winnerColor}
3 & A, B, C & 100 & 0.0     & 0.0   & 0    & 0    & None \\
\hline
\rowcolor{winnerColor}
4 & A, B, C & 100 & 0.0     & 0.0   & 0    & 0    & None \\
\addlinespace[0.5ex] \rowcolor{nearWinColor}
 & A, B, D & 100 & 1.12\%  & 0.23\%   & 0.32\%	 & 1.64\% &
Selfish: D = 100\%\\
\hline
\end{tabular}
\begin{minipage}{\textwidth}
\footnotesize
\vspace{0.1cm}
\textit{Note:} Strategies for Districts 2, 3, and 4 over 100 bootstrap samples each. District 3 only supports A,B,C within the allowance. Districts 2 and 4 can elect candidate D with additional votes. 
\end{minipage}
\end{table}

\subsubsection{Victory Gap and Margin of Victory}
Multi-winner RCV elections in Portland exhibited substantial variation in competitiveness across districts. Districts 1 and 4 proved most competitive with margins of victory of 1.60\% and 1.09\% respectively. District 2 required 5.7\% additional votes to change winners. District 3 was least competitive, with no successful strategies possible within the 12.36\% allowance; however, computational analysis indicates an upper bound on competitiveness of approximately 13.6\%, the threshold at which the fourth-place candidate could secure victory.

The bootstrap results in \Cref{tab:summary_dis1,tab:multi_district_summary} confirm the robustness of these competitive standings: the actual winning coalitions emerged in 100\% of bootstrap samples, while trailing candidates exhibited consistent victory gaps across samples (as shown by low standard deviations).

The variation across districts reflects underlying competition in the official election results \citep{MultnomahRCVResults}. For District 1, the results confirm the original winners' strength—A consistently prevails, B comes next, while C could be substituted with up to 4\% additional new votes. Districts 2 and 3 show that their top three candidates are significantly ahead of their opponents in the first round itself—(13\%, 16\%, 16\%) top 3 versus $\leq 9\%$ for remaining candidates and (24\%, 19\%, 19\%) top 3 versus $\leq 6\%$ rest, for Districts 2 and 3, respectively—while District 4 has the top 4 relatively closer to each other—(24\%, 14\%, 10\%, 11\%) versus 6\% rest.

\subsubsection{Ballot Exhaustion}
Given larger candidate fields, the Portland elections exhibited substantially higher ballot exhaustion rates than single-winner contests, from less than 1\% in early rounds to 47\% by final eliminations. This created more widespread outcome sensitivity to ballot exhaustion: all candidates in Districts 1, 2, and 4 whose victory gaps could be determined showed exhaustion rates exceeding those victory gaps, compared to the sporadic cases observed in single-winner elections in \Cref{sec:case_studies_rcv}. 

\begin{table}[htbp]
\centering
\footnotesize
\begin{tabular}{|l|c|c|c|c|c|c|c|}
\hline
\textbf{Candidate} & \textbf{Exhaust} & \textbf{Gap} & \textit{Gap Beta} & \textit{C. Bootstrap} & \textit{U. Bootstrap} & \textit{Similarity} & \textit{Prior-Post} \\
\hline
D, Dis 1 & 14.59\% & 1.60\% & 10.42\% & 0.0\% & 0.0\% & 0.0\% & 0.10\% \\
\hline
E, Dis 1 & 8.41\% & 2.81\% & 0.01\% & 0.0\% & 0.0\% & 0.0\% & 0.0\% \\
\hline
F, Dis 1 & 6.56\% & 1.93\% & 0.07\% & 0.0\% & 0.0\% & 0.0\% & 0.0\% \\
\hline
G, Dis 1 & 5.14\% & 4.21\% & 0.0\% & 0.0\% & 0.0\% & 0.0\% & 0.0\% \\
\hline
D, Dis 2 & 9.75\% & 5.69\% & 0.0\% & 0.0\% & 0.0\% & 0.0\% & 0.0\% \\
\hline
D, Dis 4 & 14.08\% & 1.12\% & 18.21\% & 0.0\% & 0.0\% & 0.27\% & 3.27\% \\
\hline
\end{tabular}
\caption{Ballot exhaustion analysis in Portland Districts: Probability of winning for trailing candidates whose exhaustion rates exceed their victory gaps, using five different ballot completion models. All models---gap-based beta, category (first-preference or similarity based) and unconditional bootstrap, similarity beta and prior-posterior beta, as detailed in \Cref{app:ballot_exhaustion}---suggest low probabilities of outcome changes. 
}
\label{tab:portland_district_exhaustion}
\end{table}

At first glance, the theoretical potential for outcome changes appears substantial. Victory gaps represented only 29.7\% of available exhausted ballots on average, suggesting that completing these ballots could easily change outcomes. However, multi-winner contexts impose more complexity in ballot completion: for a fourth-place candidate to displace the winning coalition, exhausted ballots must be completed with preferences that rank this challenger above all three actual winners simultaneously.

This demanding preference requirement is reflected in the actual probabilities of outcome changes, under our models of how votes are completed. As \Cref{tab:portland_district_exhaustion} demonstrates, even the strongest case achieved only 18.21\% probability under the ``victory gap" based beta model, which uses gap-based beta parameters for computing winner probabilities. Moreover, this phenomenon is reinforced by voter preferences toward the current winners, since the first-preference based beta and bootstrap models approach zero probabilities across all districts. Rather than creating vulnerability, multi-winner preference complexity appears to strengthen RCV's robustness to ballot exhaustion by making the required voter coordination difficult to achieve.
\subsubsection{Strategic Complexity}
Multi-winner RCV elections maintained the predominantly selfish strategic patterns observed in single-winner contests, despite operating in substantially more complex environments. All candidates in the actual election data exhibited exclusively selfish strategies.

Bootstrap analysis revealed limited instances of non-selfish strategies in Districts 1 and 2, but these remained minority patterns compared to selfish alternatives. As depicted in \Cref{tab:summary_dis1} and \Cref{tab:multi_district_summary}, even where non-selfish strategies emerged, they typically involve minimal coalition building among non-winning candidates rather than complex strategic interdependencies.

Notably, very few optimal strategies required supporting actual winners. Strategic vote additions majorly supported the candidate of interest or formed coalitions among challengers. This suggests that RCV does not create incentives for complex political deals or mutually beneficial cross-endorsements between winners and challengers, maintaining clear strategic boundaries even in multi-winner contexts.

\subsubsection{Preference Order Alignment}

Multi-winner elections showed general alignment between social choice order and victory gap order, with some exceptions revealing different competitive dynamics. District 3 had only 3 candidates' victory gaps quantified due to computational constraints. Districts 2 and 4 demonstrated perfect alignment up to the fourth candidate, indicating that elimination sequences accurately reflected competitive strength.

In District 1, candidate F exhibited a smaller victory gap than candidate E despite being eliminated earlier in the Social Choice Order. This preference order mismatch appears consistently across both actual election data and bootstrap samples, suggesting that formal elimination sequences may occasionally obscure true competitive relationships even in multi-winner contexts (See \Cref{tab:summary_dis1}).

~\\\noindent
Overall, despite having quite diverse and more complicated multi-winner dynamics, Portland's 2024 RCV elections largely echo patterns similar to NYC and Alaska's single-winner elections. Optimal strategies remain largely selfish and single-choice even in elections with up to 30 candidates. Although ballot exhaustion rates are high, the probability of alternate outcomes is near-zero. Crucially, strategies across all districts rarely support actual winners, these are either selfish or support eventual losing candidates. The lack of mutually beneficial strategies suggests that fruitful cross-endorsements or mutually beneficial coalitions would find no actual support. Furthermore, bootstrap analysis confirms the robustness of these insights, indicating that the analyzed RCV attributes apply broadly rather than being specific to the exact election data. Combined with insights from NYC and Alaska's single-winner analysis, this informs how RCV competitive dynamics, powered by the ERSF, remain largely \emph{consistent} and \emph{simple} across wide electoral contexts.

\section{Conclusion}\label{sec:conclusions}
This work presents a comprehensive empirical analysis of Ranked Choice Voting (RCV) system across 110 real-world elections, addressing longstanding questions about RCV's practical complexity through systematic computational methods. We extend recent algorithmic advances to enable a first-of-its kind, large-scale computational assessment 
of RCV's practical performance through interpretable, plurality-comparable metrics.

Four key findings emerge from our analysis. First, RCV coincided with an increase in electoral competitiveness---margins of victory dropped by 9.2 percentage points in NYC and 11.4 points in Alaska after the adoption of RCV. Second, the impact of ballot exhaustion proved largely inconsequential: completing the exhausted ballots using a variety of models would not alter the winner in 107 of 110 elections, demonstrating RCV’s robustness against ballot exhaustion effects.  Third, strategic behavior under RCV is largely selfish, with majority elections showing only selfish addition strategies to be optimal, indicating empirical resistance to elaborate manipulations. Fourth, RCV elimination order rarely masks underlying competitive dynamics, with 106 of 110 elections showing competitive candidate rankings by victory gap that fully align with elimination sequences.

As RCV adoption accelerates—reaching 14 million voters and with legislative proposals like the Fair Representation Act under consideration—our framework offers several applications: (1) providing interpretable RCV metrics on election night, (2) facilitating system assessment through standardized plurality comparisons, (3) enabling polling and bootstrap robustness analyses, and (4) supporting redistricting processes with RCV-specific competitiveness metrics, by coupling ERSF with simulated or polling-based ranked-ballot data for proposed districts. Our complete datasets and code are publicly available.

\newpage
\section*{Statements and Declarations}
Competing Interests: On behalf of all authors, the corresponding author states that there is no conflict of interest.\\
Data Availability: All datasets and analysis code supporting the findings of this study are publicly available at \url{https://github.com/sanyukta-D/Optimal_Strategies_in_RCV}.\\
Ethics Approval: This study did not involve human participants or animals, and therefore did not require ethics committee approval.

\bibliographystyle{plainnat}
\bibliography{biblio}

\newpage
\appendix

\section{Illustration: Official Election Results}

We present two real-world official election outcomes: a single-winner RCV (IRV) election and a multi-winner RCV (STV) election.

\subsection{NYC Democratic Primary Council District 1, 2021}\label{app:dis1_nyc_full}

\Cref{fig:dis_1_nyc_official} presents the official RCV election results for New York City's 2021 Democratic primary in Council District 1, obtained from the NYC Board of Elections website \cite{NYCBOE2021}. This example illustrates the standard tabular format used to present single-winner RCV outcomes.

\begin{figure}[h]
    \centering
    \includegraphics[width=18cm]{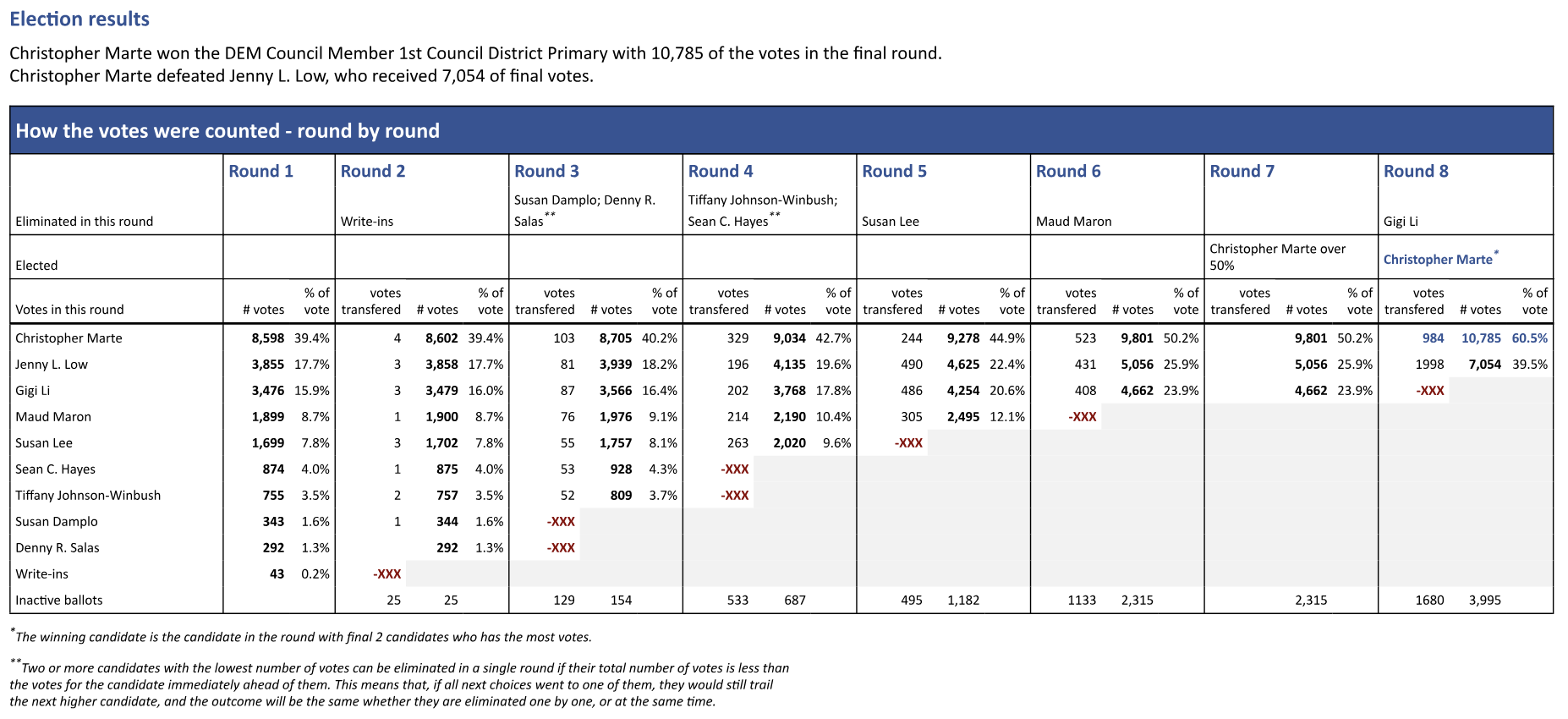}
    \caption{NYC 2021 Democratic Primary Council District 1}
    \label{fig:dis_1_nyc_official}
\end{figure}

\subsection{Portland Multi-Member District 1, 2024}\label{app:portland_dis1_info}

Portland's 2024 RCV election in District 1 featured 16 candidates competing for three seats. The complete official results comprise a 17-round RCV computation \citep{MultnomahRCVResults} that follows the same tabular format as \Cref{fig:dis_1_nyc_official}; we omit this lengthy table for brevity. Instead, we reproduce an image from the interactive results visualization presented on the official website (\Cref{fig:portland_dis1}). This visualization includes a slider bar that allows users to view vote counts for each candidate across all rounds; the figure shows the final round results.

\begin{figure}
    \centering
    \includegraphics[width=0.8\linewidth]{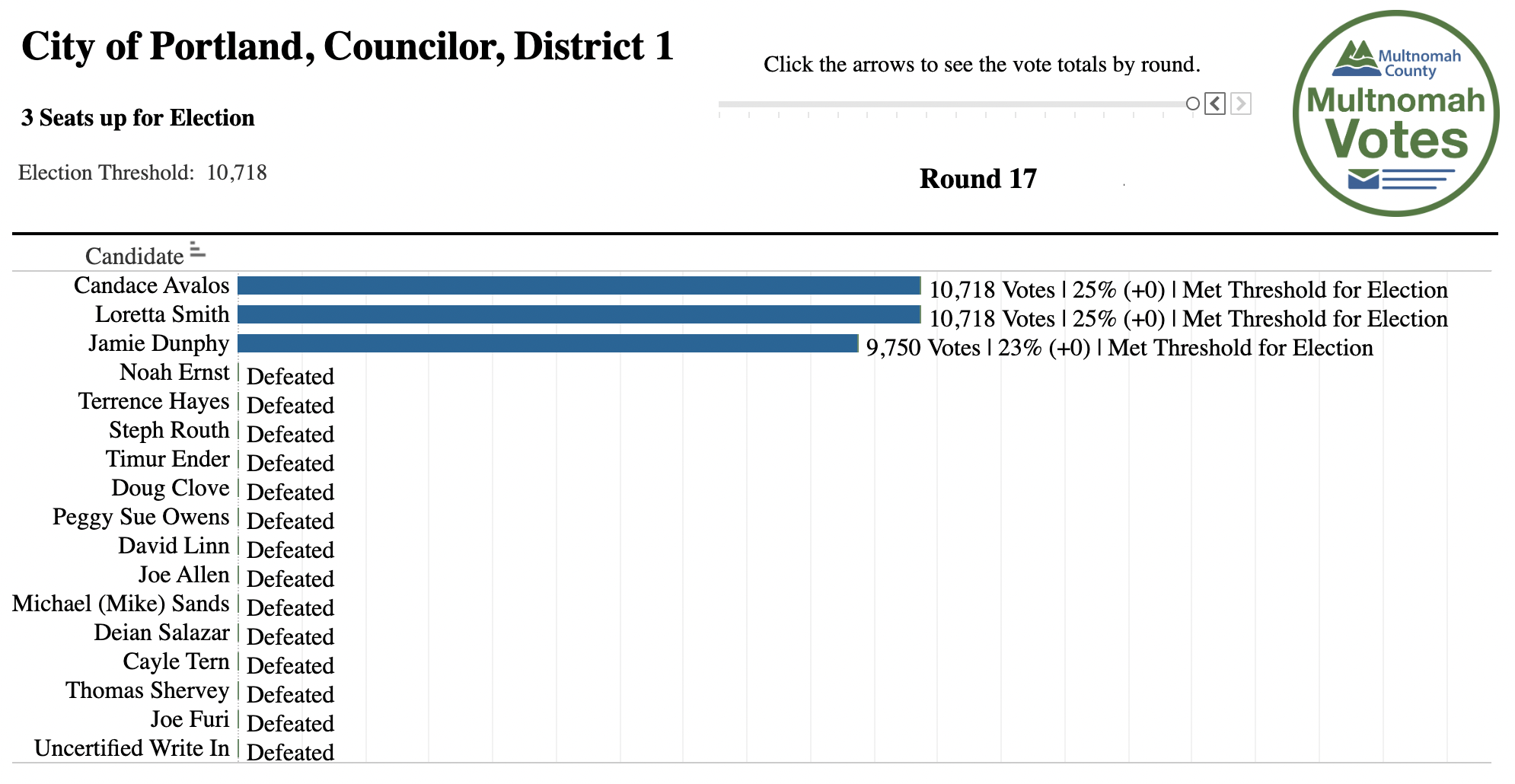}
    \caption{Portland 2024 3-Member District 1}
    \label{fig:portland_dis1}
\end{figure}
Although the tabular presentation and the interactive step-by-step visualization convey the determination of winners effectively, they fail to capture the complete competitive and comparative dynamics of the election. To grasp these dynamics, one must consider the myriad alternative scenarios that our framework concisely summarizes.
\section{Theory: Details}
\subsection{Algorithms from prior work}\label{app:prior_algo}

We briefly review Algorithms~\ref{algo: allocation} and \ref{algo:candidate-removal} from prior work for analyzing RCV. Algorithm~\ref{algo:candidate-removal} reduces computational complexity by identifying and removing provably irrelevant candidates—those guaranteed to be eliminated early regardless of strategic interventions. Algorithm~\ref{algo: allocation} then computes optimal ballot additions to achieve specific elimination/win sequences among the remaining relevant candidates. Together, these algorithms provide a framework for analyzing strategic manipulation in ranked-choice elections while maintaining computational tractability.

Algorithm~\ref{algo: allocation} (\textsc{SmartAllocation}) determines the ballot additions needed to achieve a specified elimination/win sequence in RCV. It processes each election round sequentially, adding votes to ensure the desired candidate either gets eliminated or wins that round. The algorithm handles both elimination scenarios (adding votes to other candidates) and win scenarios (adding votes to the target candidate) while tracking the total budget used.

Algorithm~\ref{algo:candidate-removal} (\textsc{IrrelevantCandidateRemoval}) attempts to reduce problem size by identifying candidates who will be eliminated in the early rounds. It builds a set $L$ of potential irrelevant candidates and checks whether they can be safely removed using the condition at line~\ref{line: ifconditions}. This condition verifies that candidates in $L$ would be eliminated before any candidate outside of $L$, even with strategic vote additions up to budget $B$.

We also reproduce the definition of Strict-Support for reference. 

\begin{definition} \label{def: strict_support}
    $\text{Strict-Support}_{L, G}(C_i)$ is the number of votes candidate $C_i$ receives from voters that rank candidates in $L$ the first, excluding votes that rank $C_i$ post any of $G$ candidates. Here, $L$ and $G$ are mutually exclusive sets. Accordingly, $\text{Strict-Support}_{L, \{\}}(C_i)$ becomes the total number of votes candidate $C_i$ receives (at any rank) from voters that rank candidates in $L$ the first.
\end{definition}

Strict-Support is useful in capturing the maximum number of votes an active candidate can have at a particular round. It is used in \Cref{algo:candidate-removal} in condition \ref{line: ifconditions}.

\begin{algorithm}
\caption{Allocation Rule (by \citet{deshpande2024optimal})} \label{algo: allocation}
\begin{algorithmic}[1]
\REQUIRE  $\mathcal{C}$, voter data $\{v\}_{\mathcal{C}}$, Constants $B$, $Q$
\ENSURE Budget allocation, under feasibility
\STATE \textbf{function} \textsc{SmartAllocation}( $\mathcal{C}$, $\{v\}_{\mathcal{C}}$, $B$, $Q$)
  \STATE Initialize $\text{votes}, \text{ votes}_{\text{new}} \leftarrow 0$.

  \FOR{processing each round}
    \STATE $C_s \leftarrow $ desired candidate for elimination or win in this round
      \IF{$C_s$ needs elimination}
       \IF{all in-contest candidates have votes $< Q$}
        \STATE Add $\text{votes}_{\text{new}} \ \forall  \ C_i : \  v^r_{C_i} + \text{votes}_{\text{new}} \geq v^r_{C_s} + 1$.
        \ELSE
        \RETURN infeasible
        \COMMENT{Increase $Q$ (i.e.,  budget $B$).}
        \ENDIF
        \ENDIF
      \IF{$C_s$ needs to win}
       \IF{all in-contest candidates have votes $< Q$}
        \STATE Add $\text{votes}_{\text{new}}\text{ to } C_s : \  v^r_{C_s} + \text{votes}_{\text{new}} \geq Q +1$.
        \ELSE
        \STATE Add $\text{votes}_{\text{new}} \text{ to } C_s : \  v^r_{C_s} + \text{votes}_{\text{new}}\geq \max(v^r_{C_i}) + 1$.
        \ENDIF
        \STATE $C_{last} \gets$ \text{ Candidate who went out of contest in the previous round}
        \IF{ $C_{last}$ doesn't transfer to $C_s$}
        \STATE Add $\text{votes}_{\text{new}}  [C_{last}, C_s] : \  v^r_{C_{last}}\leftarrow v_{C_{last}} + 1$.        \IF{$margin(v^r_{C_{last}})< 1$} 
        \RETURN infeasible \COMMENT{Increase $Q$ (i.e.,  budget $B$).}
        \ENDIF
        \ENDIF
      \ENDIF
      \STATE  $\text{votes} \leftarrow \text{votes} + \text{votes}_{\text{new}}$ 
      \STATE $\mathcal{C}\gets \mathcal{C}/ C_s$, update  $\{v\}_\mathcal{C} $
      
\ENDFOR
      \IF{votes $\geq$ budget}
       \RETURN feasible \COMMENT{Distribution of added votes}
       \ELSE
       \RETURN infeasible \COMMENT{Increase budget $B$}
    \ENDIF
\end{algorithmic}
\end{algorithm}

\begin{algorithm}
\caption{Removal of Irrelevant Candidates (by \citet{deshpande2024optimal})}
\label{algo:candidate-removal}
\begin{algorithmic}[1]
\REQUIRE  $\mathcal{C}$ in increasing order of $V^1$, Dictionary of orders (ballots) and values (votes for those ballots) $\{D\}_I$, Constants $B$, $Q$
\ENSURE Smallest relevant candidate set $\mathcal{C}$
\STATE \textbf{procedure} \textsc{IrrelevantCandidateRemoval}($\mathcal{C}, \{D\}_I, B, Q$)
\STATE $L, S \gets  \{\}, $  \text{dictionary with keys from set $\mathcal{C}$ and values as zeros}
\WHILE{$B+S[C_i]<Q$ for all {$C_i \in L$} \textbf{AND} $\mathcal{C} \neq \phi$}
    \STATE Shift first  candidate $C_1\in \mathcal{C}$ to set $L$
    \STATE $S \gets \textsc{Strict-Support}(L, \mathcal{C}\setminus L, \{D\}_I$)
    \STATE $S_j^i  \gets \textsc{Strict-Support}(C_j\cup L\setminus C_i, \{\}, \  \{D\}_I) \ \forall C_i, C_j \in L, \mathcal{C}\setminus L $
    \IF{$ B+ S[C_i] < S_j^i[C_j]< Q  \ \ \forall \ C_i, C_j \in L, \mathcal{C}\setminus L$} \label{line: ifconditions}
        \STATE $\{D\}_I \gets \textsc{ReduceElectionInstance}(L, \{D\}_I$)
        \STATE $\mathcal{C}, L \gets \mathcal{C}\setminus L, \{\} $ 
        \STATE Sort $\mathcal{C}$ in increasing order of updated $V^1$
    \ENDIF
\ENDWHILE
\RETURN $\mathcal{C}, \{D\}_I$
\end{algorithmic}
\end{algorithm}

\subsection{Proofs}\label{app:proofs}

\robustchoices*
\begin{proof}
We prove each extension separately:

(1) \textit{Robustness to subsequent choices:} When implementing spoiler-inducing strategies, subsequent preferences on added ballots can affect the election outcome after the elimination of certain candidates. The basic \textsc{SmartAllocation} procedure can be enhanced to ensure robustness against any possible subsequent preferences.

The key insight is to track `free votes' (ballots from inactive candidates) and compute worst-case transfers from these votes to each remaining candidate. We modify line 7 in Algorithm~\ref{algo: allocation} for Case-(L) to:
\begin{equation}
\text{Add } \text{votes}_{\text{new}} \ \forall  \ C_i : \  v^r_{C_i} + \text{votes}_{\text{new}} + \text{max\_transfers}[C_i] \geq v^r_{C_s} + \text{max\_transfers}[C_s] + 1
\end{equation}

And similarly, lines 11-13 for Case-(W) to account for potential transfers. Here, $\text{max\_transfers}[C_i]$ represents the maximum possible votes that could transfer to candidate $C_i$ if all free votes had $C_i$ as their next preference.

By augmenting the required vote calculations with these worst-case transfer scenarios, we ensure the desired structure is preserved regardless of subsequent preferences on added ballots. The optimality is maintained because we add precisely the minimum number of votes required under all possible subsequent preference configurations.

(2) \textit{Robustness to prefixing irrelevant candidate choices:} 
Algorithm~\ref{algo:candidate-removal} identifies consecutive blocks of candidates as irrelevant when they cannot affect the election beyond their internal ordering for up to $B$ ballot additions. Since these irrelevant candidates get eliminated first and transfer votes to strategic candidates below, optimal ballots from Algorithm~\ref{algo: allocation} remain optimal when prepended with arbitrary rankings of irrelevant candidates. 

This prefix robustness also combines with the suffix robust strategies from extension (1), using the same reasoning. The enhanced algorithm (developed from extension (1)) thus outputs optimal ballots that retain optimality when modified with irrelevant candidate prefixes and arbitrary candidate suffixes.

(3) \textit{Length-restricted strategies:} To adapt Algorithm~\ref{algo: allocation} for length-restricted strategies (particularly single-choice ballots), we modify how votes are tracked and reused. Instead of reusing votes from inactive candidates by adding subsequent preferences, we create new ballots when necessary to respect the length restriction.

For the special case of single-choice ballots, we  disable vote reuse, creating a new ballot for each required addition. 
Line 20 in \Cref{algo: allocation} requires length-two ballots to control vote timing for candidate $C_s$. Under single-choice ballot restrictions, $C_s$ would accumulate sufficient first-preference votes to exceed the Droop quota prematurely, winning in the previous round rather than the intended round. Length-two ballots (ranking $C_{last}$ first, $C_s$ second) delay $C_s$'s vote accumulation until after $C_{last}$'s elimination and subsequent vote transfers. When restricted to single-choice ballots, the algorithm would return infeasibility since the temporal vote control is impossible—resolvable only by increasing the Droop quota (raising the winning threshold) or expanding budget $B$.

Optimality within the length constraint follows from the optimality of \Cref{algo: allocation}, as we continue to add the minimum number of ballots required for each round while respecting the length restriction.

All extensions maintain the polynomial-time complexity of the original algorithm while significantly enhancing its practical applicability.
\end{proof}
\ballotexhaust*
\begin{proof}
For candidate $c$, let $g_c$ denote the number of votes in optimal winning strategy from \Cref{algo: robust_allocation} and $r_c$ be the earliest round where $c$'s optimal strategy affects the election outcome. Let $E_r$ denote the number of ballots exhausted by round $r$. Note that candidate $c$ can win iff the number of exhausted ballots available by round $r_c$ is at least the number of strategic ballots required by \Cref{algo: robust_allocation} ($E_{r_c-1}\ge g_c$). 

Appending an exhausted ballot with strategic ballots is equivalent to adding a new strategic ballot. When a ballot exhausts at round $r_e < r_c$, it contains only preferences for candidates eliminated by round $r_e$. Completing this ballot with optimal strategy preferences produces identical electoral effects to adding a new strategic ballot from round $r_c$ onward, since the exhausted portion involves only inactive candidates while the completion portion matches the optimal strategy specification. Since $r_c$ is the earliest affected round, a new ballot from round $r_c$ is equivalent to the same ballot being newly added to the election instance. 

A candidate $c$ is viable through ballot completion if and only if $g_c \leq E_{r_c-1}$.  For multi-round strategies, we check this condition for each activation round separately. If candidate $c$'s strategy requires additions in rounds $r_1, r_2, \ldots$ with gaps $g_1, g_2, \ldots$, then $c$ is viable if $g_{c_i} \leq E_{r_{c_i}-1}$ for all $i$.  The set of viable candidates $S_v = \{c, i : g_{c_i} \leq E_{r_{c_i}-1}\}$ captures all possible winners through ballot completion.

The robustness guarantee in \Cref{thm: robust_subsequent_choice_part1} extends to robust characterization of modeling ballot completion effects. The completion of exhausted ballots is not restricted to only appending the optimal strategy component---preferences may include any irrelevant as well as prior eliminated candidates at any position, and are robust to arbitrary suffix preferences.

\textbf{Multi-winner case:} In multi-winner STV, the Droop quota increases with ballot additions, which can alter the sequence within a structure. Then, the above characterization holds when \Cref{algo: robust_allocation}'s output preserves the sequence within the election structure, i.e., when additions affect which candidates win or are eliminated within rounds without converting win rounds to elimination rounds. That is, the set of viable candidates is restricted to those candidates that have strategies which retain the sequence.
\end{proof}

\augmentremoval*
\begin{proof}
This algorithmic extension works by recognizing that certain candidates may be eliminated even if they
temporarily survive one elimination round. Specifically, we replace the condition at line \ref{line: ifconditions}:
\[ B+ S[C_i] < S_j^i[C_j]< Q  \ \ \forall \ C_i, C_j \in L, \mathcal{C}\setminus L \]
with a call to \textsc{ExtendedRemovalCondition} as defined in Algorithm \ref{algo:extended-condition}. We now rigorously prove the condition.

Let $C_i \in L$ be a candidate for whom the original condition fails: $B + S[C_i] \geq \min_{C_j \in \mathcal{C}\setminus L} S_j^i[C_j]$. This means $C_i$ could potentially receive enough votes to outlast the candidate with the lowest vote count in the upper group, whom we denote as $C_{worst}$.
    
    Our extension verifies two critical conditions:
    
    \begin{enumerate}
        \item First, we check if budget $B$ is sufficient to simultaneously save both $C_i$ and $C_{worst}$ from elimination. Let $C_{second}$ and $C_{third}$ be the candidates with the second and third lowest votes in $\mathcal{C}\setminus L$. If:
        \[2 \cdot S_j^i[C_{third}] - S_j^i[C_{second}] - S_j^i[C_{worst}] > B\]
        then the budget is insufficient to manipulate both eliminations simultaneously.
        
        \item Second, we simulate the elimination of $C_{worst}$ and check if $C_i$ would still be eliminated next, even with $B$ votes added. If:
        \[S^{temp}[C_i] + B < \min_{C_j \in \mathcal{C}\setminus (L \cup \{C_{worst}\})} S_j^i[C_j]\]
        where $S^{temp}$ represents vote counts after $C_{worst}$'s elimination and vote transfers, then $C_i$ is guaranteed to be eliminated next.
    \end{enumerate}
    
    If both conditions are satisfied for all candidates in $L$ that fail the original condition, then the entire lower group $L$ can be safely removed without affecting the final outcome. This extension recognizes that temporary deviations in elimination order may still converge to identical results, allowing for more efficient pruning of the election instance. \Cref{algo:extended-condition} captures this modified removal condition.
    
\end{proof}

\begin{algorithm}
\caption{Extended Candidate Removal Condition}
\label{algo:extended-condition}
\begin{algorithmic}[1]
\REQUIRE  $L$, $U = \mathcal{C}\setminus L$, $\{D\}_I$, $B$, $S$, $S_j^i$
\ENSURE Boolean indicating if the extended removal condition is satisfied

\STATE \textbf{function} \textsc{ExtendedRemovalCondition}($L, U, \{D\}_I, B, S, S_j^i$)
\STATE $orig\_condition \gets \textbf{true if } B+ S[C_i] < S_j^i[C_j] < Q \text{ } \forall C_i \in L, C_j \in U$
\IF{$orig\_condition$}
    \RETURN \textbf{true}
\ENDIF

\FORALL{$C_i \in L$ where $B + S[C_i] \geq \min_{C_j \in U} S_j^i[C_j]$}
    \STATE $C_{worst} \gets \arg\min_{C_j \in U} S_j^i[C_j]$ \COMMENT{Lowest vote candidate in upper group}
    \STATE $U_{temp} \gets U \setminus \{C_{worst}\}$
    \IF{$|U_{temp}| < 2$}
        \RETURN \textbf{false} \COMMENT{Not enough candidates for extension}
    \ENDIF
    
    \STATE $C_{second}, C_{third} \gets$ candidates with two lowest vote counts in $U_{temp}$
    \IF{$2 \cdot S_j^i[C_{third}] - S_j^i[C_{second}] - S_j^i[C_{worst}] \leq B$}
        \RETURN \textbf{false} \COMMENT{Budget can save both $C_i$ and $C_{worst}$}
    \ENDIF
    
    \STATE $\{D\}_I^{temp} \gets \textsc{ReduceElectionInstance}(\{C_{worst}\}, \{D\}_I)$
    \STATE $S^{temp} \gets \textsc{Strict-Support}(\{C_i\}, U_{temp}, \{D\}_I^{temp})$
    \IF{$S^{temp}[C_i] + B \geq \min_{C_j \in U_{temp}} S_j^i[C_j]$}
        \RETURN \textbf{false} \COMMENT{$C_i$ can survive after $C_{worst}$ elimination}
    \ENDIF
\ENDFOR
\RETURN \textbf{true} \COMMENT{Extended condition met for all candidates}
\end{algorithmic}
\end{algorithm}

\candidateremovalextension*
\begin{proof}[Proof]

The position of \(C_w\)'s win in the elimination sequence of \(L\) influences the election in two aspects:
\begin{enumerate}
    \item Immediately after winning, \(C_w\) transfers its \emph{weighted surplus votes} to candidates still in the election, including those in \(L\).
    \item After \(C_w\) is elected, any subsequent unweighted transferred votes arriving at \(C_w\) are shifted to the next specified preferences (for example, from ballots of the form \([L_1, C_w, L_4, \dots]\)).
\end{enumerate}
As given in the proof sketch, our argument is that—even if the new votes from \(B\) shift \(C_w\)'s position within the elimination sequence of \(L\)—\(C_w\)'s influence remains bounded in both aspects, ensuring that the candidates in \(L\) remain irrelevant. We then independently bound the impact of these two aspects, and write an updated strict-support (see \Cref{def: strict_support}) condition that utilizes these bounds (as specified in line \ref{line: ifconditions} of \Cref{algo:candidate-removal}). 

(1) To bound \(C_w\)'s influence in transferring weighted surplus, suppose \(C_w\) wins in the penultimate position of \(L\)'s elimination sequence. In this case, at most all votes of the form \([L, C_w]\) are added to \(C_w\) at that point, which causes \(C_w\) to exceed \(Q\) for the first time. Denote by \(SV_0\) the maximum number of such votes. Formally, 
\[SV_0 = \text{Strict-Support}_{C_w \cup L\setminus C_i, \{\}}(C_w) - Q\]
Let \(V_{C_w}\) be the number of \(C_w\)'s original (first‐choice) votes, WLOG, we assume $V_{C_w} \leq Q$, otherwise $C_w$ would win before elimination rounds start. Then each transfer vote is assigned a weight of at most
\[
\frac{SV_0 + V_{C_w} - Q}{SV_0+V_{C_w}} \leq \frac{SV_0}{SV_0+V_{C_w}}
\]
Since \(Q\) is the maximum number of votes \(C_w\) may have before winning.
Next, if we let \(X\) denote the number of \(C_w\)'s votes for which the next choice is in $C_i$, then any candidate \(C_i \in L\) can receive at most 
\[
SV_0 \times \frac{X}{SV_0+V_{C_w}}
\]
additional weighted votes (surplus) from \(C_w\)'s win. $X$ may be formally expressed as $\sum_{C_i \in L} \tilde{V}_{C_w, C_i}$, where $\{\tilde{V}\}$ denote the ballot counts \emph{after} election instance is reduced with eliminating $L\setminus C_i$. 

(2) Next, consider the maximum number of (unweighted) votes that a candidate in \(L\) can receive via transfers from other candidates in \(L\) when \(C_w\) wins at the earliest possible point. In this case, the total unweighted transfers are bounded by \(SV_1 = \text{Strict-Support}_{C_w \cup L, \{\}}(C_i) - V_{A, C_i}\). Thus, \(SV_1\) is the upper bound on the unweighted transfers that can originate from \(C_w\)'s win and reach $C_i$.

Thus, a candidate $C_i \in L$ can have at most 
\begin{equation}
    SV_0 \times \frac{X}{SV_0+V_{C_w}} + SV_1 \label{eq:extension_bd1}
\end{equation}
votes irrespective of $C_w$'s position of win, considering both aspects.

Given $C_w$ wins during the elimination sequence of $L$, we also update the minimum votes a candidate $C_j \in \mathcal{C}\setminus L$ may have to:
\begin{equation}
    \text{Strict-Support}_{C_j \cup L\setminus C_i, \{\}}(C_j) + \max(0, \text{Strict-Support}_{L\setminus C_i, C_w}(C_j) - [V_{C_w} - Q])
    \label{eq:extension_bd2}
\end{equation}

Here, the first term is exactly the same as line \ref{line: ifconditions} in \Cref{algo:candidate-removal}, and the second term captures the minimum unweighted transfers candidate $C_j$ can receive from $L$, \emph{after} $C_w$ is out-of-contest. This assumes the worst-case scenario of at least $Q - V_{C_{w}}$ unweighted transfers from $L$ reaching $C_w$ prior to its win, i.e., while $C_w$ still in-contest, thereby restricting their subsequent transfer to $C_j$.

Finally, to determine if the addition of $B$ votes retains $L$ irrelevant, we verify the updated strict-support condition (i.e., line \ref{line: ifconditions} in \Cref{algo:candidate-removal}) for each candidate \(C_i \in L\) relative to each candidate \(C_j \in \mathcal{C} \setminus L\), now comparing their votes using Eq. \eqref{eq:extension_bd1} and \eqref{eq:extension_bd2}. Algorithm~\ref{algo:multi-winner-verification} implements this verification process efficiently, directly applying the bounds we have derived to check whether removal of candidates in $L$ is permissible even with a winning candidate $C_w$. This verification can be carried out independently for each case, even if multiple candidates win during the elimination sequence of \(L\). 

In summary, if the updated strict support condition is satisfied, then the additional transfers from \(C_w\)'s win (both weighted and unweighted) remain within the bounds that preserve the irrelevance of candidates in \(L\). Hence, the optimality of \Cref{algo:candidate-removal} in retaining \(\mathcal{C} \setminus L\) is validated.\\
\noindent
 \textbf{Complexity: }  
For computing Eq. \eqref{eq:extension_bd1}, we call \textsc{Strict-Support} twice and \textsc{ReduceElectionInstance} once, hence $O(mn + n) =O(mn)$ operations for $O(m)$ candidates. For Eq. \eqref{eq:extension_bd2}, we call \textsc{Strict-Support} twice, thus $O(mn)$ operations for $O(m)$ candidates. Hence, the complexity of the verification is at most $O(nm^2)$, for $n$ number of unique ballots and $m$ candidates, as implemented in Algorithm~\ref{algo:multi-winner-verification}.
\end{proof}

\begin{algorithm}
\caption{Verification Procedure for Multi-Winner Instances}
\label{algo:multi-winner-verification}
\begin{algorithmic}[1]
\REQUIRE Winner $C_w$, Ballot dictionary $\{D\}_I$, Quota $Q$, Candidates $\mathcal{C}$, Irrelevant set $L$, Budget $B$
\ENSURE Boolean indicating if removal is permitted

\STATE \textbf{function} \textsc{MultiWinnerVerification}($C_w, \{D\}_I, Q, \mathcal{C}, L, B$)
\STATE $V_{C_w} \gets$ First-choice votes for $C_w$ in $\{D\}_I$
\STATE $SV_0 \gets \textsc{Strict-Support}_{C_w \cup L, \{\}}(C_w) - Q$ \COMMENT{Surplus votes from $C_w$}

\FORALL{$C_i \in L$}
    \STATE $\{D\}_I^{temp} \gets \textsc{ReduceElectionInstance}(L \setminus C_i, \{D\}_I)$ \COMMENT{Reduce instance}
    \STATE $X \gets$ Count of ballots where next choice after $C_w$ is $C_i$ in $\{D\}_I^{temp}$
    \STATE $SV_1 \gets \textsc{Strict-Support}_{C_w \cup L, \{\}}(C_i) - $ First-choice votes for $C_i$
    
    \STATE $max\_votes\_C_i \gets SV_0 \times \frac{X}{SV_0 + V_{C_w}} + SV_1$ \COMMENT{Max votes for $C_i$ per Eq.~\eqref{eq:extension_bd1}}
    
    \FORALL{$C_j \in \mathcal{C} \setminus L$}
        \STATE $direct\_votes \gets \textsc{Strict-Support}_{C_j \cup L \setminus C_i, \{\}}(C_j)$
        \STATE $transfers\_after\_win \gets \textsc{Strict-Support}_{L \setminus C_i, C_w}(C_j)$
        \STATE $min\_votes\_C_j \gets direct\_votes + \max(0, transfers\_after\_win - [V_{C_w} - Q])$ \COMMENT{Min votes for $C_j$ per Eq.~\eqref{eq:extension_bd2}}
        
        \IF{$B + max\_votes\_C_i \geq min\_votes\_C_j$}
            \RETURN \textbf{false} \COMMENT{Removal not permitted}
        \ENDIF
    \ENDFOR
\ENDFOR

\RETURN \textbf{true} \COMMENT{Removal permitted}
\end{algorithmic}
\end{algorithm}

\begin{restatable}{lemma}{computationalenhancements} \label{thm: computational_enhancements}
The smart allocation procedure in Algorithm~\ref{algo: robust_allocation} produces effective proposals for pruning sub-optimal and infeasible structures, thereby improving efficiency with memory-sharing and parallelization.
\end{restatable}
\begin{proof}
We exploit Algorithm~\ref{algo: robust_allocation}'s sequential processing to prune the search space efficiently. If budget $B$ is insufficient for the first $k$ rounds of a structure, then any structure sharing those same first $k$ rounds is also infeasible.

Our pruning mechanism caches infeasible prefixes and checks each new structure against this cache before evaluation. This eliminates entire branches without individual evaluation, significantly reducing computation for large structure spaces.

Correctness follows from the algorithm's sequential nature: if required votes for rounds 1 through $k$ exceed budget $B$, then any structure with the same initial pattern is necessarily infeasible. The approach naturally parallelizes since structures can be independently tested for feasibility.
\end{proof}

\subsection{Enhanced Algorithmic Framework}\label{app:enhanced_algo}
We now present the enhanced algorithms discussed in \Cref{sec:computational_framework}. \Cref{algo: robust_allocation} extends \Cref{algo: allocation} with robustness guarantees and length constraints, while \Cref{algo:candidate-removal-updated} improves \Cref{algo:candidate-removal} with strengthened removal conditions for broader applicability. The key changes in the algorithms are highlighted in blue color.
\LineNumbersRobustBlue 
\begin{algorithm}
\renewcommand{\thealgorithm}{A}
\caption{Robust Allocation Rule} \label{algo: robust_allocation}
\begin{algorithmic}[1]
\REQUIRE  $\mathcal{C}$, voter data $\{v\}_{\mathcal{C}}$, Constants $B$, $Q$
\ENSURE Budget allocation robust to subsequent preferences

\STATE \textbf{function} \textsc{RobustAllocation}( $\mathcal{C}$, $\{v\}_{\mathcal{C}}$, $B$, $Q$)
\STATE Initialize $\text{votes}, \text{ votes}_{\text{new}} \leftarrow 0$, $\text{free\_votes} \leftarrow \{\}$

\FOR{processing each round $r$ from first to last}
    \STATE $C_s \leftarrow $ desired candidate for elimination or win in this round
    \STATE $\text{worst\_case\_transfers} \gets$ compute maximum possible transfers from $\text{free\_votes}$ to each candidate
    
    \IF{$C_s$ needs elimination}
        \IF{all in-contest candidates have votes $< Q$}
            \STATE Add $\text{votes}_{\text{new}} \ \forall  \ C_i : \  v^r_{C_i} + \text{votes}_{\text{new}} + \text{worst\_case\_transfers}[C_i] \geq v^r_{C_s} + \text{worst\_case\_transfers}[C_s] + 1$
        \ELSE
            \RETURN infeasible
        \ENDIF
    \ENDIF
    
    \IF{$C_s$ needs to win}
        \IF{all in-contest candidates have votes $< Q$}
            \STATE Add $\text{votes}_{\text{new}}\text{ to } C_s : \  v^r_{C_s} + \text{votes}_{\text{new}} \geq Q + 1 - \text{worst\_case\_transfers}[C_s]$
        \ELSE
            \STATE Add $\text{votes}_{\text{new}} \text{ to } C_s : \  v^r_{C_s} + \text{votes}_{\text{new}} + \text{worst\_case\_transfers}[C_s] \geq \max_i(v^r_{C_i} + \text{worst\_case\_transfers}[C_i]) + 1$
        \ENDIF
        
        \STATE $C_{last} \gets$ \text{Candidate who went out of contest in the previous round}
        \IF{$C_{last}$ doesn't transfer to $C_s$}
            \STATE Add $\text{votes}_{\text{new}}  [C_{last}, C_s] : \  v^r_{C_{last}} \leftarrow v_{C_{last}} + 1$
            \IF{$margin(v^r_{C_{last}}) < 1$}
                \RETURN infeasible
            \ENDIF
        \ENDIF
    \ENDIF
    
    \STATE $\text{votes} \leftarrow \text{votes} + \text{votes}_{\text{new}}$
    \STATE $free\_votes \gets free\_votes \cup \{$votes from eliminated candidates$\}$
    \STATE $\mathcal{C} \gets \mathcal{C} \setminus C_s$, update $\{v\}_\mathcal{C}$
\ENDFOR

\IF{$\text{votes} \leq B$}
    \RETURN feasible \COMMENT{Distribution of added votes}
\ELSE
    \RETURN infeasible \COMMENT{Increase budget $B$}
\ENDIF
\end{algorithmic}
\end{algorithm}
\LineNumbersDefault
\addtocounter{algorithm}{-1} 

\LineNumbersBlueVIItoXIII
\begin{algorithm}
\renewcommand{\thealgorithm}{B}
\caption{Removal of Irrelevant Candidates under Stronger Conditions}
\label{algo:candidate-removal-updated}
\begin{algorithmic}[1]
\REQUIRE  $\mathcal{C}$ in increasing order of $V^1$, Dictionary of orders (ballots) and values (votes for those ballots) $\{D\}_I$, Constants $B$, $Q$
\ENSURE Smallest relevant candidate set $\mathcal{C}$
\STATE \textbf{procedure} \textsc{IrrelevantCandidateRemoval}($\mathcal{C}, \{D\}_I, B, Q$)
\STATE $L, S \gets  \{\}, $  \text{dictionary with keys from set $\mathcal{C}$ and values as zeros}
\WHILE{$B+S[C_i]<Q$ for all {$C_i \in L$} \textbf{and} $\mathcal{C} \neq \phi$}
    \STATE Shift first  candidate $C_1\in \mathcal{C}$ to set $L$
    \STATE $S \gets \textsc{Strict-Support}(L, \mathcal{C}\setminus L, \{D\}_I$)
    \STATE $S_j^i  \gets \textsc{Strict-Support}(C_j\cup L\setminus C_i, \{\}, \  \{D\}_I) \ \forall C_i, C_j \in L, \mathcal{C}\setminus L $
    \IF{\textsc{ExtendedRemovalCondition}($L, U, \{D\}_I, B, S, S_j^i$) = True} \label{line: ifnewconditions}
        \IF{ $\{V^1_{C_i}$ after transfers from $L\}$ $> Q$ for $C_i \in \mathcal{C}$}
            \STATE $C_w \gets C_i$ 
            \IF{ \textsc{MultiWinnerVerification}($C_w, \{D\}_I, Q, \mathcal{C}, L, B$) = False}
            \STATE \textbf{break}
            \ENDIF
        \ENDIF
        \STATE $\{D\}_I \gets \textsc{ReduceElectionInstance}(L, \{D\}_I$)
        \STATE $\mathcal{C}, L \gets \mathcal{C}\setminus L, \{\} $ 
        \STATE Sort $\mathcal{C}$ in increasing order of updated $V^1$
    \ENDIF
\ENDWHILE
\RETURN $\mathcal{C}, \{D\}_I$
\end{algorithmic}
\end{algorithm}
\LineNumbersDefault
\addtocounter{algorithm}{-1} 

\section{Probability Models for Ballot Completion in RCV} \label{app:ballot_exhaustion}

We develop probability models to determine whether trailing candidates could overcome their deficits if exhausted ballots were completed. Our algorithmic framework is essential for these calculations, as it computes the proximity to victory for each candidate—information that would otherwise be computationally prohibitive to obtain. The models apply to both single-winner and multi-winner Single Transferable Vote (STV) systems.\\ 
\underline{Single-Winner Framework:} Let candidate A lead candidate B by gap percentage $g$ (as percentage of total votes), with exhaust percentage $e$ representing ballots that don't rank either A or B in the final round. 
 \\ 
\underline{Multi-Winner Framework:} Let candidate $C$ be eliminated when $k-1$ winners have been selected, competing against $m$ active candidates. The strategy percentage $s$ represents the gap $C$ needs to overcome, and exhaust percentage $e$ represents ballots that don't rank any active candidate at $C$'s elimination. \\
\underline{Required Preference Computation:} For the trailing candidate to win, they need a net gain at least equal to the gap. If proportion $p$ of exhausted ballots prefer the trailing candidate, then net gain is $e(2p-1)$ where $e$ is the number of exhausted ballots. Setting this equal to the required gap and solving yields $r$, the required preference percentage: $r = 50 + \frac{g}{2e} \times 100$ (substituting $s$ for $g$ in multi-winner cases).

\paragraph{Gap-Based Beta Model}
This model uses a Beta distribution with parameters calibrated to the observed Victory Gap between candidates. For a gap percentage $g$ (single-winner) or strategy percentage $s$ (multi-winner), we define $a = \max(50 - g \cdot 0.5, 10)$ and $b = \max(50 + g \cdot 0.5, 10)$. The probability that the trailing candidate wins is $1 - F_{\text{Beta}(a,b)}(r)$, where $r$ is the required preference proportion and $F$ is the CDF of the Beta distribution.

\paragraph{Similarity Beta Model}
This empirical model analyzes non-exhausted ballots that express relevant preferences, categorized by first preference. For each exhausted ballot category (grouped by first preference), we use the corresponding complete ballot data to estimate expected preference distributions. Specifically, we calculate the expected number of completions favoring the trailing candidate versus opponents across all categories, then convert these to percentages. The probability is calculated using $1 - F_{\text{Beta}(\alpha,\beta)}(r)$, where $\alpha$ equals the weighted trailing candidate preference percentage and $\beta$ equals the weighted opponent preference percentage (with $\alpha + \beta = 100$).

\paragraph{Prior-Posterior Beta Model}
This Bayesian model combines the Gap-Based Beta (prior) with observed preference evidence. The prior parameters are gap-based (using the same calculation as Gap-Based Beta). From observed ballot completions, we calculate expected preference percentages for the trailing candidate versus opponents across all first-preference categories. The posterior parameters are weighted combinations: $\alpha_{\text{post}} = \frac{w_1 \alpha_{\text{prior}} + w_2 p_{\text{trailing}}}{w_1 + w_2}$ and $\beta_{\text{post}} = \frac{w_1 \beta_{\text{prior}} + w_2 p_{\text{opponents}}}{w_1 + w_2}$, where $p_{\text{trailing}}$ and $p_{\text{opponents}}$ are the observed preference percentages from the Similarity model, and $w_1 = w_2 = 1.0$.

\paragraph{Similarity Bootstrap}
This bootstrap method categorizes exhausted ballots by first preference, then samples completions based on weighted observed A$>$B vs B$>$A preferences within each category. For each bootstrap iteration, we sample completions for each exhausted ballot category, calculate net votes gained by B, and determine if B wins. The probability is the proportion of iterations where B wins.

\paragraph{Rank-Restricted Bootstrap}
Similar to the Similarity Bootstrap, but respects ranking limits (e.g., NYC's 5-candidate limit) by only completing ballots with fewer than the maximum allowed rankings. This accounts for the institutional constraint that some exhausted ballots could not have been completed due to ranking limits.

\paragraph{Unconditional Bootstrap}
This method samples completions from the overall distribution of preferences between A and B across all ballots, without conditioning on first preference. For each iteration, we sample from a binomial distribution with parameter $p$ equal to the overall proportion of non-exhausted ballots that prefer B over A.

\section{Single-Winner STV: Details}\label{app:single_winner_details}
\subsection{Computational Details}

For each of the 110 elections analyzed, we first applied \Cref{algo:candidate-removal-updated} to reduce the number of relevant candidates to fewer than 10. This reduction yielded a minimum algorithmic traceability threshold of 8.5\%, observed in NYC Council District 26. Elections already having fewer than 10 candidates were analyzed with allowances of 40\% for NYC and 100\% for Alaska. Subsequently, we employed \Cref{algo: robust_allocation} to identify optimal strategies within the allowed vote additions for each relevant candidate. An optimal strategy specifies the minimal ballot additions necessary for a candidate to win, potentially involving strategic support for rival candidates under RCV. With \Cref{thm: computational_enhancements}, the algorithm runtime is approximately 2 days for 10-candidate elections, around 4 hours for 9-candidate elections, and significantly shorter for elections with fewer candidates.

Using these computed strategies, we derived key election attributes. The margin of victory (competitiveness) was determined by the runner-up candidate's victory gap. In three elections, the allowance proved insufficient for any trailing candidate to achieve victory; in these cases, we upper-bounded the victory margin by retaining only the top 10 candidates and identifying the smallest victory gap exceeding the threshold. This remains an upper bound since the eliminated candidates' potential impact on the election dynamics remain unverified.

We analyzed distributions of margins of victory and compared them against corresponding distributions from prior plurality elections in NYC and Alaska. To assess ballot exhaustion impact, we calculated exhausted ballots in each elimination round, comparing these figures to candidates' victory gaps. It suffices to evaluate only relevant candidates within the allowance since irrelevant candidates naturally exhibit victory gaps larger than the accumulated exhaustion levels, unless all relevant candidates simultaneously demonstrate outcome sensitivity to ballot exhaustion—a scenario that does not occur. For candidates whose victory gaps fall below accumulated exhaustion levels, we estimated probabilities of alternative election outcomes using probability models detailed in \Cref{app:ballot_exhaustion}.

The remaining attributes—strategic complexity and preference order alignment—were extracted directly from victory gaps and strategic descriptions for candidates with feasible strategies within allowance limits. Given the substantial number of elections analyzed for NYC (54) and Alaska (52), insights for all four attributes (competitiveness, ballot exhaustion, strategic complexity, and preference order alignment) can be directly drawn from aggregated attributes. In contrast, for Portland's multi-winner RCV elections (4 total), we utilize bootstrap analysis on each of the four elections to rigorously evaluate election attributes in the multi-winner context.

\subsection{Comprehensive Result Tables}
\footnotesize
\begin{longtable}{l | r | c | p{7cm} r | r}

\caption{2021 NYC RCV Democratic Primary: Complete Breakdown of the Strategy-Dynamic
(\textsuperscript{†}=No-Match, \textsuperscript{‡}=Non-Selfish; All others are Match/Selfish; \textsuperscript{*} = Maximal possible margin)}
\label{tab:nyc_elections} \\

\toprule
\textbf{District Name} & \textbf{ Votes } & \textbf{ \# } & \textbf{ Victory Gap Within Allowance} & \textbf{Allow.} & \textbf{Margin} \\
\midrule
\endfirsthead

\toprule
\textbf{District Name} & \textbf{ Votes } & \textbf{ \# } & \textbf{ Candidate Win Margin Within Allowance} & \textbf{Allow.} & \textbf{Margin} \\
\midrule
\endhead

\bottomrule
\multicolumn{6}{r}{\small Continued on next page} \\
\endfoot

\bottomrule
\endlastfoot

Council District 1st \textsuperscript{†} & 21839 & 9 & \colorbox{winnerColor}{\makebox[0.4em][c]{A}}:0\%, \colorbox{contenderColor}{\makebox[0.4em][c]{B}}:17.1\%, \colorbox{competitiveColor}{\makebox[0.4em][c]{C}}:20.1\%, \colorbox{competitiveColor}{\makebox[0.4em][c]{E}}:28.9\%, \colorbox{distantColor}{\makebox[0.4em][c]{D}}:37.5\% & 40.0 & 17.1\% \\
Council District 2nd & 21239 & 2 & \colorbox{winnerColor}{\makebox[0.4em][c]{A}}:0\%, \colorbox{distantColor}{\makebox[0.4em][c]{B}}:32.0\% & 40.0 & 32.0\% \\
Council District 3rd \textsuperscript{†} & 28513 & 6 & \colorbox{winnerColor}{\makebox[0.4em][c]{A}}:0\%, \colorbox{distantColor}{\makebox[0.4em][c]{C}}:34.3\%, \colorbox{distantColor}{\makebox[0.4em][c]{B}}:35.8\% & 40.0 & 34.3\% \\
Council District 5th \textsuperscript{†} \textsuperscript{‡} & 25415 & 7 & \colorbox{winnerColor}{\makebox[0.4em][c]{A}}:0\%, \colorbox{contenderColor}{\makebox[0.4em][c]{B}}:10.2\%, \colorbox{competitiveColor}{\makebox[0.4em][c]{D}}:24.0\%, \colorbox{competitiveColor}{\makebox[0.4em][c]{C}}:25.9\%, \colorbox{distantColor}{\makebox[0.4em][c]{E}}:31.9\% (CE:0.122\%) & 40.0 & 10.2\% \\
Council District 6th \textsuperscript{‡} & 39401 & 6 & \colorbox{winnerColor}{\makebox[0.4em][c]{A}}:0\%, \colorbox{distantColor}{\makebox[0.4em][c]{B}}:43.2\%, \colorbox{distantColor}{\makebox[0.4em][c]{C}}:47.8\%, \colorbox{distantColor}{\makebox[0.4em][c]{D}}:49.3\% (BD:1.228\%), \colorbox{farBehindColor}{\makebox[0.4em][c]{E}}:60.1\%, \colorbox{farBehindColor}{\makebox[0.4em][c]{F}}:65.1\% & 40.0 & 43.2\% \\
Council District 7th & 23214 & 12 & \colorbox{winnerColor}{\makebox[0.4em][c]{A}}:0\% & 18 & 18.62\% \textsuperscript{*}  \\
Council District 8th & 11692 & 4 & \colorbox{winnerColor}{\makebox[0.4em][c]{A}}:0\%, \colorbox{competitiveColor}{\makebox[0.4em][c]{B}}:28.0\% & 40.0 & 28.0\% \\
Council District 9th & 25679 & 13 & \colorbox{winnerColor}{\makebox[0.4em][c]{A}}:0\%, \colorbox{nearWinColor}{\makebox[0.4em][c]{B}}:0.3\%, \colorbox{contenderColor}{\makebox[0.4em][c]{C}}:5.8\% & 9.5 & 0.3\% \\
Council District 10th & 19601 & 8 & \colorbox{winnerColor}{\makebox[0.4em][c]{A}}:0\%, \colorbox{contenderColor}{\makebox[0.4em][c]{B}}:17.2\%, \colorbox{competitiveColor}{\makebox[0.4em][c]{C}}:26.4\% & 40.0 & 17.2\% \\
Council District 11th & 17041 & 7 & \colorbox{winnerColor}{\makebox[0.4em][c]{A}}:0\%, \colorbox{competitiveColor}{\makebox[0.4em][c]{B}}:20.5\%, \colorbox{competitiveColor}{\makebox[0.4em][c]{C}}:27.1\% & 40.0 & 20.5\% \\
Council District 12th & 17341 & 3 & \colorbox{winnerColor}{\makebox[0.4em][c]{A}}:0\%, \colorbox{contenderColor}{\makebox[0.4em][c]{B}}:17.0\%, \colorbox{competitiveColor}{\makebox[0.4em][c]{C}}:24.3\% & 40.0 & 17.0\% \\
Council District 13th & 9961 & 5 & \colorbox{winnerColor}{\makebox[0.4em][c]{A}}:0\%, \colorbox{distantColor}{\makebox[0.4em][c]{B}}:32.9\% & 40.0 & 32.9\% \\
Council District 14th & 9446 & 6 & \colorbox{winnerColor}{\makebox[0.4em][c]{A}}:0\%, \colorbox{competitiveColor}{\makebox[0.4em][c]{B}}:20.5\%, \colorbox{competitiveColor}{\makebox[0.4em][c]{C}}:21.4\%, \colorbox{distantColor}{\makebox[0.4em][c]{D}}:38.3\%, \colorbox{distantColor}{\makebox[0.4em][c]{E}}:39.1\% & 40.0 & 20.5\% \\
Council District 15th \textsuperscript{†} & 8239 & 8 & \colorbox{winnerColor}{\makebox[0.4em][c]{A}}:0\%, \colorbox{competitiveColor}{\makebox[0.4em][c]{B}}:24.6\%, \colorbox{distantColor}{\makebox[0.4em][c]{D}}:33.0\%, \colorbox{distantColor}{\makebox[0.4em][c]{C}}:36.1\% & 40.0 & 24.6\% \\
Council District 16th & 10119 & 4 & \colorbox{winnerColor}{\makebox[0.4em][c]{A}}:0\%, \colorbox{competitiveColor}{\makebox[0.4em][c]{B}}:28.0\%, \colorbox{competitiveColor}{\makebox[0.4em][c]{C}}:29.7\% & 40.0 & 28.0\% \\
Council District 17th & 9022 & 2 & \colorbox{winnerColor}{\makebox[0.4em][c]{A}}:0\%, \colorbox{contenderColor}{\makebox[0.4em][c]{B}}:13.8\% & 40.0 & 13.8\% \\
Council District 18th \textsuperscript{‡} & 14031 & 8 & \colorbox{winnerColor}{\makebox[0.4em][c]{A}}:0\%, \colorbox{nearWinColor}{\makebox[0.4em][c]{B}}:3.9\%, \colorbox{competitiveColor}{\makebox[0.4em][c]{C}}:26.8\%, \colorbox{competitiveColor}{\makebox[0.4em][c]{D}}:27.4\% (B:2.437\%), \colorbox{distantColor}{\makebox[0.4em][c]{E}}:35.6\% & 40.0 & 3.9\% \\
Council District 19th & 13077 & 6 & \colorbox{winnerColor}{\makebox[0.4em][c]{A}}:0\%, \colorbox{contenderColor}{\makebox[0.4em][c]{B}}:8.4\%, \colorbox{competitiveColor}{\makebox[0.4em][c]{C}}:21.8\%, \colorbox{distantColor}{\makebox[0.4em][c]{D}}:39.0\% & 40.0 & 8.4\% \\
Council District 20th \textsuperscript{†} \textsuperscript{‡} & 10644 & 8 & \colorbox{winnerColor}{\makebox[0.4em][c]{A}}:0\%, \colorbox{contenderColor}{\makebox[0.4em][c]{B}}:7.5\%, \colorbox{contenderColor}{\makebox[0.4em][c]{C}}:16.4\%, \colorbox{competitiveColor}{\makebox[0.4em][c]{D}}:21.0\%, \colorbox{competitiveColor}{\makebox[0.4em][c]{F}}:24.0\%, \colorbox{competitiveColor}{\makebox[0.4em][c]{G}}:28.4\% (CG:1.569\%), \colorbox{competitiveColor}{\makebox[0.4em][c]{E}}:29.2\% & 40.0 & 7.5\% \\
Council District 21st & 6839 & 5 & \colorbox{winnerColor}{\makebox[0.4em][c]{A}}:0\%, \colorbox{distantColor}{\makebox[0.4em][c]{B}}:39.6\%, \colorbox{distantColor}{\makebox[0.4em][c]{C}}:39.6\% & 40.0 & 39.6\% \\
Council District 22nd & 16241 & 6 & \colorbox{winnerColor}{\makebox[0.4em][c]{A}}:0\%, \colorbox{competitiveColor}{\makebox[0.4em][c]{B}}:22.6\% & 40.0 & 22.6\% \\
Council District 23rd \textsuperscript{†} \textsuperscript{‡} & 16532 & 7 & \colorbox{winnerColor}{\makebox[0.4em][c]{A}}:0\%, \colorbox{contenderColor}{\makebox[0.4em][c]{B}}:7.1\%, \colorbox{competitiveColor}{\makebox[0.4em][c]{D}}:20.3\% (B:2.202\%), \colorbox{competitiveColor}{\makebox[0.4em][c]{C}}:22.9\% (B:3.442\%), \colorbox{competitiveColor}{\makebox[0.4em][c]{E}}:29.2\%, \colorbox{distantColor}{\makebox[0.4em][c]{G}}:38.3\% & 40.0 & 7.1\% \\
Council District 24th & 13408 & 4 & \colorbox{winnerColor}{\makebox[0.4em][c]{A}}:0\%, \colorbox{distantColor}{\makebox[0.4em][c]{B}}:30.5\% & 40.0 & 30.5\% \\
Council District 25th \textsuperscript{‡} & 14927 & 8 & \colorbox{winnerColor}{\makebox[0.4em][c]{A}}:0\%, \colorbox{contenderColor}{\makebox[0.4em][c]{B}}:5.4\%, \colorbox{contenderColor}{\makebox[0.4em][c]{C}}:11.8\%, \colorbox{competitiveColor}{\makebox[0.4em][c]{D}}:24.0\%, \colorbox{competitiveColor}{\makebox[0.4em][c]{E}}:24.5\%, \colorbox{competitiveColor}{\makebox[0.4em][c]{F}}:29.2\% (DF:0.154\%), \colorbox{distantColor}{\makebox[0.4em][c]{G}}:33.6\% (DG:0.435\%) & 40.0 & 5.4\% \\
Council District 26th & 17925 & 15 & \colorbox{winnerColor}{\makebox[0.4em][c]{A}}:0\% & 8.5 & 9.02\% \textsuperscript{*} \\
Council District 27th & 20420 & 12 & \colorbox{winnerColor}{\makebox[0.4em][c]{A}}:0\% & 35.0 & 36.42\% \textsuperscript{*} \\
Council District 28th \textsuperscript{†} & 14026 & 3 & \colorbox{winnerColor}{\makebox[0.4em][c]{A}}:0\%, \colorbox{distantColor}{\makebox[0.4em][c]{C}}:34.2\%, \colorbox{distantColor}{\makebox[0.4em][c]{B}}:39.6\% & 40.0 & 34.2\% \\
Council District 29th \textsuperscript{†} \textsuperscript{‡} & 17131 & 9 & \colorbox{winnerColor}{\makebox[0.4em][c]{A}}:0\%, \colorbox{contenderColor}{\makebox[0.4em][c]{B}}:14.1\%, \colorbox{contenderColor}{\makebox[0.4em][c]{C}}:16.4\%, \colorbox{competitiveColor}{\makebox[0.4em][c]{D}}:24.2\%, \colorbox{competitiveColor}{\makebox[0.4em][c]{E}}:25.4\%, \colorbox{competitiveColor}{\makebox[0.4em][c]{F}}:27.0\% (DF:0.280\%), \colorbox{competitiveColor}{\makebox[0.4em][c]{G}}:29.7\%, \colorbox{distantColor}{\makebox[0.4em][c]{H}}:38.9\% & 40.0 & 14.1\% \\
Council District 30th & 9587 & 2 & \colorbox{winnerColor}{\makebox[0.4em][c]{A}}:0\%, \colorbox{contenderColor}{\makebox[0.4em][c]{B}}:6.4\% & 40.0 & 6.4\% \\
Council District 31st & 15991 & 3 & \colorbox{winnerColor}{\makebox[0.4em][c]{A}}:0\%, \colorbox{distantColor}{\makebox[0.4em][c]{B}}:48.4\%, \colorbox{farBehindColor}{\makebox[0.4em][c]{C}}:55.6\% & 40.0 & 48.4\% \\
Council District 32nd \textsuperscript{†} \textsuperscript{‡} & 10221 & 6 & \colorbox{winnerColor}{\makebox[0.4em][c]{A}}:0\%, \colorbox{nearWinColor}{\makebox[0.4em][c]{B}}:4.2\%, \colorbox{distantColor}{\makebox[0.4em][c]{C}}:33.2\% (B:1.556\%), \colorbox{distantColor}{\makebox[0.4em][c]{E}}:36.8\%, \colorbox{distantColor}{\makebox[0.4em][c]{D}}:38.2\% & 40.0 & 4.2\% \\
Council District 33rd & 29178 & 8 & \colorbox{winnerColor}{\makebox[0.4em][c]{A}}:0\%, \colorbox{competitiveColor}{\makebox[0.4em][c]{B}}:24.7\% & 40.0 & 24.7\% \\
Council District 34th & 16395 & 4 & \colorbox{winnerColor}{\makebox[0.4em][c]{A}}:0\%, \colorbox{farBehindColor}{\makebox[0.4em][c]{B}}:73.9\%, \colorbox{farBehindColor}{\makebox[0.4em][c]{C}}:75.4\%, \colorbox{farBehindColor}{\makebox[0.4em][c]{D}}:81.4\% & 40.0 & 73.9\% \\
Council District 35th \textsuperscript{‡} & 34914 & 7 & \colorbox{winnerColor}{\makebox[0.4em][c]{A}}:0\%, \colorbox{contenderColor}{\makebox[0.4em][c]{B}}:6.9\%, \colorbox{distantColor}{\makebox[0.4em][c]{C}}:35.6\% (B:5.568\%) & 40.0 & 6.9\% \\
Council District 36th \textsuperscript{†} & 23244 & 5 & \colorbox{winnerColor}{\makebox[0.4em][c]{A}}:0\%, \colorbox{contenderColor}{\makebox[0.4em][c]{C}}:7.4\%, \colorbox{contenderColor}{\makebox[0.4em][c]{B}}:11.8\%, \colorbox{competitiveColor}{\makebox[0.4em][c]{D}}:23.1\% & 40.0 & 7.4\% \\
Council District 37th & 10882 & 6 & \colorbox{winnerColor}{\makebox[0.4em][c]{A}}:0\%, \colorbox{competitiveColor}{\makebox[0.4em][c]{B}}:26.4\% & 40.0 & 26.4\% \\
Council District 38th \textsuperscript{†} & 12113 & 6 & \colorbox{winnerColor}{\makebox[0.4em][c]{A}}:0\%, \colorbox{competitiveColor}{\makebox[0.4em][c]{B}}:26.2\%, \colorbox{distantColor}{\makebox[0.4em][c]{C}}:34.9\%, \colorbox{distantColor}{\makebox[0.4em][c]{E}}:36.4\%, \colorbox{distantColor}{\makebox[0.4em][c]{D}}:36.9\% & 40.0 & 26.2\% \\
Council District 39th & 36058 & 7 & \colorbox{winnerColor}{\makebox[0.4em][c]{A}}:0\%, \colorbox{contenderColor}{\makebox[0.4em][c]{B}}:10.9\%, \colorbox{contenderColor}{\makebox[0.4em][c]{C}}:19.4\%, \colorbox{competitiveColor}{\makebox[0.4em][c]{D}}:29.7\%, \colorbox{distantColor}{\makebox[0.4em][c]{E}}:39.8\% & 40.0 & 10.9\% \\
Council District 40th \textsuperscript{†} & 22361 & 11 & \colorbox{winnerColor}{\makebox[0.4em][c]{A}}:0\%, \colorbox{contenderColor}{\makebox[0.4em][c]{C}}:13.4\% & 14.0 & 13.4\% \\
Council District 41st & 15403 & 2 & \colorbox{winnerColor}{\makebox[0.4em][c]{A}}:0\%, \colorbox{nearWinColor}{\makebox[0.4em][c]{B}}:4.8\% & 40.0 & 4.8\% \\
Council District 42nd & 15626 & 4 & \colorbox{winnerColor}{\makebox[0.4em][c]{A}}:0\%, \colorbox{contenderColor}{\makebox[0.4em][c]{B}}:7.1\% & 40.0 & 7.1\% \\
Council District 45th & 19270 & 3 & \colorbox{winnerColor}{\makebox[0.4em][c]{A}}:0\%, \colorbox{farBehindColor}{\makebox[0.4em][c]{B}}:55.1\%, \colorbox{farBehindColor}{\makebox[0.4em][c]{C}}:74.5\% & 40.0 & 55.1\% \\
Council District 46th \textsuperscript{†} & 18477 & 8 & \colorbox{winnerColor}{\makebox[0.4em][c]{A}}:0\%, \colorbox{competitiveColor}{\makebox[0.4em][c]{B}}:20.8\%, \colorbox{competitiveColor}{\makebox[0.4em][c]{C}}:27.5\%, \colorbox{distantColor}{\makebox[0.4em][c]{E}}:31.7\%, \colorbox{distantColor}{\makebox[0.4em][c]{D}}:36.5\% & 40.0 & 20.8\% \\
Council District 47th & 8022 & 4 & \colorbox{winnerColor}{\makebox[0.4em][c]{A}}:0\%, \colorbox{contenderColor}{\makebox[0.4em][c]{B}}:10.3\%, \colorbox{distantColor}{\makebox[0.4em][c]{C}}:36.9\% & 40.0 & 10.3\% \\
Council District 48th \textsuperscript{‡} & 9198 & 5 & \colorbox{winnerColor}{\makebox[0.4em][c]{A}}:0\%, \colorbox{contenderColor}{\makebox[0.4em][c]{B}}:12.2\%, \colorbox{competitiveColor}{\makebox[0.4em][c]{C}}:23.4\%, \colorbox{competitiveColor}{\makebox[0.4em][c]{D}}:25.9\% (B:5.153\%) & 40.0 & 12.2\% \\
Council District 49th & 13744 & 9 & \colorbox{winnerColor}{\makebox[0.4em][c]{A}}:0\%, \colorbox{contenderColor}{\makebox[0.4em][c]{B}}:10.7\%, \colorbox{competitiveColor}{\makebox[0.4em][c]{C}}:28.0\%, \colorbox{competitiveColor}{\makebox[0.4em][c]{D}}:28.3\%, \colorbox{competitiveColor}{\makebox[0.4em][c]{E}}:29.5\%, \colorbox{distantColor}{\makebox[0.4em][c]{F}}:38.3\%, \colorbox{distantColor}{\makebox[0.4em][c]{H}}:39.8\% & 40.0 & 10.7\% \\
Mayor Citywide & 944261 & 13 & \colorbox{winnerColor}{\makebox[0.4em][c]{A}}:0\%, \colorbox{nearWinColor}{\makebox[0.4em][c]{B}}:0.8\%, \colorbox{contenderColor}{\makebox[0.4em][c]{C}}:8.2\% & 14.3 & 0.8\% \\
Comptroller \textsuperscript{†} & 868572 & 10 & \colorbox{winnerColor}{\makebox[0.4em][c]{A}}:0\%, \colorbox{nearWinColor}{\makebox[0.4em][c]{B}}:2.9\%, \colorbox{contenderColor}{\makebox[0.4em][c]{C}}:17.0\%, \colorbox{competitiveColor}{\makebox[0.4em][c]{D}}:24.9\%, \colorbox{competitiveColor}{\makebox[0.4em][c]{F}}:28.3\%, \colorbox{distantColor}{\makebox[0.4em][c]{E}}:31.0\%, \colorbox{distantColor}{\makebox[0.4em][c]{G}}:31.2\%, \colorbox{distantColor}{\makebox[0.4em][c]{H}}:36.1\%, \colorbox{distantColor}{\makebox[0.4em][c]{I}}:38.5\%, \colorbox{distantColor}{\makebox[0.4em][c]{J}}:39.3\% & 40.0 & 2.9\% \\
Public Advocate & 814832 & 3 & \colorbox{winnerColor}{\makebox[0.4em][c]{A}}:0\%, \colorbox{distantColor}{\makebox[0.4em][c]{B}}:48.8\%, \colorbox{farBehindColor}{\makebox[0.4em][c]{C}}:62.6\% & 40.0 & 48.8\% \\
Bronx President & 101566 & 5 & \colorbox{winnerColor}{\makebox[0.4em][c]{A}}:0\%, \colorbox{contenderColor}{\makebox[0.4em][c]{B}}:6.3\%, \colorbox{competitiveColor}{\makebox[0.4em][c]{C}}:22.2\%, \colorbox{distantColor}{\makebox[0.4em][c]{D}}:30.5\% & 40.0 & 6.3\% \\
Kings President & 289426 & 12 & \colorbox{winnerColor}{\makebox[0.4em][c]{A}}:0\%, \colorbox{contenderColor}{\makebox[0.4em][c]{B}}:6.6\%, \colorbox{contenderColor}{\makebox[0.4em][c]{C}}:12.1\% & 13.5 & 6.6\% \\
NewYork President \textsuperscript{†} \textsuperscript{‡} & 237604 & 7 & \colorbox{winnerColor}{\makebox[0.4em][c]{A}}:0\%, \colorbox{contenderColor}{\makebox[0.4em][c]{B}}:5.9\%, \colorbox{competitiveColor}{\makebox[0.4em][c]{D}}:24.4\%, \colorbox{competitiveColor}{\makebox[0.4em][c]{C}}:26.2\% (B:3.611\%), \colorbox{competitiveColor}{\makebox[0.4em][c]{E}}:27.0\% (CE:1.175\%), \colorbox{distantColor}{\makebox[0.4em][c]{G}}:38.1\% (CG:1.557\%) & 40.0 & 5.9\% \\
Queens President & 195336 & 3 & \colorbox{winnerColor}{\makebox[0.4em][c]{A}}:0\%, \colorbox{nearWinColor}{\makebox[0.4em][c]{B}}:0.5\%, \colorbox{distantColor}{\makebox[0.4em][c]{C}}:30.8\% & 40.0 & 0.5\% \\
Richmond President & 28264 & 5 & \colorbox{winnerColor}{\makebox[0.4em][c]{A}}:0\%, \colorbox{competitiveColor}{\makebox[0.4em][c]{B}}:26.7\%, \colorbox{distantColor}{\makebox[0.4em][c]{C}}:36.7\% & 40.0 & 26.7\% \\
\end{longtable}

\begin{longtable}{l | r | c | p{8cm} r | r}

\caption{\small 2024 Alaska State-wide RCV: Complete Breakdown of the Strategy-Dynamic}
\label{tab:alaska_elections} \\

\toprule
\textbf{District Name} & \textbf{Votes} & \textbf{\#} &
\textbf{Victory Gap Within Allowance} & \textbf{Allow.} &
\textbf{Margin} \\
\midrule
\endfirsthead

\toprule
\textbf{District Name} & \textbf{Votes} & \textbf{\#} &
\textbf{Candidate Win Margin Within Allowance} & \textbf{Allow.} &
\textbf{Margin} \\
\midrule
\endhead

\bottomrule
\multicolumn{6}{r}{\small Continued on next page} \\
\endfoot

\bottomrule
\endlastfoot

House of Rep. & 189389 & 4 & \colorbox{winnerColor}{\makebox[0.5em][c]{A}}: 0.0\%, \colorbox{nearWinColor}{\makebox[0.5em][c]{B}}: 2.7\%, \colorbox{nearWinColor}{\makebox[0.5em][c]{C}}: 2.7\% & 40 & 2.7\% \\
State House D1 & 8177 & 4 & \colorbox{winnerColor}{\makebox[0.5em][c]{A}}: 0.0\%, \colorbox{competitiveColor}{\makebox[0.5em][c]{B}}: 20.7\%, \colorbox{competitiveColor}{\makebox[0.5em][c]{C}}: 23.8\% & 40 & 20.7\% \\
State House D2 & 7362 & 2 & \colorbox{winnerColor}{\makebox[0.5em][c]{A}}: 0.0\%, \colorbox{farBehindColor}{\makebox[0.5em][c]{B}}: 94.5\% & 40 & 94.5\% \\
State House D3 & 8415 & 2 & \colorbox{winnerColor}{\makebox[0.5em][c]{A}}: 0.0\%, \colorbox{farBehindColor}{\makebox[0.5em][c]{B}}: 92.0\% & 40 & 92.0\% \\
State House D4 & 7368 & 2 & \colorbox{winnerColor}{\makebox[0.5em][c]{A}}: 0.0\%, \colorbox{farBehindColor}{\makebox[0.5em][c]{B}}: 93.0\% & 40 & 93.0\% \\
State House D5 & 7059 & 3 & \colorbox{winnerColor}{\makebox[0.5em][c]{A}}: 0.0\%, \colorbox{farBehindColor}{\makebox[0.5em][c]{B}}: 55.2\%, \colorbox{farBehindColor}{\makebox[0.5em][c]{C}}: 82.2\% & 40 & 55.2\% \\
State House D6 & 11431 & 4 & \colorbox{winnerColor}{\makebox[0.5em][c]{A}}: 0.0\%, \colorbox{nearWinColor}{\makebox[0.5em][c]{B}}: 4.3\%, \colorbox{distantColor}{\makebox[0.5em][c]{C}}: 34.0\% & 40 & 4.3\% \\
State House D7 & 8495 & 3 & \colorbox{winnerColor}{\makebox[0.5em][c]{A}}: 0.0\%, \colorbox{contenderColor}{\makebox[0.5em][c]{B}}: 18.8\%, \colorbox{farBehindColor}{\makebox[0.5em][c]{C}}: 68.1\% & 40 & 18.8\% \\
State House D8 & 9053 & 3 & \colorbox{winnerColor}{\makebox[0.5em][c]{A}}: 0.0\%, \colorbox{nearWinColor}{\makebox[0.5em][c]{B}}: 4.4\%, \colorbox{farBehindColor}{\makebox[0.5em][c]{C}}: 62.1\% & 40 & 4.4\% \\
State House D9 & 11213 & 3 & \colorbox{winnerColor}{\makebox[0.5em][c]{A}}: 0.0\%, \colorbox{contenderColor}{\makebox[0.5em][c]{B}}: 8.9\%, \colorbox{farBehindColor}{\makebox[0.5em][c]{C}}: 59.3\% & 40 & 8.9\% \\
State House D10 & 7604 & 3 & \colorbox{winnerColor}{\makebox[0.5em][c]{A}}: 0.0\%, \colorbox{competitiveColor}{\makebox[0.5em][c]{B}}: 24.5\%, \colorbox{farBehindColor}{\makebox[0.5em][c]{C}}: 69.9\% & 40 & 24.5\% \\
State House D11 & 9207 & 3 & \colorbox{winnerColor}{\makebox[0.5em][c]{A}}: 0.0\%, \colorbox{contenderColor}{\makebox[0.5em][c]{B}}: 5.3\%, \colorbox{farBehindColor}{\makebox[0.5em][c]{C}}: 61.1\% & 40 & 5.3\% \\
State House D12 & 7848 & 3 & \colorbox{winnerColor}{\makebox[0.5em][c]{A}}: 0.0\%, \colorbox{competitiveColor}{\makebox[0.5em][c]{B}}: 22.0\%, \colorbox{farBehindColor}{\makebox[0.5em][c]{C}}: 67.0\% & 40 & 22.0\% \\
State House D13 & 7031 & 3 & \colorbox{winnerColor}{\makebox[0.5em][c]{A}}: 0.0\%, \colorbox{contenderColor}{\makebox[0.5em][c]{B}}: 6.9\%, \colorbox{farBehindColor}{\makebox[0.5em][c]{C}}: 60.2\% & 40 & 6.9\% \\
State House D14 & 6284 & 3 & \colorbox{winnerColor}{\makebox[0.5em][c]{A}}: 0.0\%, \colorbox{farBehindColor}{\makebox[0.5em][c]{B}}: 56.6\%, \colorbox{farBehindColor}{\makebox[0.5em][c]{C}}: 82.7\% & 40 & 56.6\% \\
State House D15 & 8832 & 4 & \colorbox{winnerColor}{\makebox[0.5em][c]{A}}: 0.0\%, \colorbox{nearWinColor}{\makebox[0.5em][c]{B}}: 4.6\%, \colorbox{distantColor}{\makebox[0.5em][c]{C}}: 42.8\%, \colorbox{farBehindColor}{\makebox[0.5em][c]{D}}: 62.2\% & 40 & 4.6\% \\
State House D16 & 8611 & 3 & \colorbox{winnerColor}{\makebox[0.5em][c]{A}}: 0.0\%, \colorbox{contenderColor}{\makebox[0.5em][c]{B}}: 14.0\%, \colorbox{farBehindColor}{\makebox[0.5em][c]{C}}: 63.3\% & 40 & 14.0\% \\
State House D17 & 5744 & 2 & \colorbox{winnerColor}{\makebox[0.5em][c]{A}}: 0.0\%, \colorbox{farBehindColor}{\makebox[0.5em][c]{B}}: 87.6\% & 40 & 87.6\% \\
State House D18 & 3745 & 3 & \colorbox{winnerColor}{\makebox[0.5em][c]{A}}: 0.0\%, \colorbox{nearWinColor}{\makebox[0.5em][c]{B}}: 0.5\%, \colorbox{farBehindColor}{\makebox[0.5em][c]{C}}: 57.8\% & 40 & 0.5\% \\
State House D19 & 4390 & 4 & \colorbox{winnerColor}{\makebox[0.5em][c]{A}}: 0.0\%, \colorbox{distantColor}{\makebox[0.5em][c]{B}}: 30.5\%, \colorbox{distantColor}{\makebox[0.5em][c]{C}}: 49.7\%, \colorbox{farBehindColor}{\makebox[0.5em][c]{D}}: 67.3\% & 40 & 30.5\% \\
State House D20 & 6182 & 3 & \colorbox{winnerColor}{\makebox[0.5em][c]{A}}: 0.0\%, \colorbox{competitiveColor}{\makebox[0.5em][c]{B}}: 28.5\%, \colorbox{farBehindColor}{\makebox[0.5em][c]{C}}: 69.2\% & 40 & 28.5\% \\
State House D21 & 8452 & 3 & \colorbox{winnerColor}{\makebox[0.5em][c]{A}}: 0.0\%, \colorbox{contenderColor}{\makebox[0.5em][c]{B}}: 10.8\%, \colorbox{farBehindColor}{\makebox[0.5em][c]{C}}: 60.0\% & 40 & 10.8\% \\
State House D22 & 5049 & 3 & \colorbox{winnerColor}{\makebox[0.5em][c]{A}}: 0.0\%, \colorbox{nearWinColor}{\makebox[0.5em][c]{B}}: 4.9\%, \colorbox{farBehindColor}{\makebox[0.5em][c]{C}}: 60.9\% & 40 & 4.9\% \\
State House D23 & 9928 & 3 & \colorbox{winnerColor}{\makebox[0.5em][c]{A}}: 0.0\%, \colorbox{competitiveColor}{\makebox[0.5em][c]{B}}: 24.0\%, \colorbox{farBehindColor}{\makebox[0.5em][c]{C}}: 67.9\% & 40 & 24.0\% \\
State House D24 & 8317 & 2 & \colorbox{winnerColor}{\makebox[0.5em][c]{A}}: 0.0\%, \colorbox{farBehindColor}{\makebox[0.5em][c]{B}}: 93.0\% & 40 & 93.0\% \\
State House D25 & 8269 & 2 & \colorbox{winnerColor}{\makebox[0.5em][c]{A}}: 0.0\%, \colorbox{farBehindColor}{\makebox[0.5em][c]{B}}: 91.5\% & 40 & 91.5\% \\
State House D26 & 7686 & 2 & \colorbox{winnerColor}{\makebox[0.5em][c]{A}}: 0.0\%, \colorbox{farBehindColor}{\makebox[0.5em][c]{B}}: 92.2\% & 40 & 92.2\% \\
State House D27 & 7636 & 3 & \colorbox{winnerColor}{\makebox[0.5em][c]{A}}: 0.0\%, \colorbox{nearWinColor}{\makebox[0.5em][c]{B}}: 2.5\%, \colorbox{farBehindColor}{\makebox[0.5em][c]{C}}: 61.4\% & 40 & 2.5\% \\
State House D28 & 7827 & 4 & \colorbox{winnerColor}{\makebox[0.5em][c]{A}}: 0.0\%, \colorbox{nearWinColor}{\makebox[0.5em][c]{B}}: 0.1\%, \colorbox{contenderColor}{\makebox[0.5em][c]{C}}: 10.7\%, \colorbox{farBehindColor}{\makebox[0.5em][c]{D}}: 52.3\% & 40 & 0.1\% \\
State House D29 & 8252 & 2 & \colorbox{winnerColor}{\makebox[0.5em][c]{A}}: 0.0\%, \colorbox{farBehindColor}{\makebox[0.5em][c]{B}}: 90.8\% & 40 & 90.8\% \\
State House D30 & 8755 & 3 & \colorbox{winnerColor}{\makebox[0.5em][c]{A}}: 0.0\%, \colorbox{contenderColor}{\makebox[0.5em][c]{B}}: 9.9\%, \colorbox{farBehindColor}{\makebox[0.5em][c]{C}}: 62.0\% & 40 & 9.9\% \\
State House D31 & 6501 & 3 & \colorbox{winnerColor}{\makebox[0.5em][c]{A}}: 0.0\%, \colorbox{contenderColor}{\makebox[0.5em][c]{B}}: 9.0\%, \colorbox{farBehindColor}{\makebox[0.5em][c]{C}}: 59.9\% & 40 & 9.0\% \\
State House D32 & 4806 & 3 & \colorbox{winnerColor}{\makebox[0.5em][c]{A}}: 0.0\%, \colorbox{distantColor}{\makebox[0.5em][c]{B}}: 35.0\%, \colorbox{farBehindColor}{\makebox[0.5em][c]{C}}: 74.6\% & 40 & 35.0\% \\
State House D34 & 8750 & 3 & \colorbox{winnerColor}{\makebox[0.5em][c]{A}}: 0.0\%, \colorbox{contenderColor}{\makebox[0.5em][c]{B}}: 12.8\%, \colorbox{farBehindColor}{\makebox[0.5em][c]{C}}: 64.9\% & 40 & 12.8\% \\
State House D35 & 9146 & 3 & \colorbox{winnerColor}{\makebox[0.5em][c]{A}}: 0.0\%, \colorbox{contenderColor}{\makebox[0.5em][c]{B}}: 10.6\%, \colorbox{farBehindColor}{\makebox[0.5em][c]{C}}: 59.8\% & 40 & 10.6\% \\
State House D36 & 8961 & 5 & \colorbox{winnerColor}{\makebox[0.5em][c]{A}}: 0.0\%, \colorbox{contenderColor}{\makebox[0.5em][c]{B}}: 10.2\%, \colorbox{contenderColor}{\makebox[0.5em][c]{C}}: 11.3\%, \colorbox{distantColor}{\makebox[0.5em][c]{D}}: 33.6\%, \colorbox{farBehindColor}{\makebox[0.5em][c]{E}}: 53.2\% & 40 & 10.2\% \\
State House D37 & 3827 & 3 & \colorbox{winnerColor}{\makebox[0.5em][c]{A}}: 0.0\%, \colorbox{distantColor}{\makebox[0.5em][c]{B}}: 45.9\%, \colorbox{farBehindColor}{\makebox[0.5em][c]{C}}: 80.0\% & 40 & 45.9\% \\
State House D38 & 3614 & 5 & \colorbox{winnerColor}{\makebox[0.5em][c]{A}}: 0.0\%, \colorbox{nearWinColor}{\makebox[0.5em][c]{B}}: 3.6\%, \colorbox{contenderColor}{\makebox[0.5em][c]{C}}: 18.0\%, \colorbox{distantColor}{\makebox[0.5em][c]{D}}: 41.6\%, \colorbox{farBehindColor}{\makebox[0.5em][c]{E}}: 62.5\% & 40 & 3.6\% \\
State House D39 & 4050 & 3 & \colorbox{winnerColor}{\makebox[0.5em][c]{A}}: 0.0\%, \colorbox{contenderColor}{\makebox[0.5em][c]{B}}: 16.7\%, \colorbox{farBehindColor}{\makebox[0.5em][c]{C}}: 75.1\% & 40 & 16.7\% \\
State House D40 & 3071 & 4 & \colorbox{winnerColor}{\makebox[0.5em][c]{A}}: 0.0\%, \colorbox{contenderColor}{\makebox[0.5em][c]{B}}: 17.5\%, \colorbox{competitiveColor}{\makebox[0.5em][c]{C}}: 24.8\%, \colorbox{farBehindColor}{\makebox[0.5em][c]{D}}: 66.2\% & 40 & 17.5\% \\
State Senate B & 16320 & 2 & \colorbox{winnerColor}{\makebox[0.5em][c]{A}}: 0.0\%, \colorbox{farBehindColor}{\makebox[0.5em][c]{B}}: 92.7\% & 40 & 92.7\% \\
State Senate D & 19362 & 4 & \colorbox{winnerColor}{\makebox[0.5em][c]{A}}: 0.0\%, \colorbox{contenderColor}{\makebox[0.5em][c]{B}}: 8.8\%, \colorbox{distantColor}{\makebox[0.5em][c]{C}}: 45.5\%, \colorbox{farBehindColor}{\makebox[0.5em][c]{D}}: 61.8\% & 40 & 8.8\% \\
State Senate F & 17066 & 4 & \colorbox{winnerColor}{\makebox[0.5em][c]{A}}: 0.0\%, \colorbox{contenderColor}{\makebox[0.5em][c]{B}}: 5.5\%, \colorbox{distantColor}{\makebox[0.5em][c]{C}}: 43.2\%, \colorbox{farBehindColor}{\makebox[0.5em][c]{D}}: 60.3\% & 40 & 5.5\% \\
State Senate H & 18002 & 3 & \colorbox{winnerColor}{\makebox[0.5em][c]{A}}: 0.0\%, \colorbox{contenderColor}{\makebox[0.5em][c]{B}}: 10.7\%, \colorbox{farBehindColor}{\makebox[0.5em][c]{C}}: 60.6\% & 40 & 10.7\% \\
State Senate J & 10498 & 3 & \colorbox{winnerColor}{\makebox[0.5em][c]{A}}: 0.0\%, \colorbox{distantColor}{\makebox[0.5em][c]{B}}: 40.4\%, \colorbox{farBehindColor}{\makebox[0.5em][c]{C}}: 73.7\% & 40 & 40.4\% \\
State Senate L & 19600 & 4 & \colorbox{winnerColor}{\makebox[0.5em][c]{A}}: 0.0\%, \colorbox{contenderColor}{\makebox[0.5em][c]{B}}: 9.9\%, \colorbox{distantColor}{\makebox[0.5em][c]{C}}: 41.0\%, \colorbox{farBehindColor}{\makebox[0.5em][c]{D}}: 61.7\% & 40 & 9.9\% \\
State Senate N & 15562 & 4 & \colorbox{winnerColor}{\makebox[0.5em][c]{A}}: 0.0\%, \colorbox{competitiveColor}{\makebox[0.5em][c]{B}}: 23.0\%, \colorbox{distantColor}{\makebox[0.5em][c]{C}}: 35.6\%, \colorbox{farBehindColor}{\makebox[0.5em][c]{D}}: 66.5\% & 40 & 23.0\% \\
State Senate P & 11532 & 3 & \colorbox{winnerColor}{\makebox[0.5em][c]{A}}: 0.0\%, \colorbox{nearWinColor}{\makebox[0.5em][c]{B}}: 3.0\%, \colorbox{farBehindColor}{\makebox[0.5em][c]{C}}: 57.8\% & 40 & 3.0\% \\
State Senate R & 18165 & 4 & \colorbox{winnerColor}{\makebox[0.5em][c]{A}}: 0.0\%, \colorbox{contenderColor}{\makebox[0.5em][c]{B}}: 9.8\%, \colorbox{distantColor}{\makebox[0.5em][c]{C}}: 39.8\%, \colorbox{farBehindColor}{\makebox[0.5em][c]{D}}: 62.1\% & 40 & 9.8\% \\
State Senate T & 6638 & 2 & \colorbox{winnerColor}{\makebox[0.5em][c]{A}}: 0.0\%, \colorbox{farBehindColor}{\makebox[0.5em][c]{B}}: 94.2\% & 40 & 94.2\% \\
U.S. House & 330298 & 5 & \colorbox{winnerColor}{\makebox[0.5em][c]{A}}: 0.0\%, \colorbox{nearWinColor}{\makebox[0.5em][c]{B}}: 2.4\%, \colorbox{distantColor}{\makebox[0.5em][c]{C}}: 42.5\%, \colorbox{distantColor}{\makebox[0.5em][c]{D}}: 45.3\%, \colorbox{farBehindColor}{\makebox[0.5em][c]{E}}: 57.4\% & 40 & 2.4\% \\
Presidential \textsuperscript{†} & 338650 & 8 & \colorbox{winnerColor}{\makebox[0.5em][c]{A}}: 0.0\%, \colorbox{contenderColor}{\makebox[0.5em][c]{B}}: 13.0\%, \colorbox{distantColor}{\makebox[0.5em][c]{E}}: 46.2\%, \colorbox{distantColor}{\makebox[0.5em][c]{C}}: 46.3\%, \colorbox{distantColor}{\makebox[0.5em][c]{D}}: 47.3\%, \colorbox{distantColor}{\makebox[0.5em][c]{G}}: 49.0\%, \colorbox{farBehindColor}{\makebox[0.5em][c]{H}}: 50.5\%, \colorbox{farBehindColor}{\makebox[0.5em][c]{F}}: 50.6\% & 40 & 13.0\% \\

\end{longtable}

\normalsize
\subsection{Supplementary empirical analyses}
\subsubsection{Summary of Election Attributes}\label{sec:summary_attributes}
\Cref{tab:side_by_side_rcv_summary} summarizes all 106 single winner elections in NYC'21 and Alaska'24.
\begin{table}[h]
\caption{\small RCV Election Attributes Summary: NYC 2021 and Alaska 2024. Here, $
\leq 20\%$ allowance denotes computations run with allowance = min(20\%, traceability threshold) per election.}
\label{tab:side_by_side_rcv_summary}
\centering
\footnotesize
\begin{subtable}[t]{0.47\textwidth}
\caption{\small 2021 NYC Democratic Primary}
\label{tab:nyc_rcv_extended}
\centering
\begin{tabular}{l r r}
\toprule
\textbf{Election Attributes} & \textbf{No.} & \textbf{\% of 54} \\
\midrule
\multicolumn{3}{l}{\textit{Competitiveness bands}}\\
\quad Near Winner (0–5\%)       & 7  & 13.0 \\
\quad Contender (5–20\%)        & 24 & 44.4 \\
\quad Competitive (20–30\%)     & 11 & 20.4 \\
\quad Distant (30–45\%)         & 8  & 14.8 \\
\quad Far Behind ($>$45\%)      & 4  &  7.4 \\
\textbf{\quad $\le$30\% total}  & 42 & 77.8 \\
\midrule
\multicolumn{3}{l}{\textit{Ballot exhaustion impact (strategy $<$ exhaust)}}\\
\quad High flip probability     & 2  &  \textbf{3.7} \\
\midrule
\multicolumn{3}{l}{\textit{Strategic complexity (non-selfish)}}\\
\quad Full allowance            & 11 & 20.4 \\
\quad $\le$20\% allowance       & 0  &  \textbf{0.0} \\
\midrule
\multicolumn{3}{l}{\textit{Preference-order alignment}}\\
\quad Full allowance            & 41 & 75.9 \\
\quad $\le$20\% allowance       & 52 & \textbf{96.3} \\
\bottomrule
\end{tabular}
\end{subtable}
\hfill
\begin{subtable}[t]{0.47\textwidth}
\caption{\small 2024 Alaska State-wide}
\label{tab:ak_rcv_extended}
\centering
\begin{tabular}{l r r}
\toprule
\textbf{Election Attributes} & \textbf{No.} & \textbf{\% of 52} \\
\midrule
\multicolumn{3}{l}{\textit{Competitiveness bands}}\\
\quad Near Winner (0–5\%)       & 11 & 21.2 \\
\quad Contender (5–20\%)        & 19 & 36.5 \\
\quad Competitive (20–30\%)     &  6 & 11.5 \\
\quad Distant (30–45\%)         &  3 &  5.8 \\
\quad Far Behind ($>$45\%)      & 13 & 25.0 \\
\textbf{\quad $\le$30\% total}  & 36 & 69.2 \\
\midrule
\multicolumn{3}{l}{\textit{Ballot exhaustion impact (strategy $<$ exhaust)}}\\
\quad High flip probability     & 1  &  \textbf{1.9} \\
\midrule
\multicolumn{3}{l}{\textit{Strategic complexity (non-selfish)}}\\
\quad Full allowance            & 0  &  0.0 \\
\quad $\le$20\% allowance       & 0  &  \textbf{0.0} \\
\midrule
\multicolumn{3}{l}{\textit{Preference-order alignment}}\\
\quad Full allowance            & 51 & 98.1 \\
\quad $\le$20\% allowance       & 52 & \textbf{100} \\
\bottomrule
\end{tabular}
\end{subtable}
\end{table}

\subsubsection{Detailed Non-Selfish Strategy Example}

Table~\ref{tab:district23_detail} illustrates the complete strategic landscape for candidate D in NYC Council District 23, showing how non-selfish strategy achieves victory with fewer ballot additions by engineering favorable elimination sequences.

\begin{table}[H]
\centering
\caption{Winning strategies within $40\%$ allowance for candidate D, NYC Council District 23}
\begin{tabular}{lcc}
\hline
Social Choice Order & Additions (\%) & Ballot Composition \\
\hline
DBACE & 20.3 (\textit{optimal}) & B: 2.2\%, D: 18.1\% \\
DABCE & 27.0 (\textit{selfish}) & D: 27.0\% \\
DAEBC & 33.4 & D: 17.3\%, E: 1.3\%, ED: 14.8\% \\
DBCAE & 36.0 & B: 0.5\%, C: 17.3\%, D: 18.1\% \\
DACBE & 36.2 & C: 7.5\%, D: 19.4\%, CD: 9.4\% \\
DCBAE & 38.0 & B: 0.5\%, C: 19.4\%, D: 18.1\% \\
\hline
\end{tabular}
\begin{tablenotes}
\small
\item The optimal strategy for candidate D requires 20.3\% total additions: 2.2\% ranking rival B first and 18.1\% ranking D first. The purely selfish strategy requires 27.0\% additions---demonstrating how supporting rival B to manipulate the elimination order reduces the required ballot additions by 6.7 percentage points.
\end{tablenotes}
\label{tab:district23_detail}
\end{table}

The original social choice order of District 23 is ABCDE. D's optimal strategy works by eliminating candidate A prior to B, so that D would get more significant transfer votes from A, and would require fewer additions against B in the final round.

\subsubsection{Impact of Stronger Removal Conditions}\label{app:strong_removal_impact}

We applied the strengthened candidate–removal test to 33 real-world elections with at least six candidates.  
Overall, 18 contests (55 \%) exhibited a measurable reduction in required search effort.  
Contests with \textit{11–15 candidates} improved in 4 of 7 cases (57 \%), averaging a 39.2 \% reduction, indicating that tighter traceability screening is especially valuable in crowded fields.  
In the \textit{6–10 candidate} bracket, 14 of 26 elections (54 \%) saw gains, with an average improvement of 28.8 \%.  
Five elections recorded improvements above 60 \%, highlighting the practical payoff of the refined rule in demanding instances.

\begin{table}[H]
\centering
\caption{Effect of stronger removal conditions, by candidate field size}
\label{tab:performance_by_size}
\begin{tabular}{lcccc}
\toprule
Category & Elections & Improved & Avg.\ Improvement & Max.\ Improvement  \\
\midrule
6--10 Candidates & 26 & 14 (54\,\%) & 28.8\,\% & 119.5\,\%  \\
11--15 Candidates & 7  & 4  (57\,\%) & 39.2\,\% & 76.0\,\%  \\
\midrule
Overall & 33 & 18 (55\,\%) & 31.1\,\% & 119.5\,\%  \\
\bottomrule
\end{tabular}
\end{table}

\subsubsection{Ballot Exhaustion Analysis} \label{app:ballot_exhaustion_single_winner}
We apply the probability models detailed in \Cref{app:ballot_exhaustion} to analyze the impact of ballot completion. The heatmaps in \Cref{fig:heatmap_prb} demonstrate that only highly competitive elections exhibit non-trivial probabilities of outcome changes. Three elections show near-certain flip probabilities (approaching 1.0) under similarity bootstrap models and high probabilities under similarity-based beta models, suggesting significant bias in ballot exhaustion patterns. While ballot exhaustion enables potential outcome changes, we find that election competitiveness and systematic bias in exhaustion patterns are the critical determinants of whether completing exhausted ballots would alter electoral outcomes.

\begin{figure}[h]
    \centering
    \includegraphics[width=1\linewidth]{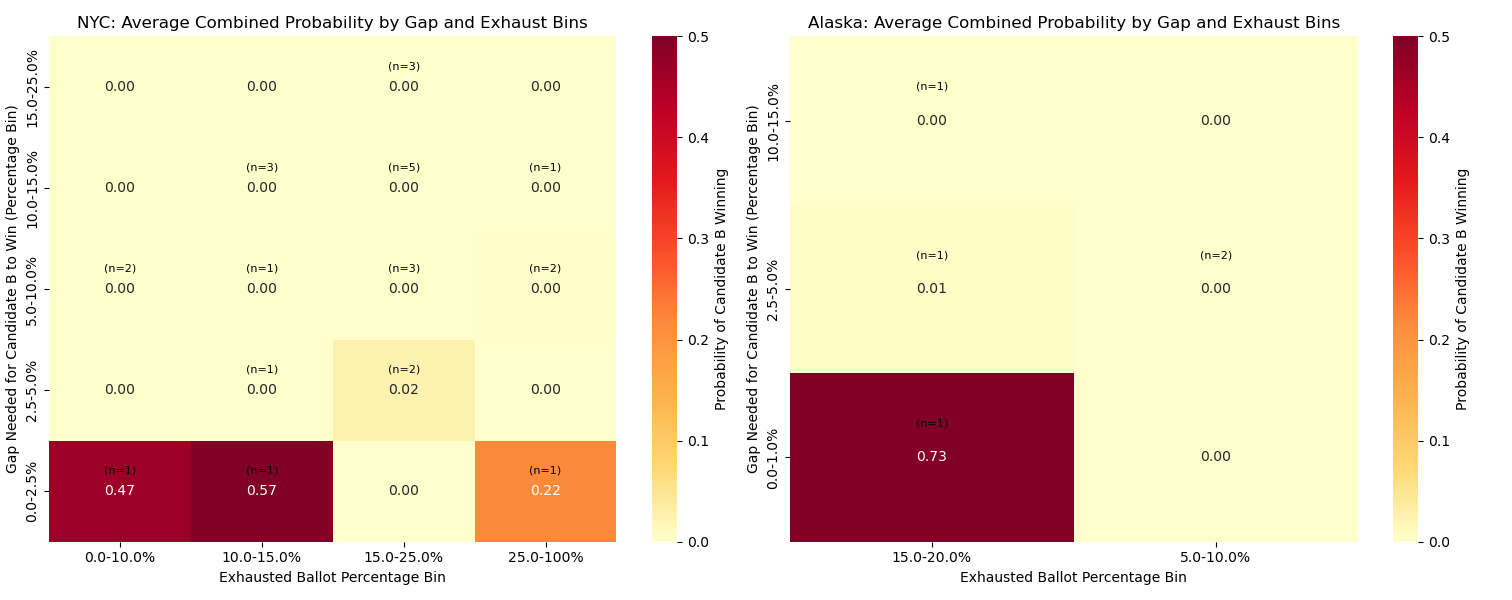}
    \caption{Ballot Exhaustion and Outcome Sensitivity: Heatmaps showing the average probability of alternate outcomes with respect to the percent of candidate's victory gap and exhausted ballots at the time of their elimination.}
    \label{fig:heatmap_prb}
\end{figure}
\section{Multi-Winner STV: Details} \label{app:portland}
We used the official cast vote record (CVR) released in December 2024 in this case study. The record included improperly marked ballots, which were processed in accordance with the official RCV ballot adjudication procedures \cite{MultnomahCounty2024Turnout}. Each election followed the single transferable vote (STV) rules, electing exactly three winners using the Droop quota of 25\% \citep{multnomah_rcv2025}.
Using ERSF, we examined strategic behavior across all four elections.  Below, we provide a brief analysis of each district's elections and the subsequent bootstrap analysis.

\subsection{District 1}\label{app:portland_dis1}

District 1 includes 16 candidates, with the first candidate winning in the 13th round following the elimination of 11 candidates. Given \(n = 16\) and up to \(k = 3\) winners, the total number of possible orders is \(16!\), and the number of possible sequences is \(\sum_{j=1}^3 \binom{16}{j} = 696\). 
For \(B = 4.7\%\), \Cref{algo:candidate-removal-updated} removes 8 candidates; among the remaining 8, only 7 are able to win by adding up to 4.7\% allowance, which is the algorithmic traceability threshold. Note that candidate H, the 8th candidate, is included in the set of relevant candidates but remains unable to secure a win through strategic additions. This occurs because H has the potential to attain higher positions in the ranking, thereby influencing the election dynamics, yet ultimately falls short of meeting the necessary threshold to win. The subsequent strategic analysis of this set using \Cref{algo: robust_allocation} takes approximately 814 minutes on a modern laptop. The resulting optimal strategies for 7 candidates are presented in \Cref{tab:portland_dis1_strats}. Reducing \(B\) to 4.17\% allows the removal of 9 candidates, cutting the corresponding analysis time to about 12 minutes.

For District 1 bootstrap analysis, we generated 1,000 samples using sampling with replacement. We then applied a slightly lower allowance of \(B = 4\%\), as 4.17\% is precise for the original dataset. Under these conditions, 9 candidates were removed in 807 samples and 8 in 190 samples, leaving 3 samples unsolved. From these, we randomly selected 100 for in-depth strategic analysis, summarized in \Cref{tab:summary_dis1}.

\subsection{Districts 2, 3 and 4}
District 2 had 22 candidates, with the first win occurring in the 20th round after the elimination of 18 candidates. The subsequent analysis, including both the election data and bootstrap procedures, follows the same methodology as in District 1 (\Cref{app:portland_dis1}). Given the lower competitiveness of the District 2 election, the algorithms achieved a higher threshold, reaching up to 6.5\% while removing 18 candidates.
We analyzed 100 bootstrap samples with 6\% allowance. The candidate-elimination algorithm was effective for all samples, reducing the number of relevant candidates to 4 in each case. Candidate D secured a place in the winning set in 85\% of the samples, requiring an average of 5.64\% additional votes.

In Districts 3 and 4, the first election winner emerged in rounds 20 and 7, following the elimination of 18 and 5 candidates, respectively. In both cases, the first win occurred while at least 11 candidates remained active in the election. 
Thus, for Districts 3 and 4, \Cref{thm: irrelevant_extension} was specifically useful in enabling candidate elimination even when the removal set includes an election winner. Using the corresponding algorithm, we determine the algorithmic traceability threshold that satisfies the elimination condition. The thresholds for Districts 3 and 4 are 12.36\% and 9.6\%, respectively.
For bootstrap samples, we used 11.5\% allowance for District 3, and 9\% for District 4 and analyzed 100 bootstrap samples. Within these, the algorithms achieve 100\% efficiency in eliminating irrelevant candidates and analyzing samples in both districts.

\subsection{Ballot Exhaustion Analysis}
We analyzed ballot exhaustion's strategic impact in this multi-winner STV setting. For candidates where ballot exhaustion exceeded their victory gap, we calculated the required preference percentage among exhausted ballots using: $\text{Required Net Advantage} = \frac{\text{Gap}}{\text{Exhaust}} \times 100$ and $\text{Required Preference Percent} = (1 + \frac{\text{Required Net Advantage}}{100})/2 \times 100$. This formula determines what percentage of exhausted voters must prefer the trailing candidate to generate sufficient net votes for victory. The multi-winner context required identifying active competitors at each candidate's specific elimination round, as strategic coordination effectiveness depends on the number and identity of remaining candidates rather than the full initial field. 

While adding ballots to multi-winner STV increases the Droop quota, we computationally verified that completing exhausted ballots (or adding equivalent strategic ballots) in the Portland elections affect only which candidates win or are eliminated within rounds, without altering the sequence in the structure. Specifically, `win' rounds remain win rounds (though potentially for different candidates), and elimination rounds remain elimination rounds. As specified in \Cref{prop: ballot_exhaust}, this preserves the operational equivalence between ballot completion and strategic ballot addition used in our probability models.

Results showed that required preference percentages ranged from 53.98\% to 83.41\% across six qualifying candidates. In multi-winner elections, candidates face elimination at different stages with varying numbers of active competitors: District 4's most promising case involved 2 final candidates requiring 53.98\% preference among 14.08\% exhausted ballots, while District 1 cases involved 4-8 active candidates requiring 66.62\%-83.41\% preference rates. The empirical analysis revealed systematic preference deficits: observed complete ballot patterns consistently showed trailing candidates preferred by 33.41\%-40.13\% of voters---substantially below required thresholds of 53.98\%-83.41\%. Thus, unlike the single-winner analyses, the natural voter preferences consistently fell short of coordination requirements.
\end{document}